\documentclass[reprint,showpacs,twocolumn,superscriptaddress,aps]{revtex4-2}
\usepackage{graphicx}
\usepackage{amsthm,amsfonts,amstext,amssymb,amsmath}
\usepackage{latexsym}
\usepackage{braket}
\usepackage{bbold}
\usepackage{bm}
\usepackage{placeins}
\usepackage{enumerate}
\usepackage{standalone}
\usepackage{float}
\usepackage{mathtools}
\usepackage{xcolor}
\usepackage[colorlinks=true,
  linkcolor=black,
  citecolor=blue,
  urlcolor =blue]{hyperref}

\newtheorem{Observation}{Observation}
\newcommand{\tr}{\mathrm{tr}}

\begin{document}
\title{Entanglement phase diagrams from partial transpose moments}
\author{Jose Carrasco}
\thanks{These authors contributed equally.}
\affiliation{Institute for Theoretical Physics, University of Innsbruck, Innsbruck A-6020, Austria}
\affiliation{Dahlem Center for Complex Quantum Systems, Freie Universität Berlin, 14195 Berlin, Germany}

\author{Matteo Votto}
\thanks{These authors contributed equally.}
\affiliation{Univ. Grenoble Alpes, CNRS, LPMMC, 38000 Grenoble, France}

\author{Vittorio Vitale}
\affiliation{SISSA, via Bonomea 265, 34136 Trieste, Italy}
\affiliation{International Centre for Theoretical Physics (ICTP), Strada Costiera 11, 34151 Trieste, Italy}

\author{Christian Kokail}
\affiliation{Institute for Theoretical Physics, University of Innsbruck, Innsbruck A-6020, Austria}	
\affiliation{Institute for Quantum Optics and Quantum Information of the Austrian Academy of Sciences,  Innsbruck A-6020, Austria}

\author{Antoine Neven}
\affiliation{Institute for Theoretical Physics, University of Innsbruck, Innsbruck A-6020, Austria}	

\author{Peter Zoller}
\affiliation{Institute for Theoretical Physics, University of Innsbruck, Innsbruck A-6020, Austria}	
\affiliation{Institute for Quantum Optics and Quantum Information of the Austrian Academy of Sciences,  Innsbruck A-6020, Austria}

\author{Beno\^it Vermersch}
\affiliation{Institute for Theoretical Physics, University of Innsbruck, Innsbruck A-6020, Austria}	
\affiliation{Institute for Quantum Optics and Quantum Information of the Austrian Academy of Sciences,  Innsbruck A-6020, Austria}
\affiliation{Univ.  Grenoble Alpes, CNRS, LPMMC, 38000 Grenoble, France}

\author{Barbara Kraus}
\affiliation{Institute for Theoretical Physics, University of Innsbruck, Innsbruck A-6020, Austria}

\begin{abstract}
     We present experimentally and numerically accessible quantities that can be used to differentiate among various families of random entangled states. To this end, we analyze the entanglement properties of bipartite reduced states of a tripartite pure state. 
     We introduce a ratio of  simple polynomials of low-order moments of the partially transposed reduced density matrix and show that this ratio takes well-defined values in the thermodynamic limit for various families of entangled states. This allows to sharply distinguish entanglement  phases, in a way that can be understood from a quantum information perspective based on the spectrum of the partially transposed density matrix.
     We analyze in particular the entanglement phase diagram of Haar random states, states resulting form the evolution of chaotic Hamiltonians, stabilizer states, which are outputs of Clifford circuits,  Matrix Product States, and fermionic Gaussian states. We show that for Haar random states the resulting phase diagram resembles the one obtained via the negativity and that for all the cases mentioned above a very distinctive behaviour is observed. 
      Our results can be used to experimentally test necessary conditions for different types of mixed-state randomness, in quantum states formed in quantum computers and programmable quantum simulators.
\end{abstract}

\maketitle

\section{Introduction} 

Many-body quantum states and quantum phases, as prepared today in equilibrium or non-equilibrium dynamics on experimental quantum devices \cite{Altman2021qsim}, can be characterized and classified according to their entanglement properties. Recent examples of interest include the study of `entanglement phases' appearing in ensembles of Haar-random induced mixed states \cite{shapourian2021entanglement, Au11, Au12, Au13}, and the measurement-driven `entanglement transition' in hybrid quantum circuits \cite{fisher2022random, Potter_2022}, where a volume to area-law `entanglement transition' is observed as a function of the measurement rate. In a broader context, this leads to the challenge of identifying observables allowing to distinguish entanglement phases, playing essentially the role of `entanglement order parameters' in entanglement phase diagrams, and the development of experimentally accessible protocols to measure these quantities on present quantum devices. 

In the present work we address this problem in context of random many-body quantum states, with focus on entanglement properties of bipartite reduced states of a tripartite pure state. These quantum states include Haar-random states resulting from evolution of chaotic Hamiltonians, stabilizer states as outputs of Clifford circuits, Matrix Product States, and fermionic Gaussian states. As observables, which distinguish sharply between different entanglement phases, we introduce the ratio of  simple polynomials of low-order moments of the partially transposed reduced density matrix \cite{Pe96, Horodecki2003, calabresenegativity2012, elben2020mixed, neven2021symmetry}, and we show that this ratio takes on well-defined values in the thermodynamic limit for various families of entangled states. Besides providing a convenient tool in numerical studies \cite{wybo2020mbl, wu2020montecarlofiniteT, feldman2022entanglement}, such observables are experimentally accessible, in particular within the randomized measurement toolbox \cite{elben2020mixed,zhou2020single,elben2022therandomized}, paving the way to an experimental exploration of entanglement phases and phase diagrams. 

The outline of the remainder of the paper is the following. In Sec.~\ref{sec:summary} we present a summary of our findings. In particular, we introduce the quantity $r_2$ and its generalizations which are central in our studies of entanglement phases. Despite the simplicity of these quantities, which are ratios of polynomials of moments of the partially transposed bipartite state, we show that they capture the entanglement structure of Haar-random induced mixed states in Sec.~\ref{sec:haar} and extend the analysis to  pseudo Haar-random induced mixed states in Sec.~\ref{sec:noise}. In Sec.~\ref{sec:stabilizers}, Sec.~\ref{sec:rydberg} and  Sec.~\ref{sec:MPS-FF} we show that $r_2$ displays a very different behaviour for random, but not Haar-random states, that it allows to observe the transition to Haar random states, and that it is capable of identifying other quantum resources. We comment on experimental possibilities to access $r_2$ in  Sec.~\ref{sec:measurements}. Finally, in Sec.~\ref{sec:detection}, we analyze the power of $r_2$ for detecting entanglement. The manuscripts ends with a short summary in Sec.~\ref{sec:conclusion}.

\begin{figure*}[t]
    \includegraphics[width=0.99\textwidth]{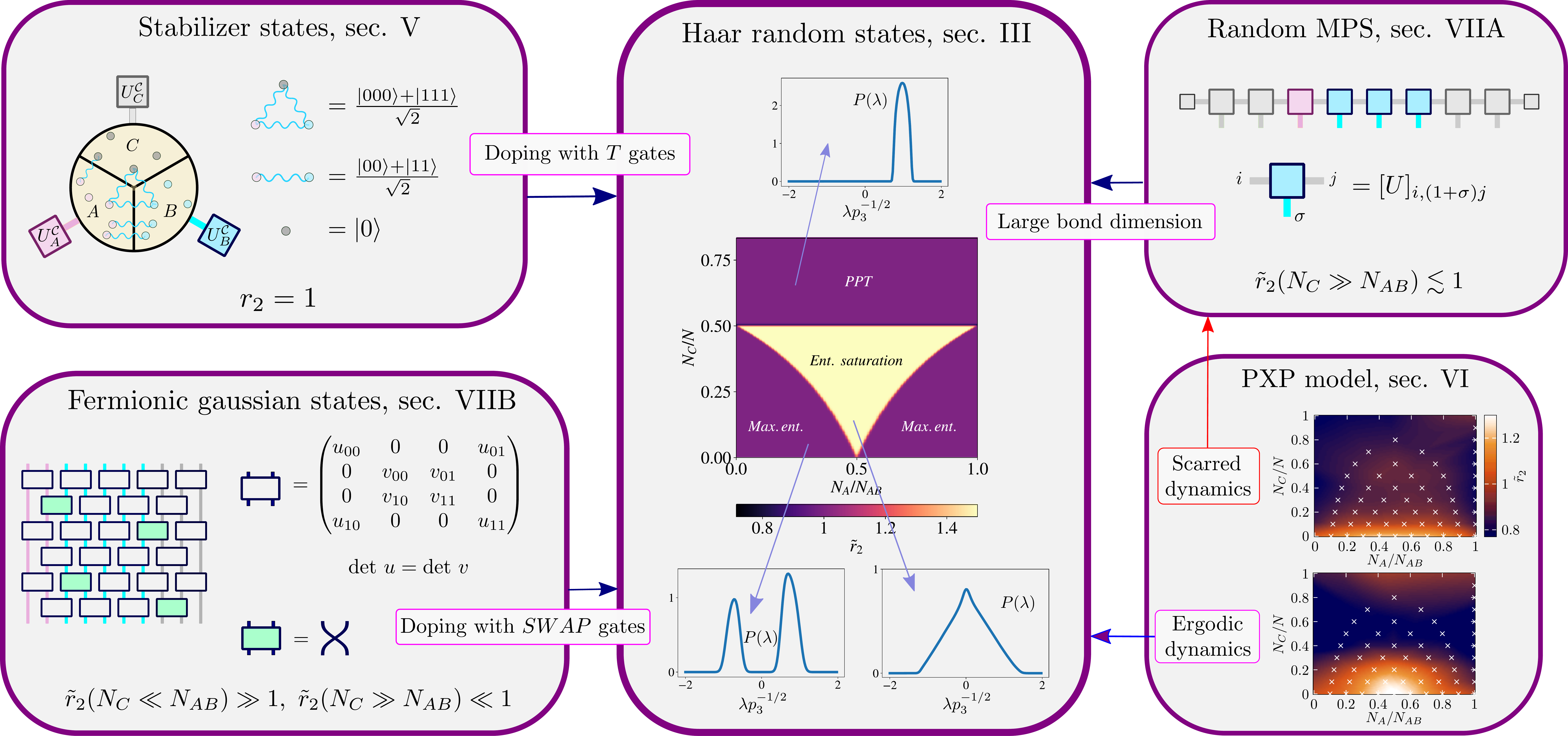}
    \caption{
    Summary of our results and structure of our paper.
    Middle Panel: In Sec.~\ref{sec:haar}, we show that $\tilde r_2$ takes well-defined values in the three different entangled phases of Haar random states ~\cite{shapourian2021entanglement}. 
   As shown in Ref.~\cite{shapourian2021entanglement}, in each of the entanglement phases, the probability distribution of the spectrum $P(\lambda)$ of the partial transpose operator $\rho^{\Gamma}$ has a characteristic shape, as depicted in the figure. 
   As shown in Sec.~\ref{sec:spectrum_subsec}, 
   the behavior of $r_2$ can be related to the shape of  $P(\lambda)$. In particular, when $r_2=1$, we can understand the positions of the peaks.
    Side Panels (other classes of states): For random, but not Haar random states a very different behaviour of $r_2$ is observed. For stabilizers states [Sec.~\ref{sec:stabilizers}], using the decomposition from Eq.~\eqref{eq:GHZ} (shown here as cartoon, see also Ref.~\cite{bravyi2006ghz}), we obtain $r_2 =1$. For fermionic Gaussian states [Sec.~\ref{sec:MPS-FF}], we obtain exponentially large (or small) values of $\tilde{r}_2$.
    Whereas the evolution under Clifford gates (Matchgates) can be simulated classically efficiently, the addition
of resourceful $T$-gates (SWAPs) makes the computation universal, respectively.
We show in Sec.~\ref{sec:MPS-FF} how $\tilde {r}_2$ changes from the typical values for fermionic Gaussian states towards the value obtained for  Haar random states, as the number of SWAP gates is increased.
    For states generated via dynamics from the Rydberg PXP model [Sec.~\ref{sec:rydberg}], the values of $\tilde{r}_2$ resemble the ones obtained for random Matrix-Product-States (MPS) if the dynamics is `scarred', and the ones of Haar random states if the dynamics is ergodic.
    For random MPS, we observe a region with $\tilde{r}_2 <1$.}
    \label{fig:summaryB}
\end{figure*}

\section{Preliminaries and summary of results}
 \label{sec:summary}

 We first introduce our notation and recall previous results regarding the bipartite entanglement content of random states. Then, we summarize our main findings, which are also illustrated in Figure~\ref{fig:summaryB}.

\subsection{Notation}
Throughout this work we consider a tripartite system in a pure state
$\ket\psi\in\mathcal H_{A}\otimes\mathcal H_{B} \otimes\mathcal H_C$ .
One can think of $A$, $B$, $C$ consisting of $N_A$, $N_B$, $N_C$ qubits, respectively. The associated Hilbert spaces are of dimension $L_X=2^{N_X}$, with $X=A,B,C$, and the dimension of the total Hilbert space is $L=L_{A}L_{B}L_C=2^N$. We will also use the notation  $N_{AB}=N_{A}+N_{B}$ and  $L_{AB}=L_{A}L_{B}=2^{N_{AB}}$. We analyze the bipartite mixed state entanglement properties of reduced states $\rho=\tr_C\ket\psi\bra\psi$. Haar-random {\em induced} mixed states are states where the tripartite pure state is Haar-random. Their entanglement properties have been studied in~\cite{Au11,Au12,Au13}.
The partial transpose (PT) of a density operator $\rho$, \begin{equation}
  \rho^\Gamma=(1_A\otimes T_{B})\rho,
\end{equation}
where $T$ denotes the transposition, plays a central role in  bipartite entanglement theory. Separable (non--entangled) state can be written as a convex combination of local density operators, which implies that their PT is always positive semidefinite~\cite{Pe96,HHH96}. That is a state with a non--positive PT (NPT state) is entangled. The converse is not true as there exists PPT entangled states, i.e. entangled states which have a positive semidefinite PT.
The entanglement monotone related to the PPT-condition is the logarithmic negativity~\cite{VW02,Pl05} ${\mathcal E}(\rho)$, given by
\begin{equation}\label{eq:LN}
{\mathcal E}(\rho)=\log\|\rho^\Gamma\|_1=\log\sum_i|\lambda_i|\,,
\end{equation}
where the sum is over eigenvalues $\{\lambda_i\}$ of the partial transpose (PT) operator. Our study is based on partial transpose moments

\begin{equation}
  p_n=\tr(\rho^\Gamma)^n
\end{equation}
for $n=1,2,\ldots,\dim({\mathcal H}_{AB})$. As shown in  Refs.~\cite{elben2020mixed,neven2021symmetry,yu2021optimal}, these moments can be utilized to derive necessary and sufficient conditions (in terms of inequalities) of NPT-entanglement.

\subsection{Preliminaries}

Before summarizing our results we review here some previous findings which are relevant for our work. As mentioned before, the entanglement content of random quantum states is central in a variety of scenarios including the  study of quantum many-body chaotic systems~\cite{dalessio2016ETH,nahum2017quantum}, the certification of quantum
computers~\cite{arute2019quantum}, properties of random quantum circuits~\cite{lashkari2013scrambling, hosur2016chaoschannels, nahum2018operator}, and the description of black-holes~\cite{hayden2007BH, penington2020replica, piroli2020BH}.
Generic properties of quantum chaotic systems can be derived using random matrix theory~\cite{dalessio2016ETH, bohigas1984chaosuniversality, GUHR1998rmtchaos, kos2018RMTmbchaos, chen2018chaosuniversality}. A seminal result in this context is due to Page~\cite{Page1993}, who showed that the averaged bipartite entanglement entropy of Haar-random pure states obeys a volume-law. An extension of this result to Haar-random induced mixed states has been achieved in ~\cite{shapourian2021entanglement, Au11,Au12,Au13}. In particular the bipartite entanglement properties of Haar random induced states have been analyzed with the logarithmic negativity,  ${\mathcal E}(\rho)$,  for different partitions sizes $(N_A,N_B,N_C)$~\cite{shapourian2021entanglement} . The scaling behavior of the expected value of ${\mathcal E}(\rho)$, ${\mathbb E}[{\mathcal E}(\rho)]$, determines a characteristic phase diagram for random states (see Fig. 2 of Ref.~\cite{shapourian2021entanglement}, which shows a similar phase diagram as the one presented in the center of Fig.~\ref{fig:summaryB}). Depending on the partition sizes, the system can be in three different ``entanglement phases''.
Roughly speaking, the phase diagram presented in Ref.~\cite{shapourian2021entanglement} shows the following three different phases: 
(Phase~I)~For $N_C$ larger than $N_{AB}$, ${\mathbb E}[{\mathcal E}(\rho)]$ vanishes and thus, on average, $\rho$ is PPT. For obvious reasons, this phase is called the PPT phase. (Phase II)~For $N_C$ smaller than $N_{AB}$ and $N_{A}<N_{B}$ (with $N_{A}\not\simeq N_{B}$), the subsystem $A$ is not entangled with the subsystem $C$ but is maximally entangled with the subsystem $B$ and ${\mathbb E}[{\mathcal E}(\rho)]\sim N_{A}$. Obviously, similar results hold for $N_{B}<N_{A}$. This phase is called the maximally entangled (ME) phase. (Phase III) For $N_C$ smaller than $N_{AB}$ and $N_{A}\simeq N_{B}$, subsystems $A$ and $B$ are not maximally entangled and ${\mathbb E}[{\mathcal E}(\rho)]\sim (N_{AB}-N_C)/2$. This phase is called the Entanglement Saturation (ES) phase. Whereas the PPT and ME phases are expected (also due to the results on random pure states~\cite{Page1993}), Ref.~\cite{shapourian2021entanglement} showed the existence of the ES phase for mixed bipartite states. As we recall below, those results can be obtained from random matrix theory in the limit of high--dimensional Hilbert spaces. In Ref.~\cite{shapourian2021entanglement} it has also been shown that the probability distribution $P(\lambda)$ of the spectrum of $\rho^\Gamma$ shows a distinctive behavior in all three phases.
In the PPT and ES phases, it follows a semicircle distribution (with support only on the positive domain in the PPT phase).
In contrast to that, in the ME phase, the spectrum is bimodal, following two separate Marčenko-Pastur laws in the positive and negative domain (see also middle panel of Fig. 1).

\subsection{Summary of results}

We identify the following ratios as central quantities in the study of entanglement of random states: 
\begin{equation}\label{eq:r2}
r_2=\frac{p_2p_3}{p_4}\,,\quad\tilde r_2=\frac{{\mathbb E}[p_2]{\mathbb E}[p_3]}{{\mathbb E}[p_4]}
\end{equation}
and higher order generalizations of the form $r_n=p_np_{n+1}/(p_{n+2}p_{n-1})$ and ${\mathbb E}[p_n]{\mathbb E}[p_{n+1}]/({\mathbb E}[p_{n+2}]{\mathbb E}[p_{n-1}])$, respectively. Here, ${\mathbb E}$ denotes the ensemble average. We show that the quantity ${\mathbb E}[r_2]$ can be approximated by $\tilde r_2$ for Haar random states and used to detect and classify various types of entanglement phases. 

These definitions are inspired by the study of the entanglement structure of Haar-random induced mixed states presented in Ref.~\cite{shapourian2021entanglement} (see below). However, 
in contrast to the negativity, these  quantities only involve few moments of the PT, which makes them experimentally and also numerically more accessible than the negativity. $r_2$ and its generalizations do not only allow us to reproduce the phase diagram of Haar random induced states, but are capable of identifying various entanglement phases of different kind of random states.We show that they are capable of differentiating between Haar-random states and other sets of states. This is highly relevant within quantum computation and beyond as they can be used to confirm the behavior of random states or to show that the system of interest does not generate (enough) randomness. Moreover, other quantum resources, such as ``non-stabilizerness''\cite{PiroliNonstabilizerness} can be detected with these quantities.
Our main findings are summarized in Figure 1, which we explain now in more detail. 

\subsubsection{Haar random induced states (middle panel of Fig.~\ref{fig:summaryB})} As our first main result we show that $\tilde r_2$, which depends only on up to the fourth PT moment captures the entanglement structure of Haar-random induced mixed states.
In particular, $\tilde r_2$ takes quantized values $1,3/2,1$ for the different entanglement phases of Haar-random induced mixed states, identifying sharply the phase diagram shown in Fig. 2 of \cite{shapourian2021entanglement} (see middle panel of Figs.~\ref{fig:summaryB} and~\ref{fig:phase_diagram_r_2}). Moreover, we can understand these properties based on the universal properties of the negativity spectrum $\{\lambda_i\}$, the eigenvalues of $\rho^\Gamma$,  which are reflected in the value of $r_2$.
In particular, we show numerically that the typical spectrum in the ME and PPT phases displays one or two peaks around $\pm \sqrt{p_3}$.
As we will show below, having such a spectrum necessarily implies the property $r_2 =1$. Let us mention here that the two phases, for which $r_2$ is 1 can be easily differentiated using another quantity, which involves only the first two (non--trivial) moments (See Sec.~\ref{sec:p3neg} and App.~\ref{app:p3neg}). 

Despite the fact, that the PT moments are strongly related to the entanglement properties of a mixed state, it is rather surprising that the behaviour of ratios of these quantities show such a strong agreement with the one of the much more involved negativity, which is an entanglement monotone. However, as will become clear below, the analytic expressions derived in Ref.~\cite{shapourian2021entanglement} for the negativity, which involves all PT moments, in the thermodynamic limit motivated us to introduce and study the quantities presented in Eq. (\ref{eq:r2}), which are functions of only a few PT moments. 

\subsubsection{Random, but not Haar-random states
(side panels of Fig.~\ref{fig:summaryB})}

As a second main result we show that $r_2$ displays a very different behaviour for random, but not Haar-random states and that it reveals also other resources in quantum information. Furthermore, we show that the transition from randomly chosen states from a set of classically simulable 
states to Haar random states can be observed with the help of $r_2$. 
In order to demonstrate that, we consider various sets of physically relevant states, as illustrated in Fig.~\ref{fig:summaryB}.

{\it From classically simulable states to random states}:

Despite the simplicity of $r_2$, it exhibits a completely different behaviour for families of states, which, if viewed as output states of a quantum computation or simulation, can be simulated classically efficiently. Instances of such sets of states, which we investigate here, are (i) stabilizer states, which are generated by Clifford circuits acting on computational basis states, (ii) random fermionic Gaussian states, which result form random Matchgate circuits, and (iii) random matrix-product states (rMPS). The behaviour of $r_2$ as a function of the system size is very different for these sets of states compared to Haar-random induced states. To give an example,
$r_2$ takes a fixed value $1$ for any stabilizer state. Hence, any different value shows that the state cannot be generated by a Clifford circuit (acting on a computational basis state), which is classically efficiently simulable. In this sense, $r_2$ can be viewed as an indicator of ``magic''~\cite{haferkamp2020quantum}. It is well known that the inclusion of additional resources, such as the T-gate for Clifford circuits or the SWAP-gate for Matchgate computation (fermionic Gaussian states), elevates the computational power of a computation from classically efficiently simulable to a universal quantum computation. Interestingly, this transition can be made apparent by studying $r_2$. We show this explicitly for Fermionic Gaussian states in Sec.~\ref{sec:fermionic}. Similarly we show that random MPS states with low bond dimenions show a distinctive behavior of $r_2$ compared to random states. However, the ``phase diagram" resembles the one of Haar random states if the bond dimension increases. 

{\it Chaotic and non-chaotic evolutions}:
We show, with the help of an example, that the behavior of $r_2$ is different for  states generated by chaotic or non-chaotic Hamiltonian evolutions. To illustrate this, we discuss below an experimentally relevant situation based on the spin constrained dynamics of Rydberg atoms. Therefore, $r_2$ can serve as an indication of entanglement and/or Haar-``randomness'' that can be useful both numerically and experimentally.

\subsubsection{Measuring $r_2$ and Entanglement detection via $r_2$ and }

In contrast to the negativity, PT moments can be measured in an experiment using either randomized measurements, or quantum circuits using physical copies.
The PT moments of Haar random states being exponentially small in system size, and the quantity $r_2$ being a ratio of such small numbers, one needs however a high accuracy in estimating each PT moment.
We discuss these requirements in terms of number of measuremnts to overcome statistical errors in Sec.~\ref{sec:measurements} and App.~\ref{app:measuring}.

As a final result we study the capability of $r_2$ in detecting entanglement contained in single states. That is, we show that $r_2$, if evaluated on a single state can be used to detect entanglement and analyze its relation to other means of entanglement detection.

\section{$r_2$ reveals the phase diagram of Haar random states}\label{sec:haar}
In this section we focus on Haar-random induced mixed states. As mentioned before, 
the phase diagram of Ref.~\cite{shapourian2021entanglement} (which is similar to the one presented in Fig.~\ref{fig:summaryB}) is obtained by considering the logarithmic negativity $\mathcal{E}(\rho)$
as a function of the subsystem sizes. We first recall some details of the results obtained in 
Ref.~\cite{shapourian2021entanglement}, which also will  motivate the introduction of the quantities $r_2$ and $\tilde r_2$, and more generally $r_n$ and $\tilde r_n$. We show that $\tilde r_2$ takes well-defined quantized values in each of the phases identified originally by the behaviour of the negativity in Ref.~\cite{shapourian2021entanglement}. To relate the quantities introduced here, we finally provide numerical evidence (see Fig.~\ref{fig:phase_diagram_r_2}) that, for the ensemble of Haar-random induced mixed states, the average ${\mathbb E}[r_2]$ can be well approximated by $\tilde r_2$. 

\subsection{The ratios $\tilde r_n$ for Haar random states}

\begin{figure}[ht]
    \centering
    
    \includegraphics[width = 0.99\linewidth]{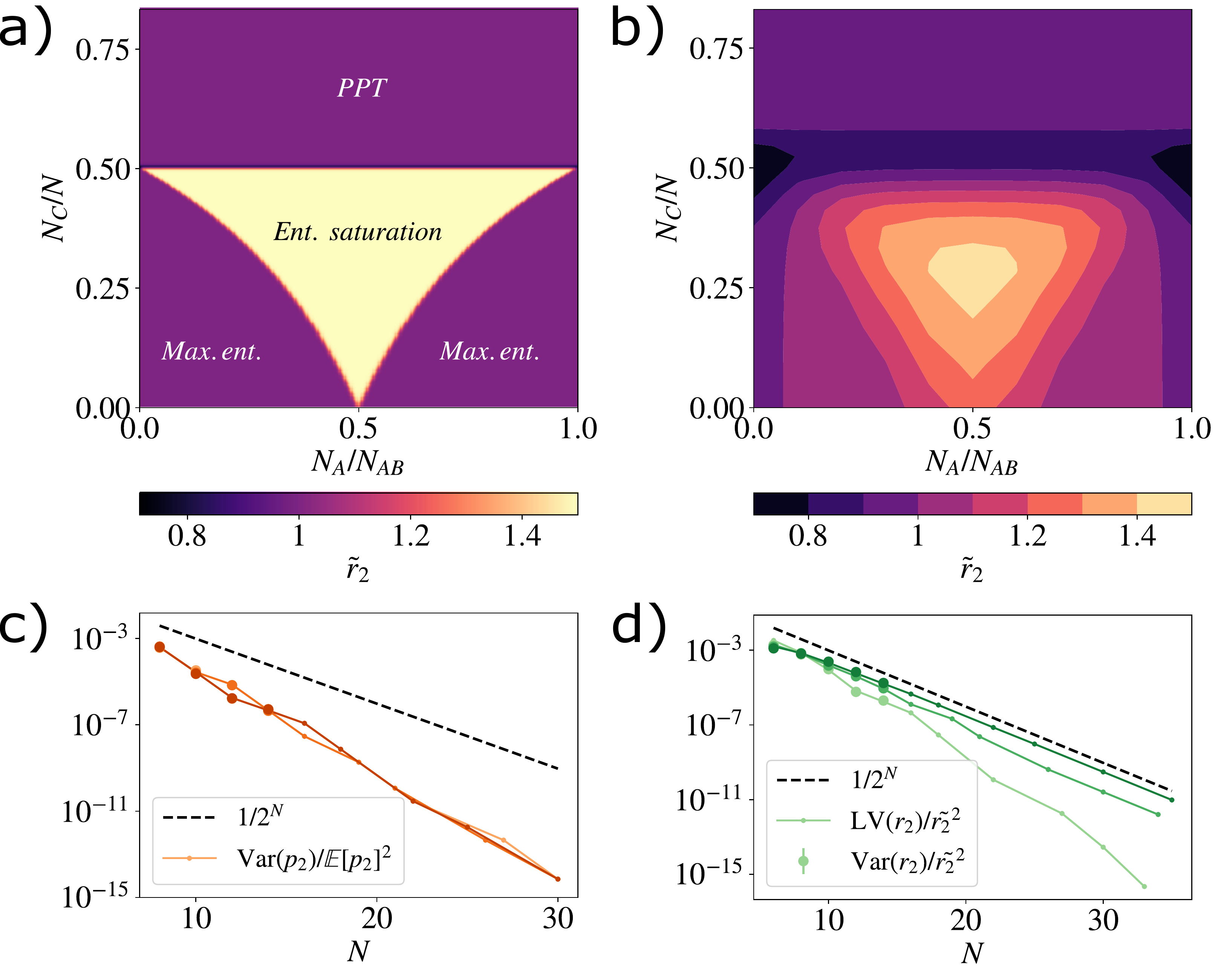} 
    
    \caption{a), b) Phase diagram for the ratio $\tilde r_2$ of Haar random states with fixed $N_{AB} = 256$ a) and $N_{AB} = 10$ b). The PT moments are computed using Eq.~\eqref{eq:ExactExpression}. It is possible to see that already at experimentally relevant sizes we recover features of the phase diagram in~\cite{shapourian2021entanglement}. Approaching the thermodynamic limit, we observe the emergence of the asymptotic values of $\tilde r_2 = 1$ for the ME and PPT phases, and of $\tilde r_2 = 3/2$ for the ES. c) Variance of $p_2$ as a function of $N$.
    The circles represent the variance for an ensemble of $100$ Haar random states, computed numerically. The dotted lines correspond to the analytical expressions obtained from random matrix theory. We choose $N_A=N_B$, with $N_{AB}\approx 2N/5$ (orange), and $N_{AB}\approx 4N/5$ (red) [We choose the closest integer for each $N$].
    d) Variance (circles) and linearized variance (dotted line) of $r_2$, showing an exponential reduction of statistical fluctuations with increasing $N$.
    Light green corresponds to the PPT phase, where for each $N$, we choose $N_A=N_B\approx N/5$.
    Green: we choose $N_A\approx N/5$, and $N_B\approx 3N/5$ (ME phase). 
    Dark green: We set $N_A=N_B\approx 2N/5$ (ES phase)
    }
    \label{fig:phase_diagram_r_2}
\end{figure}

Given the clear distinction between the three phases in terms of the logarithmic negativity (which contains information about PT moments of any possible order), it is interesting, from a more practical point of view, to see if the phases can also be resolved using only a few low-order PT moments. Low-order PT moments have the advantage (compared to the negativity), that they can be easily determined numerically and can be estimated experimentally, using randomized measurements~\cite{zhou2020single,elben2020mixed} or interferometric ``swap-tests''~\cite{Horodecki2003,carteret2005noiseless,graymachine2018}. Here we show that the phase diagram can be observed utilizing only low order PT moments.
More precisely, the expectation value of Haar-random induced mixed states ${\mathbb E}[r_2]$ is well approximated by $\tilde r_2$ (see Eq.~\eqref{eq:tilde r_n} below) and takes well-defined quantized values in each of these phases (see the middle panel of Fig.~\ref{fig:summaryB} and Fig.~\ref{fig:phase_diagram_r_2}). In particular, $\tilde r_2$ can be interpreted as an ``order parameter'' for the entanglement saturation phase, as it takes a fixed value $3/2$ only in this phase. 

Let us first recall a general method to compute PT moments of Haar-random induced mixed states. In general, the PT moments $p_n=\tr[(\rho^\Gamma)^n]$ of any bipartite state $\rho$ can be expressed as the following expectation value: 
\begin{equation}
    p_n ={\rm tr}\left(\Pi_A(\sigma_+)\otimes\Pi_B(\sigma_-)\rho_{AB}^{\otimes n}\right)\,.
\label{eq:SWAP}
\end{equation}
In the previous equation, we introduced the permutation operations $\Pi_A(\sigma_+)$, $\Pi_B(\sigma_-)$.
Let $S_n$ be the symmetric group over $n$ elements. For any permutation $\tau\in S_n$ and any subsystem $X=A,B,AB,ABC,\ldots$, we write $\Pi_X(\tau)$ for the following operator acting on $n$ copies of subsystem $X$,
\[
\Pi_X(\tau)\ket{\phi_1}_X\otimes\cdots\otimes\ket{\phi_n}_X=\ket{\phi_{\tau(1)}}_X\otimes\cdots\otimes\ket{\phi_{\tau(n)}}_X\,.
\]
Finally, $\sigma_{\pm}$ are two special permutations defined as $\sigma_\pm(k)=(k\pm 1) \mod n$, i.e.  cyclic (and anti-cyclic) permutations.

We are interested in the expectation value ${\mathbb E}[p_n]$ over Haar-random states $\ket{\psi}$ with $\rho_{AB}={\rm tr}_C\ket{\psi}\bra{\psi}$. In other words,
\begin{multline}\label{eq:SWAPbis}
{\mathbb E}[p_n]={\mathbb E}\Big[{\rm tr}\Big([\Pi_A(\sigma_+)\otimes\Pi_B(\sigma_-)\otimes1_C](\ket{\psi}\bra{\psi})^{\otimes n}\Big)\Big]\\
={\rm tr}\Big([\Pi_A(\sigma_+)\otimes\Pi_B(\sigma_-)\otimes1_C]{\mathbb E}\big[(\ket{\psi}\bra{\psi})^{\otimes n}\big]\Big)\,.
\end{multline}
Hence the problem can be reduced to compute the average of the linear operator $(\ket{\psi}\bra{\psi})^{\otimes n}$. It is well--known that for Haar-random states $\ket\psi$, the previous average equals the trace-one projector onto the symmetric subspace. This average can be obtained (see~\cite{Collins2006RMT,Elben2019correlations} and also Appendix~\ref{app:weingarten}) via the twirling formula~\eqref{eq:app-twirl}. Inserting this expression  in Eq.~\eqref{eq:SWAPbis} leads to 

\begin{equation}
    \mathbb{E}[p_n] = \frac{\sum_{\tau\in S_n} \text{tr}\Big[(\Pi_A(\sigma_+)\otimes\Pi_B(\sigma_-)\otimes 1_C)\Pi_{ABC}(\tau)\Big]}{\sum_{\tau\in S_n} \text{tr}\Big[\Pi_{ABC}(\tau)\Big]}\,.
    \label{eq:RMTexpression}
\end{equation}

As noted in Ref.~\cite{shapourian2021entanglement}, when working in the thermodynamic limit it is also possible to develop a diagrammatic approach to obtain a leading order expression for the average of PT moments of Haar random states:
\begin{equation}
    \mathbb{E}[p_n]\simeq\frac{1}{(L_A L_BL_C)^n} \sum_{\tau\in S_n}L_C^{c(\tau)} L_{A}^{c(\sigma_+\circ \tau)} L_{B}^{c(\sigma_-\circ \tau)},
    \label{eq:ExactExpression}
\end{equation}
where for any permutation $\tau\in S_n$, $c(\tau)$ is the number of cycles in $\tau$, counting also the cycles of length one. Using diagrammatic rules~\cite{shapourian2021entanglement}, one can obtain the thermodynamic limit of the expected values of PT moments by computing the leading contribution in $L$ of the previous expression.

One can show that in case $N_C>N_{AB}$ one obtains ${\mathbb E}[p_n]\simeq L_{AB}^{1-n}$ in the thermodynamic limit~\cite{shapourian2021entanglement}. For $N_C<N_{AB}$ and both $N_{A}<N/2$ and $N_{B}<N/2$, one gets in the thermodynamic limit
\[
{\mathbb E}[p_n]\simeq\begin{cases}
\displaystyle{\frac{C_k L_{AB}}{(L_{AB} L_C)^k}},&n=2k\\[4mm]
\displaystyle{\frac{(2k+1)C_k}{(L_{AB} L_C)^k}},&n=2k+1,
\end{cases}
\]
where $C_k=\binom{2k}{k}/(k+1)$ is the $k$th Catalan number. Finally, when $N_{AB}>N_C$ and $N_{A}>N/2$, we obtain in the thermodynamic limit
\[
{\mathbb E}[p_n]\simeq\begin{cases}
L_C^{1-n}L_{B}^{2-n},&n=2k\\[4mm]
(L_C L_{B})^{1-n}&n=2k+1.
\end{cases}
\]
The case in which $N_{AB}>N_C$ and $N_{B}>N/2$ is obtained by replacing $L_{B}$ with $L_{A}$ in the latter formula.

While in Ref.~\cite{shapourian2021entanglement}, the previous expressions are used to compute the expectation value of the logarithmic negativity in the thermodynamic limit via
\[
{\mathbb E}[{\mathcal E}(\rho)]\simeq\lim_{n\to1/2}\log{\mathbb E}[p_{2n}]\,,
\]
in this work we are interested in quantities that involve only few low-order PT moments. The previous discussion suggests that in order to distinguish the phases in the asymptotic limit, one could also consider ratios of these PT moments, and combine them in such a way that the numerator and denominator have the same total degree in $L$. This led us to introduce the ratios of averaged PT moments $\tilde r_n$ as:
\begin{equation}
    \tilde r_n =
    \frac{
    \mathbb{E}[p_n]
    \mathbb{E}[p_{n+1}]
    }
    {
      \mathbb{E}[p_{n+2}]
    \mathbb{E}[p_{n-1}]
    }\,.
    \label{eq:tilde r_n}
\end{equation}

In this work we will consider the two quantities, $r_2$ evaluated for a single state and $\tilde{r}_2$ evaluated for an ensembles of states.
Despite the fact that the average of $r_2$ over an ensemble of states does not necessarily coincide with the corresponding $\tilde{r}_2$, we show that in App.~\ref{sec:cov} for Haar random states, the average of $r_2$ approximates $\tilde{r}_2$. More importantly, the phase diagram can be reproduced  with both quantities. As shown in this App.~\ref{sec:cov}, the fact that $r_2$ has small statistical fluctuations around $\tilde r_2$ is due to the fact that each PT moment $p_n$ has a relative variance $\mathrm{Var}(p_n)/\mathbb{E}[p_n]^2$ that decays exponentially with $N$, see illustration in Fig.~\ref{fig:phase_diagram_r_2}a) for $p_2$.
As a consequence, we can Taylor expand $r_2$ around $\tilde r_2$, and find that the relative variance $\mathrm{Var}(r_2)/\tilde r_2^2$ also decays exponentially with $N$, see Fig.~\ref{fig:phase_diagram_r_2}b). 
We have discussed to which extent $r_2$ evaluated for a single random state can be utilized to approximate $\tilde r_2$ (for Haar random states). Note that additional statistical fluctuations arise when estimating $\tilde r_2$ from a finite number of measurements in an experiment. We discuss those effects in Sec.~\ref{sec:measurements} and App.~\ref{app:measuring}.

It can be easily checked, using the formulas below Eq.~\eqref{eq:ExactExpression} that, in the asymptotic limit, the ratios given in Eq.~\eqref{eq:tilde r_n} only depend on $n$ and take different but constant values in the two entangled phases. Because it involves the lowest-order moments, we will mostly focus on the ratio $\tilde {r}_2$. It is also the ratio that best allows to distinguish the entangled phases, taking value $3/2$ in the ES phase and value $1$ in the PPT and ME phase, see Fig.~\ref{fig:phase_diagram_r_2}. 
As shown in the second panel of Fig.~\ref{fig:phase_diagram_r_2}, the phases can also be seen in the finite dimensional case, here for $N_{AB}=10$.
Let us remark here, that another simple function of the first three moments can be used to distinguish between the PPT and the ME phases (see Sec.~\ref{sec:detection}). We will describe below measurement protocols that would allow to measure this phase diagram of Haar-random induced mixed states with system sizes that are compatible with current experimental systems.

Moreover, even though it cannot be seen in the asymptotic limit, there is a small region where $\tilde r_2<1$  in finite dimensional systems (see Fig.~\ref{fig:phase_diagram_r_2}b). In Fig. ~\ref{fig:phase_diagram_r_2} we show that the value of $r_2$ is below $1$ for $N_C\in[N_{AB},N_{AB}+2]$ for finite $N$. From numerical computations, it can be seen that the minimum of $r_2$ in the phase diagram seems to be always within this interval.

It is interesting to note in this context that it was shown in Ref.~\cite{Au12} that PPT-entangled states can be found with high probability in a region slightly above that region. More precisely, Theorems 1 and 2 of Ref.~\cite{Au12} show that for $4L_{AB}\le L_C\le L_{AB}^{3/2}$ the two following statements hold:  (a) the probability that $\rho^\Gamma$ has negative eigenvalues is exponentially small, and (b) the probability that $\rho$ is separable is exponentially small (see Ref.~\cite{Au12} for details). This implies that for $N_C\in[N_{AB}+2,3N_{AB}/2]$, the probability of $\rho$ being PPT-entangled is large. Equivalently, this will occur when $N_C/N\in[0.5+1/N,0.6]$.

\begin{figure*}[t]
 \includegraphics[width = 0.99\linewidth]
 {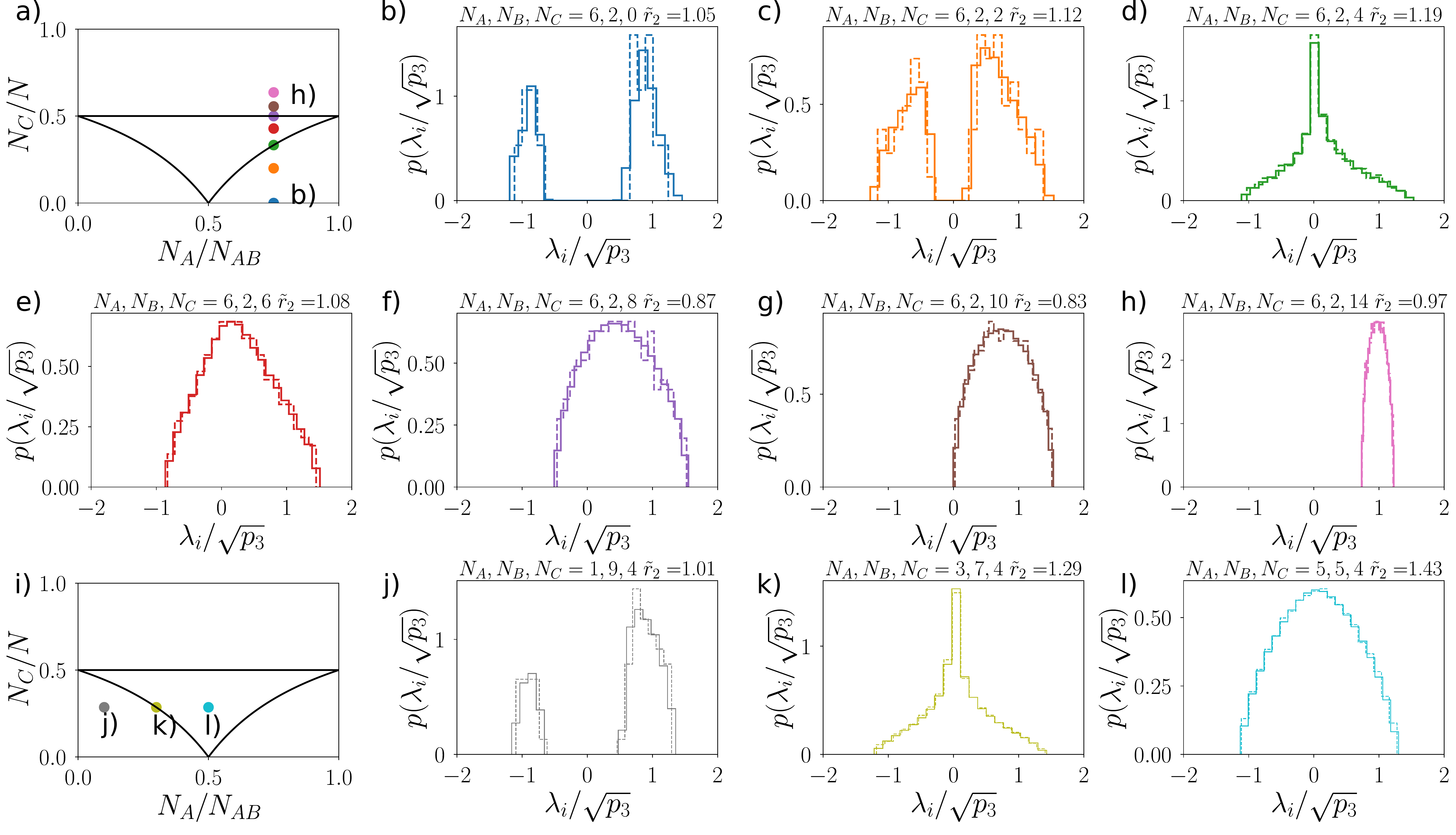}
    
  \caption{Distribution of the negativity spectrum $\{\lambda_i\}$ and $\tilde r_2$ for Haar random states.
  a)-h) Distribution of rescaled negativity spectra $\lambda_i/\sqrt{p_3}$ for a single random state (blue dashed line), and an average of $100$ random states (orange solid) for various points of the phase diagrams obtained from different $N_C=0,2,4,6,8,10,14$, with a fixed $N_A=6$, $N_B=2$.
  (same parameters as in Fig. 7 of Ref.~\cite{shapourian2021entanglement})
  i)-l) Same as a)-h) but for different values of $(N_A,N_B)=(1,9),(3,7),(5,5)$, and fixed $N_C=4$. 
  In the maximally entangled phase and the PPT phase, we observed two (one respectively) peaks centered around $\lambda_i=\pm \sqrt{p_3}$ ($1$). As a consequence, $\epsilon_i\approx 1$, and $\tilde r_2\approx 1$. Instead, in the entanglement saturation phase [panels e), l)], the peak is centered around $0$. 
  }
    \label{fig:haarspectra}
\end{figure*}

\subsection{Relation between $r_2$ and the negativity spectrum}
\label{sec:spectrum_subsec}

As shown in Ref.~\cite{shapourian2021entanglement}, the shape of the negativity spectrum (the spectrum of $\rho^\Gamma$) is very distinct for the three phases (see in particular Fig.~2 of Ref.~\cite{shapourian2021entanglement}).
More precisely, plotting the density of the eigenvalues of the PT operator, leads either to a single positive ``peak'' (in the PPT phase), to a function resembling a triangle around 0 (in the ES phase), or to two separate ``peaks'' one around a positive and one around a negative eigenvalue (in the ME phase). Here, we want to use this result and the behavior of $\tilde{r}_2$ to determine where these peaks are centered. To this end, we will consider the density not as a function of the eigenvalues, but as a function of rescaled and squared eigenvalues. 

We will first consider analytically a simplified situation, where we consider $r_2$ evaluated on a single state and model the peaks by delta distributions. Then, we will consider finite system sizes and illustrate this behaviour for $\tilde{r}_2$ (i.e. the average) numerically. 

Let us start out by considering $r_2$ evaluated on a single state to explain this behavior. Rewriting $r_2$ as a function of the negativity spectrum leads to 
\begin{equation}
    r_2= 
    \frac{\sum_i \lambda_i^2 p_3}{\sum_i \lambda_i^4}
    =
    \frac{\sum_i \epsilon_i}{\sum_i \epsilon_i^2}
    ,\label{eq:r2vsepsilon}
\end{equation}
where the sum is over eigenvalues $\{\lambda_i\}_i$ of the PT operator $\rho^\Gamma$. In the second equality, we have introduced the rescaled squared spectrum $\epsilon_i=\lambda_i^2/p_3$.

We show now that in the PPT and ME phases, where peaks occur in the negativity spectrum, the peaks are centered at $\epsilon_i=1$. Stated differently, we show that $r_2=1$ implies that $\epsilon_i=1,0$ for all $i$.

In case there is one peak, i.e. $\lambda_i=\lambda_1$ for all $i$, we have $r_2=1$ iff $L_{AB} \epsilon_1=L_{AB} \epsilon_1^2$, which proves the statement. In case there are two peaks, i.e. $\lambda_i$ occurs $k_i \neq 0$ times for $i=1,2$, the condition $r_2=1$ holds iff $k_1 \epsilon_1+k_2 \epsilon_2=k_1 \epsilon_1^2+k_2 \epsilon_2^2$. Inserting the definition of $\epsilon_i$ and using the normalization condition ${\rm tr}\rho^\Gamma=1$, the previous condition implies straightforwardly that $\epsilon_i=1,0$ for all $i$~\footnote{The normalization condition implies that $k_2=(1-k_1\lambda_1)/\lambda_2$. Inserting this expression for $k_2$ into the equation $r_2=1$ leads to $k_1\lambda_1(-1+k_1\lambda_1)(\lambda_1-\lambda_2)^2(\lambda_1+\lambda_2)=0$. Hence, the only non-trivial solutions to this equation and the normalization condition are $\lambda_1=\pm\lambda_2=(k_1\pm k_2)^{-1}$. In both cases $\epsilon_i=0,1$ for all $i$.}.

 In order to demonstrate that this conclusion does not only hold for the extreme case of delta distributions (and single states), we analyze the negativity spectrum numerically. In Fig.~\ref{fig:haarspectra}, we show that the  behavior mentioned above is robust within the various phases by studying how the negativity spectrum changes when the relative sizes of the tripartition vary. More precisely, panels (b)--(h) show the negativity spectrum for different values of $N_C/N$ for a fixed ratio $N_A/N_B$, as depicted in panel (a). In addition, in panels (j), (k) and (l) we show the negativity spectrum for different values of $N_A/N_{AB}$ with fixed $N_C$, as indicated in panel (i). Clearly, the PPT phase is characterized by a negativity spectrum with positive semi-definite support whereas in the ME and ES phases the probability density of negative eigenvalues is non-vanishing~\cite{shapourian2021entanglement}. Note that, as illustrated in panel (h), in the interior of the PPT phase $r_2\simeq1$ and the probability density has a single ``peak'' around $+\sqrt{p_3}$. Finally, panels (b), (c), and (j) show that the probability density has two ``peaks'' around $\pm\sqrt{p_3}$ in the interior of the ME phase for different values of $N_A$, $N_B$, $N_C$. Summarizing, while Ref.~\cite{shapourian2021entanglement} showed the shapes of the negativity spectrum in the various entanglement phases, here the value of $r_2$ allows us, using these results, to identify the locations of these peaks when $r_2=1$.

As we will show in Sec.~\ref{sec:detection} $r_2$ can also be utilized to detect a negative eigenvalue of the PT, i.e. to detect entanglement. Moreover, we will show there that the $p_3$-PPT condition~\cite{elben2020mixed} can be utilized to differentiate between phase I and phase II.

\section{Effect of white noise}
\label{sec:noise}

In this section we apply our procedure to pseudo Haar-random induced mixed states, which are convex combinations of Haar-random induced mixed states and some amount of white noise determined by a parameter $\epsilon$ (see Eq. (\ref{Eq:PseudoRandom})). 
As we will see, $\tilde r_2$, i.e. the mean value of low order PT moments, can be easily computed for the ensemble of pseudo Haar-random induced mixed states. This allows us to compute the corresponding phase diagram. In case it is known that the states generated in an experiment are pseudo random, the phase diagram can be utilized to determine the value of the parameter $\epsilon$.

We start defining the ensemble of pseudo Haar-random pure states with parameter $\epsilon$ as the set of states
\begin{equation} \label{Eq:PseudoRandom}
\rho'_{ABC}=(1-\epsilon)\ket\psi_{ABC}\bra\psi+\epsilon{\bf 1}_{ABC}/L\,,
\end{equation}
where ${\bf 1}_{ABC}$ denotes the $L\times L$ dimensional identity matrix with $L=2^N$ and $\ket\psi_{ABC}$ is Haar-random. From here, we obtain pseudo Haar-random {\em induced} mixed states as
\begin{equation}
\rho'={\rm tr}_C(\rho'_{ABC})=(1-\epsilon)\rho+\epsilon{\bf 1}_{AB}/L_{AB}\,,
\end{equation}
where $\rho={\rm tr}_C\ket\psi_{ABC}\bra\psi$ is a Haar-random induced mixed state, ${\bf 1}_{AB}$ denotes the identity matrix, and \mbox{$L_{AB}=2^{N_{AB}}$}.

Let $p_n'$ be the PT moment of a pseudo Haar-random induced mixed state and $p_n$ those of a Haar-random induced mixed state. Then, the mean values ${\mathbb E}[p_n']$ can be expressed in terms of the mean values ${\mathbb E}[p_n]$ as
\begin{equation}
{\mathbb E}[p_n']=\sum_{k=0}^n\binom{n}{k}(1-\epsilon)^k{\mathbb E}[p_k](\epsilon/L_{AB})^{n-k}\,,
\end{equation}
with ${\mathbb E}[p_0]=p_0=L_{AB}$. Clearly we have ${\mathbb E}[p_n']={\mathbb E}[p_n]$ if $\epsilon=0$; and ${\mathbb E}[p'_n]=L_{AB}^{1-n}$ if $\epsilon=1$. One can use the previous expression to compute the phase diagram of the ensemble of pseudo Haar-random induced mixed states with noise parameter $\epsilon$ for all $\epsilon\in[0,1]$. This phase diagram interpolates between the one for Haar-random induced mixed-states for $\epsilon=0$ and the trivial one associated to the maximally mixed states with $\tilde{r}_2\equiv1$ everywhere for $\epsilon=1$. For intermediate values $0<\epsilon<1$, the following two observations can be made: (i) the lower boundary of the PPT phase (located along the horizontal line $N_C/N=0.5$ for $\epsilon=0$) goes down to some value $N_C/N<0.5$ that depends on $\epsilon$ as depicted in Fig.~\ref{fig:whitenoise}; (ii) in the region corresponding to the ME phase for Haar-random states (where $\tilde r_2=1$ for $\epsilon=0$), $\tilde r_2\simeq1-\epsilon$ for the values of $1-\epsilon$ considered in Fig. \ref{fig:whitenoise} b). These results can be understood from the leading order contributions in the thermodynamic limit to PT moments in the two entangled phases. 

\begin{figure}[t]
    \centering
    
    \includegraphics[width = 0.99\linewidth]{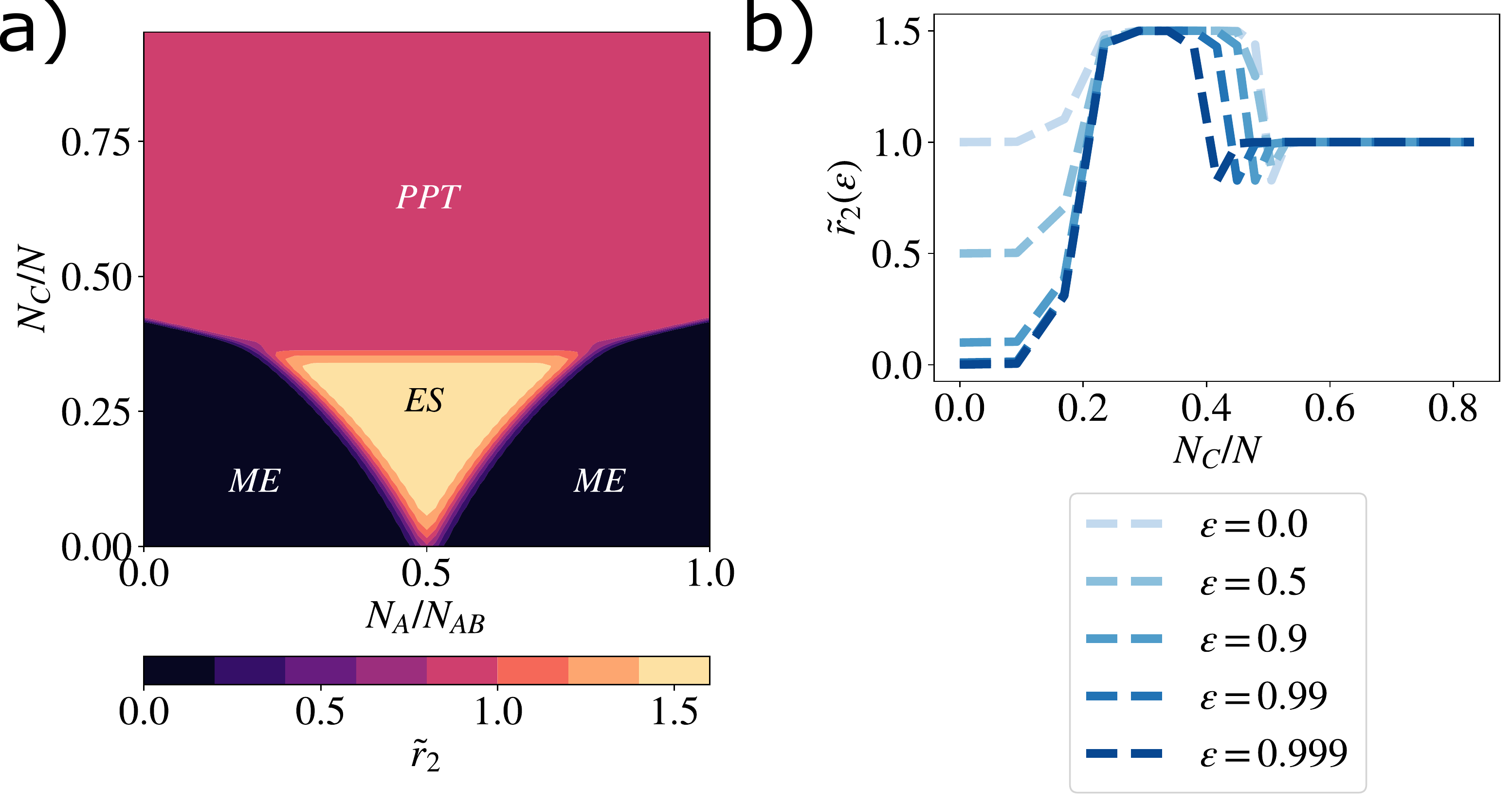}
    
    \caption{$\tilde{r}_2$ for Haar-random induced mixed states in presence of white noise. a) $\tilde{r}_2$ for all the possible tripartitions  with $N_{AB}=64$, and a white noise contribution of $\epsilon = 1 - 10^{-4}$, comparable with those in the experiment of Ref. \cite{arute2019quantum}.
    b) $\tilde{r}_2$ for different $\epsilon$, $N_{AB} = 64$ and $N_A=24$. In the ME region, $\tilde{r}_2$ decreases linearly with $1 - \epsilon$ for the parameters considered, whereas the $\tilde{r}_2=3/2$ region associated with the ES phase shrinks with increasing $\epsilon$.  
    }
    \label{fig:whitenoise}
\end{figure}

\section{Aspects of simulatability revealed by $r_2$: Stabilizer states}
\label{sec:stabilizers}
One may wonder whether the phase diagram revealed by $r_2$, $\tilde r_2$ or the negativity for Haar-random induced mixed states changes if one consider a different ensemble of quantum states. Here, we determine the values of $r_2$ for a class of quantum states which play an important role in the classical simulation of quantum computations: stabilizer states. We observe strong differences compared to the situation of Haar-random states. We will complement these results in Sec.~\ref{sec:MPS-FF} for other classes of states which are classically simulable, namely a class of random MPS and the class of fermionic Gaussian states. Note that in contrast to before, we consider here $r_2$ evaluated for a single stabilizer state and do not consider an average. This will be enough, as we will show, since for stabilizer states, $r_2$ can be seen to be always $1$.

Stabilizer states, sometimes also referred to as Clifford states, can be written as \mbox{$\ket{\psi}=U\ket{0}^{\otimes N}$}, where $U$ belongs to the $N$-qubit Clifford group. This group contains all unitary operators $U$ which map (under conjugation) any $N$-qubit Pauli operator $\sigma$ to some $N$-qubit Pauli operator, $\sigma'$, i.e. $\sigma'=U\sigma U^\dag$. According to the Gottesman-Knill theorem, the output of a Clifford circuit $U$ applied to a computational basis state can be simulated classically efficiently~\cite{gottesman1998heisenberg, gottesman2004simulation}. 

As shown in Ref.~\cite{bravyi2006ghz},
any three-partite stabilizer states $\ket{\psi}$ can be decomposed into GHZ states, Bell states, and product states, distributed among the three parties, $A$, $B$, and $C$~\cite{bravyi2006ghz}. That is, $\ket{\psi}$ can be written as
\begin{eqnarray}
\label{eq:GHZ}
\ket{\psi} &=& U_A U_B U_C\ket{0}^{\otimes s_{A}}
\ket{0}^{\otimes s_{B}}
\ket{0}^{\otimes s_{C}}
\ket{\mathrm{GHZ}}^{\otimes g_{ABC}}
\nonumber \\
&&
\ket{\mathrm{EPR}}^{\otimes e_{AB}}
\ket{\mathrm{EPR}}^{\otimes e_{AC}}
\ket{\mathrm{EPR}}^{\otimes e_{BC}},
\end{eqnarray}
with $U_A,U_B,U_C$ unitary Clifford operators on $A,B,C$, respectively. Using this decomposition 
and the fact that $p_n(\rho\otimes \sigma)=p_n(\rho)p_n(\sigma)$, 
it is straightforward to obtain the following PT moments 
\begin{eqnarray}
p_2 &=& \left(\frac{1}{2}\right)^{g_{ABC}+e_{AC}+e_{BC}}
\nonumber \\
p_3 &=& \left(\frac{1}{4}\right)^{e_{AB}} \left(\frac{1}{4}\right)^{g_{ABC}+e_{AC}+e_{BC}}
\nonumber \\
p_4 &=& \left(\frac{1}{4}\right)^{e_{AB}} \left(\frac{1}{8}\right)^{g_{ABC}+e_{AC}+e_{BC}}.
\label{eq:pnstab}
\end{eqnarray}
With all that it is easy to see that 
\begin{equation}
    r_2=1 \text{ for all Clifford states.}
\end{equation}
Hence, in stark contrast to random states, $r_2$ takes  a fixed value, which is independent of the stabilizer state and the system sizes. 

Given the decomposition above, it can also be seen that the negativity spectrum of stabilizer states  is constrained to two values $\lambda_i=\pm \sqrt{p_3}$, i.e all eigenvalues $\lambda_i$ of $\rho^\Gamma$ are either $\sqrt{p_3}$ or $-\sqrt{p_3}$. This is because each Bell pair between $A$ or $B$ and $C$, and each GHZ state in Eq.~\eqref{eq:GHZ} gives a $1/2$ multiplicative contribution to the negativity spectrum, while the $e_{AB}$ Bell pairs between $A$ and $B$ give a $\pm 1/2$ contribution. 
Therefore, this type of negativity spectrum is analogous to the ones of the PPT and ME phases of Haar-random states with $r_2\approx 1$.
However, if one measures in an experiment $r_2\neq 1$, e.g, in the ES phase for a Haar-random state, it proves that the state is not a stabilizer state and thus cannot be generated via Clifford gates, which are classically efficiently  simulable~\cite{haferkamp2020quantum}.

One may wonder what happens when  Clifford circuits are doped with $T$ gates, which make them universal for quantum computations. 
The question of convergence of the output of doped Clifford circuits to Haar-random states has been studied in Ref.~\cite{leone2021quantum}. In this work, we will focus on the transition of another class of constraint states to Haar-random states by considering fermionic Gaussian states (see Sec.~\ref{sec:MPS-FF}). However, let us mention here that recently, measures of ``magic'' have been introduced to quantify how distant a given quantum states is from the set of stabilizer states, in particular in terms of quantum resources ~\cite{haferkamp2020quantum,leone2022magic}.
The quantity $r_2-1$ vanish for Clifford states, but it does not measure how resourceful a state is. This can be easily understood by the fact that $r_2$ is invariant under local unitaries i.e.,  applying a local, non-Clifford, operation on a stabilizer state will also result in $r_2=1$. This is in contrast to the measures of ``magic'' introduced in Ref.~\cite{leone2022magic} that would detect such non-Clifford operations, and that are invariant under global entangling Clifford operations applied on non-Clifford states. Instead, what $r_2$ characterizes is a sort of \emph{magic entanglement} structure of non-Clifford states: any state with $r_2\neq 1$ has an entanglement content that cannot be generated using a Clifford circuit followed by local unitary operations.

Let us finally mention that for stabilizer states the negativity is given by the simple function~\cite{bravyi2006ghz} 
\begin{equation}\label{eq:negp3}
\mathcal{E}(\rho)=e_{AB}=\frac{1}{2}\log_2(p_2^2/p_3)\text{ for all Clifford states.}
\end{equation}
This shows that stabilizer state are PPT iff they satisfy the  $p_3$-PPT condition~\cite{elben2020mixed}, which states that for any PPT state it holds that $p_3\ge p_2^2$. In other words, $p_3<p_2^2$ implies that the partial transpose of the state is not positive semi-definite. In general there exist, of course, states which are not PPT and for which $p_3\ge p_2^2$. However, equation~\eqref{eq:negp3} shows that a Clifford state is PPT (has zero negativity) if and only if the $p_3$-PPT condition is satisfied. In fact, for stabilizer states Eq. (\ref{eq:pnstab}) implies that  the state is separable if and only if $p_3=p_2^2$. Otherwise, the $p_3$-PPT condition is violated. Thus, we always have $p_3\ge p_2^2$, which ensures that the expression of the negativity in Eq.~\eqref{eq:negp3} is always non-negative.

\section{$r_2$ is a test for Haar random states: Case study with the PXP model}
\label{sec:rydberg}

\begin{figure}
    \centering
    \includegraphics[width = 0.95\columnwidth]{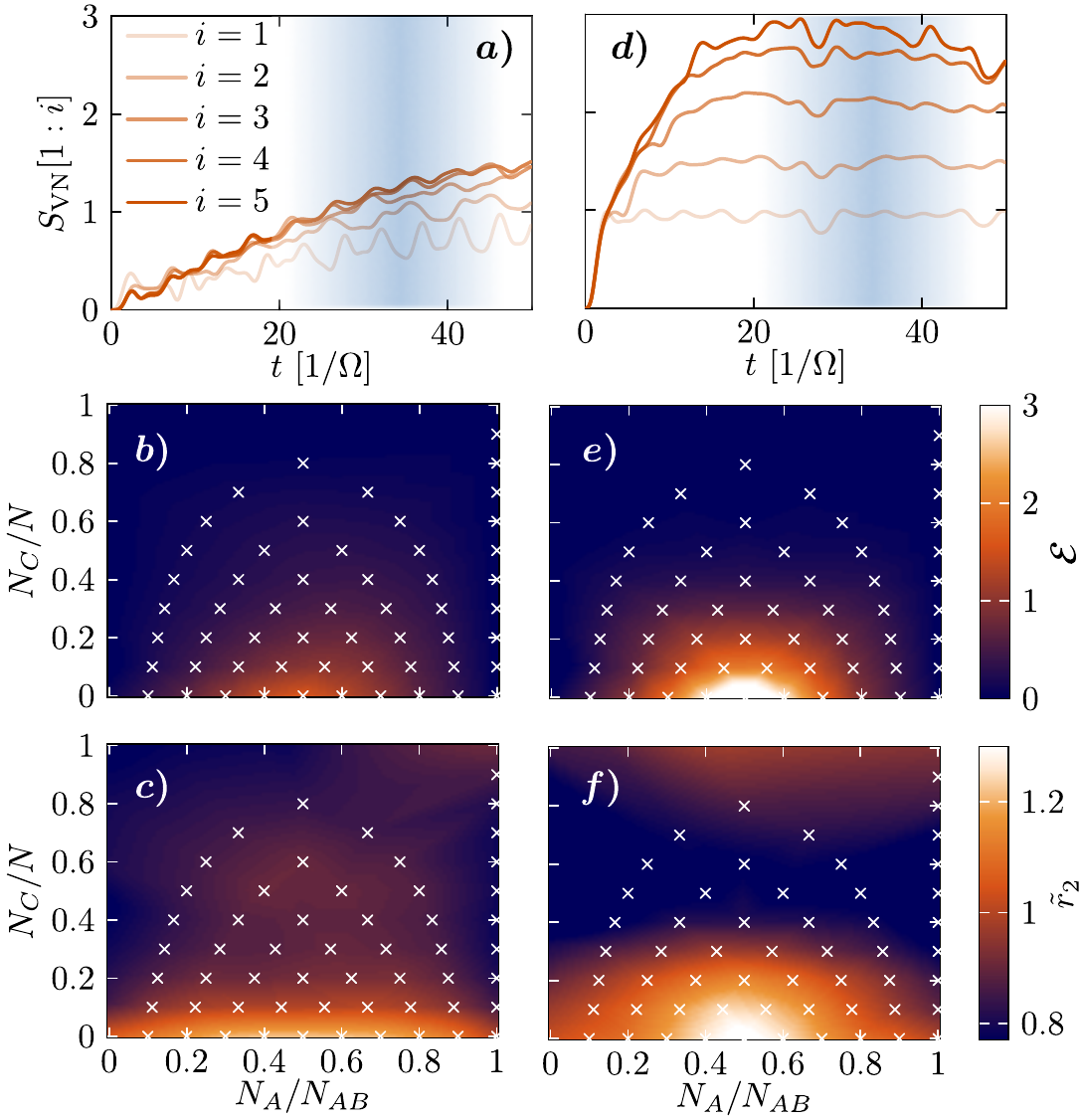}
    \caption{Entanglement structure of quantum many-body states in nonequilibrium dynamics of a PXP-model [Eq.~(\ref{eq:pxp}) in main text]. Panels a)-c): constrained dynamics, performing a quench from an antiferromagnetic initial state $\ket{\psi_0} = \ket{\mathbb{Z}_2} = \ket{0 1}^{\otimes N/2}$. Panels d)-f): quench from a maximally polarized state $\ket{\psi_0} = \ket{0}^{\otimes N}$, resulting in ergodic many-body dynamics. a) Von-Neumann entanglement entropy as a function of time for different bipartitions of a spin chain with $N = 10$. b) Negativity for fixed $N = 10$ for all possible connected tripartitions of the chain denoted by the white crosses. c) Corresponding phase diagram for the ratio $\tilde r_2$. The data in the panels b), c) and e), f) have been obtained by averaging the quantities over 300 states from the time interval $t \in [20, 50]\,1/\Omega$ indicated by the shaded area in the panels a) and d).}
    \label{fig:pxp}
\end{figure}

We now turn to the discussion on how to exploit the properties of $r_2$ to characterize the entanglement structure of quantum many-body states in quantum simulation experiments. In this section we focus on systems based on Rydberg atoms trapped in optical tweezers \cite{Browaeys2020}, which have been used recently to realize a large variety of correlated phases of matter, ranging from ground states of 1D and 2D spin models \cite{Bernien2017, Ebadi2021} to topological states \cite{Syl2019} and quantum spin liquids \cite{Semeghini2021}. In our context, Rydberg systems are of particular interest as they allow for the implementation of chaotic quantum many-body systems, where the entanglement structure of states generated by quenching in the long-time limit shares properties with the entanglement structure of Haar-random states~\cite{Wen2019}.

For the subsequent analysis we will focus on the dynamics of Rydberg atoms in a 1$D$-chain as previously studied in Ref.~\cite{Wen2019, Bernien2017}. Here, entanglement is generated via the Rydberg blockade mechanism. In particular, atoms located within the blockade radius cannot be simultaneously excited to the Rydberg state, due to the large interaction between Rydberg excited atoms. For a 1$D$-chain where the blockade affects only nearest-neighbour sites, the system is effectively described by a PXP-model
\begin{align} \label{eq:pxp}
H = \Omega \sum_{i} \mathcal{P} X_i \mathcal{P}.
\end{align}
Here the operator $\mathcal{P}$ constraints the Hilbert space by projecting out all states where two adjacent atoms are in the Rydberg state, i.e. $\mathcal{P} = \prod_i \left( \mathbb{1}_i  \mathbb{1}_{i+1} - Q_{i}Q_{i+1}  \right)$, where the operators  $Q_i$ are the local projectors $Q_i = \ket{1}_i \! \bra{1}$. Recently, the model (\ref{eq:pxp}) has attracted great interest due to its connection to quantum many-body scarring \cite{Serbyn2021, TurnerB2018, Cheng-Ju2020}. Despite the fact that the Hamiltonian (\ref{eq:pxp}) is non-integrable and quantum-chaotic \cite{Wen2019}, quench dynamics from specific unentangled product states lead to constrained dynamics with long-lived periodic revivals accompanied by suppression of thermalization. 

We first study the entanglement structure of states in the constrained case, by simulating a quantum quench $\ket{\psi(t)} = e^{-i H t} \ket{\psi_0}$ from a staggered initial state $\ket{\psi_0} = \ket{\mathbb{Z}_2} = \ket{10}^{\otimes N/2}$. As shown in Ref.~\cite{Wen2019}, in this case the state $\ket{\psi(t)}$ is well described by a MPS with low bond dimension.

This is also reflected in the slow growth of entanglement entropy in Fig.~\ref{fig:pxp} a). In Fig.~\ref{fig:pxp} b) we analyse the averaged entanglement negativity of the partial transpose $\rho^{\Gamma}$ for all possible connected tripartitions $\{N_A, N_B, N_C\}$ of the chain. For $N_C \ll N$, the negativity is maximal around $N_A/N_{AB} = 0.5$. Interestingly, the ratio $\tilde{r}_2$ in Fig.~\ref{fig:pxp} c) shows a quantitative different behavior that is not captured by the negativity. Close to $N_C/N = 0$, we observe a band in the horizontal direction in which $\tilde{r}_2$ saturates to a value $\tilde{r}_2>1$. With increasing $N_C$, the phase diagram shows an extended region where $\tilde{r}_2<1$. As we will see in section \ref{sec:mps}, both features are related to the finite correlation length and associated finite bond-dimension of the underlying state.
This example shows that, when the dynamics is constrained, $\tilde{r}_2$ shows a different behavior compared with random states.

For generic unentangled initial states, the dynamics of the system is ergodic with quick thermalization of local observables. The entanglement entropy Fig.~\ref{fig:pxp} d) grows linearly and quickly saturates to a value close to the Page entropy of a random state \cite{Page1993}. In this case the averaged Negativity Fig.~\ref{fig:pxp} e) essentially shows the same features as for Haar-random states \cite{shapourian2021entanglement}. We observe a peak in the Negativity for $N_C = 0$ and $N_A/N_{AB} = 0.5$, which broadens and fades out as the size of the bath $N_C$ is increased. Similar features are visible when analysing the ratio $\tilde{r}_2$. Here we additionally observe a band close $N_C/N = 0.5$ with $\tilde{r}_2 < 1$. As discussed above, slightly above this region it has been proven that PPT-entangled states are likely to be found.

We emphasize that in contrast to the negativity [Fig.~\ref{fig:pxp} b), e)], the ratio $\tilde{r}_2$ is easily accessible in current experimental settings. As discussed in Ref.~\cite{notarnicola2021randomized}, randomized measurements for obtaining moments of $\rho^{\Gamma}$ can be implemented in settings based on Rydberg atoms. Recently, direct measurement of R\'enyi entanglement entropies has been experimentally demonstrated in dynamically reconfigurable Rydberg arrays by applying beam-splitting operations as Bell-measurements between two copies of an atom array \cite{bluvstein2021quantum}. These ideas can be readily extended to measuring moments of the partially transposed density matrix based on preparing multiple copies of the same quantum state~\cite{Horodecki2003,carteret2005noiseless,graymachine2018}, see also our discussion in  Sec.~\ref{sec:measurements}.

\section{$r_2$ for two classically simulable class of states} 
\label{sec:MPS-FF}

We have discussed how PT moments reveal via the quantity $\tilde r_2$ the phase diagram of Haar random states, while exhibiting striking differences with Clifford states, and non-ergodic states of the PXP model.
We now show that $\tilde r_2$ shows also a distinctive behavior for two other important classes of quantum states: MPS and fermionic Gaussian states.

\subsection{Matrix-product states} \label{sec:mps}

\begin{figure}[t]
    \centering
    \includegraphics[width = 0.99\linewidth]{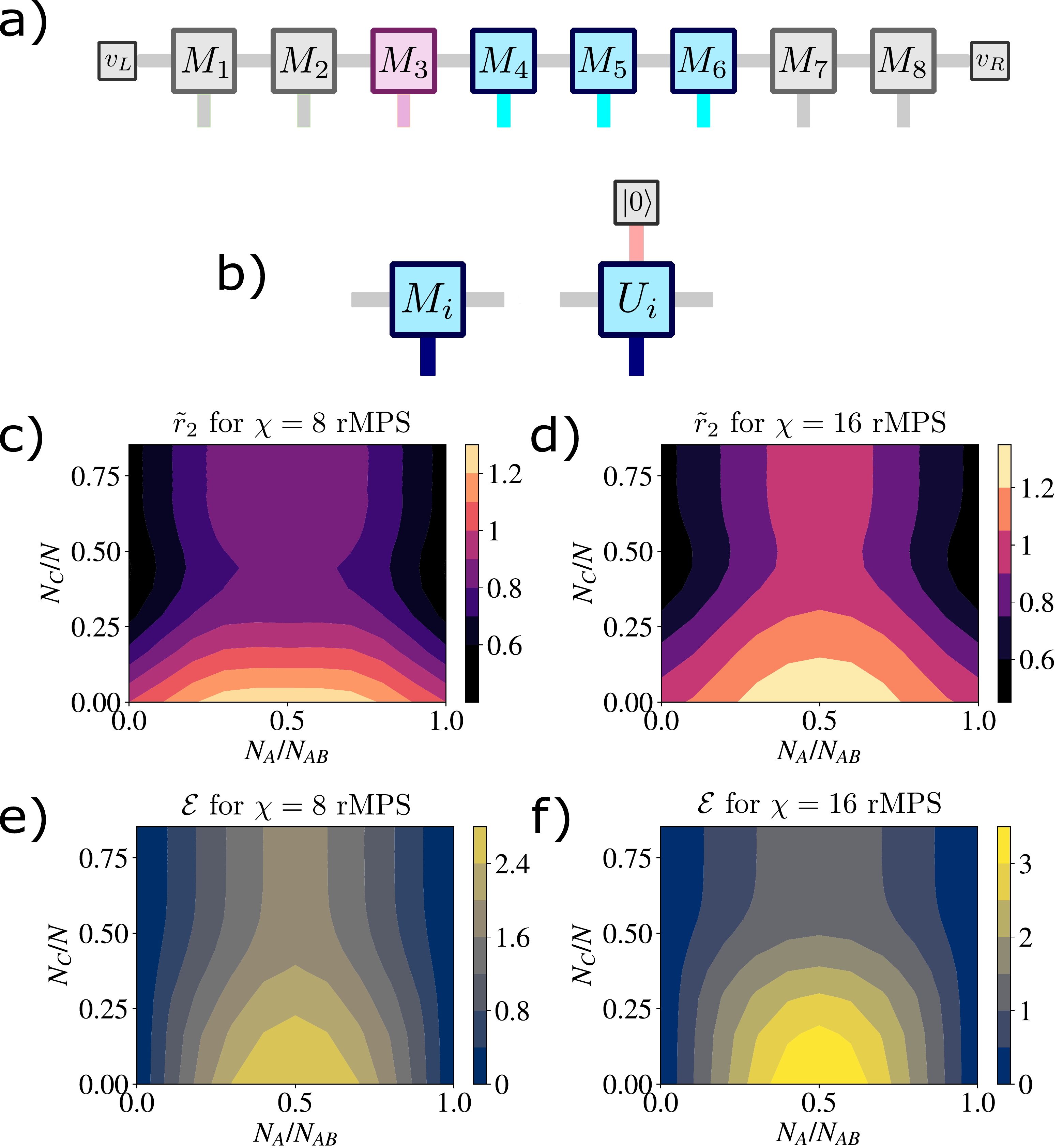}
     \caption{
     a) Pictorial representation of a matrix-product state of $N = 8$ qubits with the typical tripartition we consider. The gray matrices represent the qubits belonging to the region $C$, while the red (blue) matrix belongs to the region $A$ ($B$) respectively. (b) Definition of the matrices for rMPS. A random unitary $U_i$ is reshaped, i.e. its first index is partially contracted with a qubit in the state $\ket{0}$. 
    c), d) Ratio $\tilde{r}_2$ of the PT moments and  e),f) average logarithmic negativity for an ensemble of rMPS (definition in the text), with fixed $N_{AB}=10$ and bond dimensions respectively $\chi = 8$ and $\chi = 16$ (equivalent colorscales). These data have been obtained numerically averaging over $\sim 10^2$ random realizations of the rMPS.}
    \label{fig:MPS_r2}
\end{figure}

MPS form a class of quantum states with low level of entanglement that can describe in particular  ground states of gapped local Hamiltonians in one dimension \cite{eisert2010colloquium, hastings2007}.
In this section we describe how the ratio $\tilde{r}_2$ shows a different behavior compared to Haar random states.

A MPS describing the state of $N$ qubits can be written as \begin{equation}
    \ket{\psi} = \sum_{\sigma\in \{0,1\}^N} v_L^{T}M^{\sigma_1}_{1}M^{\sigma_1}_{2}\dots M^{\sigma_N}_{N}v_R\ket{\sigma_1,\dots,\sigma_N},
\end{equation}
where $M_{i}^{\sigma_i}$ are $\chi \times \chi$ matrices, $v^L$ and $v^R$ are vectors of length $\chi$.
Noting that the von Neumann entropy (and any R\'enyi entropy) between two connected partitions $A$ and $B$ is upper bounded by $\log \chi$ \cite{schollwock2011mps}, the bond dimension $\chi$ is the key parameter that controls the amount of entanglement of the MPS. 

Here, we consider a distribution~\cite{garnerone2010tipicality} of rMPS, which are obtained by drawing from the Haar measure a unitary matrix from the $U_i(2\chi)$ group for each site $i=1,\dots,N$ independently, and defining:
\begin{equation}
    [M_i^{\sigma_i}]_{\ell,\ell'} = [U_i]_{\ell,  \ell'+\chi\sigma_i}.
\end{equation}
The vectors components of $v^L$ and $v^R$ are sampled using independent Gaussian complex variables of zero mean and unit variance. When all the random variables have been initialized, we normalize the vector $\ket{\psi}$. 
We calculate numerically both the negativity and the PT moments using the algorithm presented in Ref.~\cite{ruggiero2016randomsinglet}, for various bond dimensions $\chi$. We consider here that the partitions $A$ and $B$ are adjacent and placed at the middle of the chain, see  Fig.~\ref{fig:MPS_r2}a).

In Fig.~\ref{fig:MPS_r2}c)-d) we show $\Tilde{r}_2$
for two values of $\chi=8,\ 16$. As a first notable difference with respect to Haar random states, we observe that for $N_C/N\ll1/2 $, $r_2>1$ for a large interval of values of $N_A/N_{AB}$. 
Interestingly, this region $r_2>1$ corresponds to a saturation of the negativity when varying $N_A/N_{AB}$ for a fixed $N_C$, c.f panels (e) and (f). 
In the limiting case of a pure state $N_C\to 0$, we can understand this saturation of the negativity
as a consequence of the finite bond dimension of the rMPS. Indeed, for pure states, the negativity can be shown to be upper bounded by $\log(\chi)$~\cite{calabresenegativity2012}, which is consistent with the two plateau values shown in panels (e) and (f) for $\chi=8$ and $\chi=16$. 

A second important observation is that $r_2<1$ in the limit $N_C\gg N/2$. In this case, the two partitions $A$ and $B$ are NPT entangled, as shown by the finite value of the negativity in panels (e) and (f). This can be interpreted as follows: for Haar random states, we have seen that the density matrix $\rho$ converges to a PPT density matrix with $r_2=1$ as $N_C$ increases (intuitively, adding a qubit in the bath $C$ always make the reduced state $\rho$ more mixed, until we reach the maximally mixed state). 
Here instead with rMPS, we obtain a NPT state for arbitrary large $N_C$ because the bond dimension introduces a finite correlation length between $AB$ and $C$~\cite{schollwock2011mps, eisert2010colloquium, eisert2021randomMPS}. 

\subsection{Fermionic states}\label{sec:fermionic}

In this section we study the behavior of the ratios $r_2$, $\tilde r_2$ for the ensemble of random fermionic Gaussian mixed states and show that it is again distinct from all the previously studied classes of states. Moreover we us $\tilde r_2$ to observe the transition from classically simulable states (fermionic Gaussian states) to Haar random states. To this end, we consider the change of $\tilde r_2$ as a function of the number of SWAP gates which dope the corresponding classically simulable circuit. 

\subsubsection{Definitions}

Fermionic Gaussian states have being studied in the context of entanglement characterization ~\cite{shapourian2021entanglement,murciano2021quench,murciano2021quench,Eisler_2015,eislerzimboras2016} and are also of interest in quantum computation as fermionic Gaussian {\em pure} states can be seen as the output of Matchgate (MG) circuits~\cite{Valiant2001,JM08,TD02}. This connection between fermionic Gaussian states and MG circuits can be used to define properly an ensemble of random mixed states and to compute the corresponding phase diagram associated with the ratios $r_2$, $\tilde r_2$. The idea is to uniformly sample MG circuits and then consider the reduced states of the resulting wavefunction. As we will explain below, the uniform sampling of a MG circuit $U$ acting on $N$ qubits can be done efficiently since they are characterized by a special orthogonal matrix $R\in{\rm SO}(2N)$, and the special orthogonal group has a unique invariant (Haar) measure induced by that of the unitary group ${\rm U}(2N)\supset{\rm SO}(2N)$. The reduced state of the pure state $U\ket{0^N}$ will be a fermionic Gaussian (mixed) state completely characterized by a correlation matrix scaling linearly with $N$ that can be efficiently computed~\cite{Peschel_2003,Peschel_2009,Eisler_2015} from the one of $U\ket{0^N}$. From this correlation matrix the PT moments can be determined, as we will explain below. Therefore, we can deal with much larger system sizes compared to the case in which we consider the output of a universal quantum computation.
In contrast to that, we consider in the subsequent subsection quantum circuits that are no longer  efficiently classically simulable by including additional resourceful gates like the SWAP gate.
As the number of resourceful gates increases, the circuits become universal. As we show here, this transition, from fermionic Gaussian states to Haar-random states as a function of the number of SWAP gates can be observed with $\tilde{r}_2$.  

We are interested in fermionic Gaussian mixed states defined on the Hilbert space of $N$ (ordered) fermionic modes/sites that we identify with the numbers $1,2,\ldots,N$. A fermionic Gaussian state can be written in the form
\begin{equation}\label{eq:FGstateExp}
\rho_{ABC}\propto\exp\bigg(\frac14\sum_{j,k=1}^{2N} W_{jk}c_jc_k\bigg)\,,
\end{equation}
where $W$ is a ($2N\times 2N$) purely imaginary antisymmetric matrix and $c_j$ are (anticommuting) Majorana fermionic operators.
Due to the relation 
$G=\tanh(W/2)$ \cite{Peschel_2003,Peschel_2009}, with the ($2N\times 2N$) covariance matrix $G$ with matrix elements given by $G_{jk}=(1/2){\rm tr}(\rho_{ABC}[c_j,c_k])$, such a density matrix can be uniquely characterized by its  covariance matrix. 
As mentioned before, fermionic Gaussian pure states have been shown to be equivalent to those states generated by MG circuits through a Jordan-Wigner (JW) transformation~\cite{Valiant2001,JM08,TD02}. The JW transformation is a unitary mapping from a $N$-modes fermionic state to a $N$-qubits ($N$-spins) state. In terms of the $2N$ Majorana fermionic operators, the JW mapping can be described by the well-known relations
\begin{equation}
\begin{aligned}
    c_{2k-1}&=\prod_{i<k}(\sigma^z_i)\sigma^x_k\\
    c_{2k}&=\prod_{i<k}(\sigma^z_i)\sigma^y_k,
\end{aligned}
\end{equation}
where $\sigma^x$, $\sigma^y$ and $\sigma^z$ denote the Pauli matrices. The $N$ fermionic creation (annihilation) operators $a^{\dagger}_k$ ($a_k$), for $k=1,\dots,N$ are related to the $2N$ Majorana fermionic operators via the equations  $c_{2k-1}=a_k+a_k^{\dagger}$ and $c_{2k}=-i(a_k-a_k^{\dagger})$. A state 
\begin{equation}\label{eq:fermionic}
\ket{\Psi}=\sum_{i_1,\ldots,i_N\in\{0,1\}}\alpha_{i_1,\ldots,i_N}(a_1^\dagger)^{i_1}\cdots(a_N^\dagger)^{i_N}\ket{\Omega}\,,
\end{equation}
with $\ket{\Omega}$ the Fock vacuum, can be related to the $N$-qubits state
\begin{equation}\label{eq:qubits}
\ket{\Phi}=\sum_{i_1,\ldots,i_N\in\{0,1\}}\alpha_{i_1,\ldots,i_N}\ket{i_1,\ldots,i_N}\,.
\end{equation}
Fermionic states~\cite{BK02} are those states of the form~\eqref{eq:fermionic} whose $N$-qubits representation~\eqref{eq:qubits} is an eigenstate of $\sigma_z^{\otimes N}$. 

Let us consider $N$-qubit states that are the output of nearest-neighbors MG circuits~\cite{Valiant2001,JM08,TD02}, i.e. we consider states of the form  $\ket\Phi=U\ket{0^N}$ where $U$ is a product of two-qubits match gates $M$ acting on nearest neighbors. Any match gate, $M$, can be written as 
\[
    M= \begin{pmatrix}
    u_{00} & 0 & 0 & u_{01} \\
    0 & v_{00} & v_{01} & 0 \\
    0 & v_{10} & v_{11} & 0 \\
    u_{10} & 0 & 0 & u_{11} 
    \end{pmatrix}\,,
\]
with $u=(u_{ij})$ and $v=(v_{ij})$ in ${\rm U}(2)$, and $\det u=\det v$. This automatically implies that the state $\ket\Phi=U\ket{0^N}$ is an eigenstate of the operator $(\sigma^z)^{\otimes N}$. Hence, the corresponding state $\ket{\Psi}$ (via Eqs.~\eqref{eq:fermionic} and~\eqref{eq:qubits}) can be written in the form of Eq.~\eqref{eq:FGstateExp} and is thus a fermionic {\em Gaussian} pure state. In particular, its reduced state in a connected subsystem is a fermionic {\em Gaussian} (mixed) state whose correlation matrix can be computed efficiently from that of $\ket\Phi$~\cite{Eisler_2015}.

Note that partial transposition in the fermionic case can be defined in different, in general non-equivalent ways~\cite{Eisler_2015,shapourian2017,murciano2021quench,murciano2022negativity}. Let us write $\rho={\rm tr}_C\ket{\Psi}\bra\Psi$ and $\rho'={\rm tr}_C\ket{\Phi}\bra\Phi$, where $\ket\Psi$ and $\ket\Phi$ are related via Eqs.~\eqref{eq:fermionic} and~\eqref{eq:qubits} and $\ket{\Phi}=U\ket{0^N}$ is the output of a MG circuit $U$. For simplicity, in what follows we will assume that subsystems $A$, $B$ and $C$ are connected and also that subsystems $A$ and $B$ are adjacent~\footnote{In all other cases, one can bring the systems in this order by applying the corresponding fermionic SWAP gates, defined by $\ket{ab}\mapsto(-1)^{ab}\ket{ba}$, to the state $\ket{\Phi}$.}. Then, the definition for the PT operator which we consider here~\cite{Eisler_2015} has the property~\cite{TD02} that the PT moments of $\rho$ coincide with the PT moments of $\rho'$.

\subsubsection{Sampling fermionic Gaussian states}

In what follows we denote by $G_0$ the correlation matrix of the fermionic state $\ket{\Omega}$ representing the vacuum (associated with the state $\ket{0^N}$).
Then, the correlation matrix of $\ket{\Psi}$  corresponding to the state $\ket{\Phi}=U\ket{0^N}$ (see Eqs.~\eqref{eq:fermionic} and~\eqref{eq:qubits}), where $U$ denotes a MG circuit is given by $G=R\cdot G_0\cdot R^T$. Here, $R\in{\rm SO}(2N)$ is related to the MG circuit via the equation ~\cite{TD02,josza2008mg}
\begin{equation}\label{eq:UR}
    U^\dagger c_i U =\sum_j R_{ij} c_j\,.
\end{equation}
This can be easily verified using Eq.~\eqref{eq:UR}, which implies that
\begin{multline*}
\langle c_ic_j\rangle_{U\ket{0^N}}=\langle U^\dagger c_i c_jU\rangle_{\ket{0^N}}=\langle U^\dagger c_i UU^\dagger  c_jU\rangle_{\ket{0^N}}\\
=\sum_{kl}R_{ik}R_{jl}\langle c_k c_l\rangle_{\ket{0^N}}=\sum_{kl}R_{ik}R_{jl}(G_0)_{kl}\,.
\end{multline*}

Note that the relation between a MG circuit $U$ and $R\in{\rm SO}(2N)$ is one-to-one. Let us mention here that sampling uniformly-random special orthogonal matrices $R\in{\rm SO}(2N)$ is equivalent to sample from particularly structured MG circuits~\cite{HNRW22} with ${\rm O}(N^2)$ number of MGs. 

The ($2N_{AB}\times 2N_{AB}$) correlation matrix of the reduced state in the connected subsystem $AB$ can be obtained from $G$ by deleting the rows and columns with indices that correspond to the modes in $C$, the complement of $AB$ in $1,2,\ldots,N$. Therefore, the fermionic Gaussian state in $AB$ can be expressed as
\begin{equation}
    \rho=\frac{1}{Z}\exp\bigg(\frac12\sum_{j,k\in AB} [\tanh^{-1}G]_{jk}c_jc_k\bigg)\,,
\end{equation}
where $Z$ is a normalization factor such that $\tr{\rho_{AB}}=1$. Let us denote by $G'$ the correlation matrix of the previous state. 

As explained in Appendix~\ref{app:freefermions} (see, e.g., Eqs.~\eqref{eq:FFPTtransform} and~Eq.~\eqref{eq:FFPTmoments}) the matrix $G'$, that is efficiently computable allows one to compute the desired PT moments.

Summarizing, the procedure to compute the required PT moments for the ensemble of fermionic Gaussian states is the following. First, one calculates the correlation matrix $G_0$ corresponding to the vacuum (associated to the state $\ket{0^N}$ in the qubits picture). Second, a ($2N\times 2N$) special orthogonal matrix $R$ is sampled uniformly random according to the unique invariant measure (Haar) of ${\rm SO}(2N)$. Third, the correlation matrix $G=R \cdot G_0 \cdot R^{T}$ is constructed. Fourth, the correlation matrix $G'$, corresponding to the reduced state, is obtained from $G$ by deleting the rows and columns with indices that correspond to the modes in $C$. Finally, one uses the formulas of Appendix~\ref{app:freefermions} (see, e.g., Eqs.~\eqref{eq:FFPTtransform} and Eq.~\eqref{eq:FFPTmoments}) to compute $\tilde{r}_2$.

 In Fig.~\ref{fig:fermions} we show the phase diagram  of $\tilde r_2$ as a function of $N_C$ and $N_A$, with $N_{AB}=32$ averaging over $200$ repetitions in panel a). We observe  qualitative differences with respect to Haar random states.
(i) First, we notice the presence of a region with large $\tilde{r}_2\gg 1$ for $N_C/N \ll 1$.
(ii) Second, when $N_C \gg N_{AB}$ we observe a large region with $\tilde{r}_2\ll 1$. 
Interestingly, $\tilde r_2$ does not converge to a fixed value when $N$ increases (keeping the ratios $N_{A,B,C}/N$ fixed).
This is shown in panel b) for $N_A=N_B$, using different values of $N_C/N=0,1/33,1/17,1/9,1/5,1/3,1/2,2/3$. 
For $N_C\le N/33$, $\tilde r_2 $ increases exponentially with system size. Instead for $N_C> N/33$, we observe that $\tilde r_2$ exponentially approaches $0$.

\begin{figure}
    \centering
    \includegraphics[width = 0.99\linewidth]{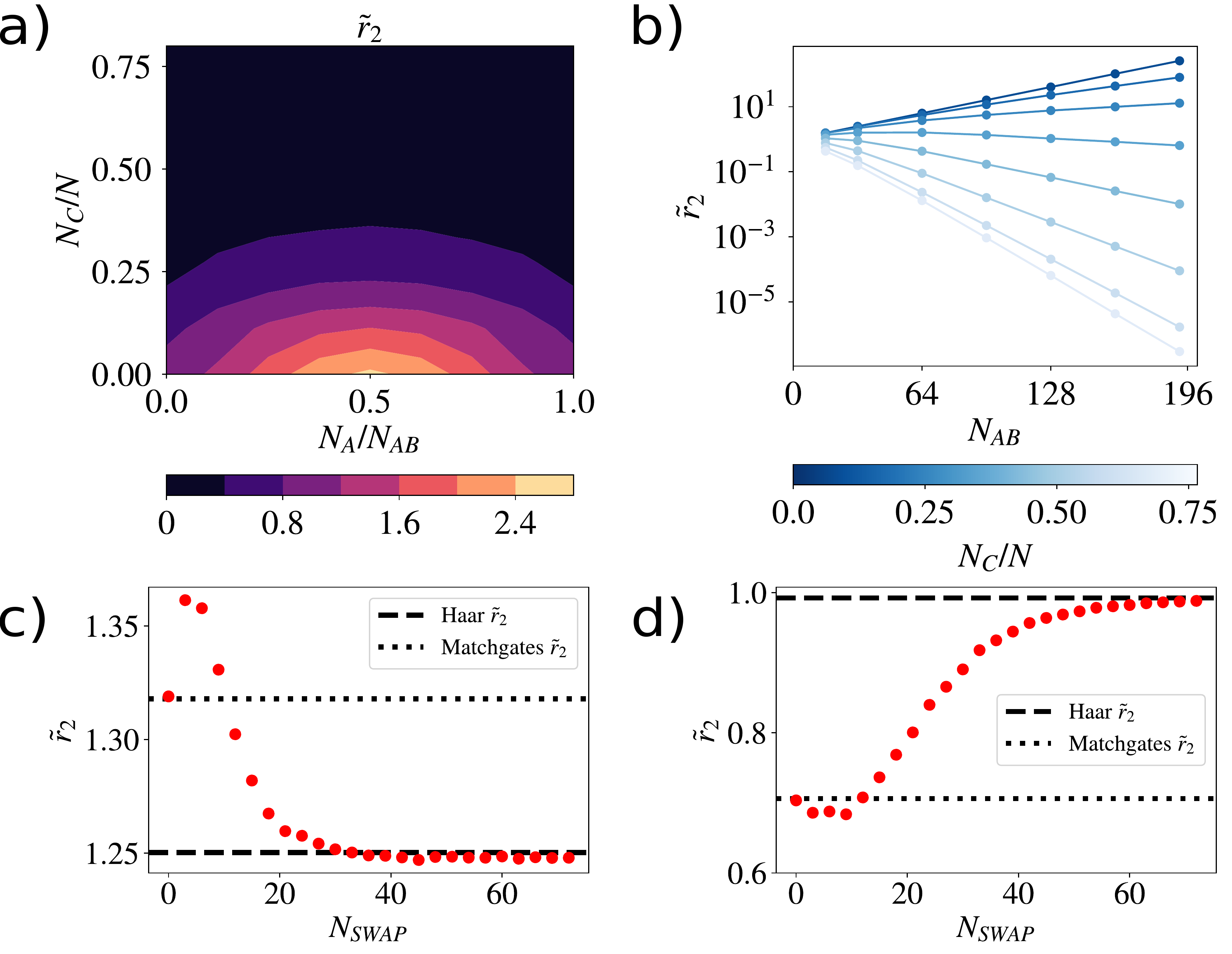}
    \caption{
    a) Phase diagram of the fermionic Gaussian states ensemble as a function of $N_A$ and $N_C$, $N_{AB}=32$ ($N=N_A+N_B+N_C$).
    b) $\tilde r_2$ for fermionic Gaussian states, with $N_A=N_B$, as a function of $N_{AB}$ for different values of $N_C/N$. These results are averaged over 200 random states.
    c) $\tilde r_2$ for a doped MG circuit with $N_A=5, N_B=5, N_C=0$ (corresponding to the ME phase for Haar random states)
    d) Same as c), but for $N_A=4, N_B=2, N_C=14$ (in the PPT phase for Haar random states). 
    Both in c), d), $\tilde r_2$ reveals the transition from Gaussian to Haar random states.
    }\label{fig:fermions}
\end{figure} 
\subsubsection{From Gaussian  to arbitrary states}\label{sec:doping}
Let us consider now nearest-neighbor MGs circuits that are doped with SWAP gates, which make a MG computation universal ~\cite{josza2008mg, hebenstreit2019magic, hebenstreit2020mg}.
By sampling numerically randomly states generated by such (MGs$+$SWAP) circuits, we investigate the transition from  Gaussian fermionic states to random states, as the number $N_{\mathrm{SWAP}}$ of SWAP gates increases. We consider quantum circuits $U$ composed of $3N$ layers each of it consists in the parallel application of $N/2$ (even layer) or $N/2-1$ (odd layer, respectively) nearest-neighbor random two-qubit gates. Among these $3N^2/4+3N(N/2-1)/2$ gates, $N_{\mathrm{SWAP}}$ of them are chosen randomly as SWAP gates, the rest are sampled as random MGs. Therefore the probability to apply a SWAP gate instead of an MG is approximately  $p_{\mathrm{SWAP}} \approx 2N_{\mathrm{SWAP}}/(3N^2)$.

The results for $\tilde r_2$ are shown in Fig.~\ref{fig:fermions} for partitions sizes $N_A,N_B,N_C$ belonging to the ME phases [panel c)], and the PPT phase [panel d)], respectively. In both cases, we observe that for $p_{\mathrm{SWAP}} = 0$ we recover the results of the previous subsection, as we sample approximately random Gaussian states with order $N^2$ MGs. Note that the sampling described in the previous subsection was equivalent to sample from the particularly structured circuits of Ref.~\cite{HNRW22}, where the number of MGs in each circuit was also ${\rm O}(N^2)$. 
As the number of SWAP gates, $N_{\mathrm{SWAP}}$ increases, $\tilde r_2$ converges to the value obtained by Eq.~\eqref{eq:RMTexpression}, indicating the generation of approximate Haar random states. 

\section{Measuring $r_2$ in experiments}\label{sec:measurements}
In this section we address the problem of measuring the ratio $r_2$ in experiments. Being a non-linear functional of the density matrix, $r_2$ cannot be `directly' measured, i.e. as the expectation value of an Hermitian operator. However, one can use approaches based on randomized measurements or physical copies, as we explain below.

For these two approaches, 
an important aspect to have in mind is that, in order to faithfully estimate a ratio of PT moments such as $r_2$, each PT moment must be estimated with a small relative error $\Delta p_n\ll p_n$.

Since $p_n$ is typically an exponentially small number, the determination of $r_2$ via measuring $p_n$ and taking the ratio of these quantities requires a very large number of measurements. 
However, the key features of $r_2$ are already visible for moderate system sizes $N\sim 8$, as shown in the various  numerical examples presented here and in Appendix~\ref{app:measuring}.
 
 \subsection{Randomized measurements}

The idea of randomized measurements consists of using statistical estimators of PT moments based on projective measurements that are performed after random unitary operations~\cite{huang2020shadows,zhou2020single,elben2020mixed,rath2021quantum}. The measurement protocol resembles the one of  quantum state tomography. However, a full quantum state tomography with accuracy $\epsilon$ on the matrix elements of $\rho$ requires at least $N_{\mathrm{tot}}\sim4^{N_{AB}}/\epsilon^2$ measurements \cite{haah2017sample}.
Randomized measurements estimation methods allow us to estimate PT moments with small error $\epsilon$,  it only requires  $N_{\mathrm{tot}}\sim \beta 2^{N_{AB}}/\epsilon^2$, with $\beta$ a prefactor that is state-dependent, and typically decreases with $N$ ~\cite{huang2020shadows,zhou2020single,elben2020mixed,rath2021quantum} (e.g,  $\beta=4p_2$ for  estimating $p_2$).
Due to this `friendly' exponential scaling, the PT moments  $p_2$, $p_3$ have been recently  measured experimentally for systems of up to $7$ qubits~\cite{elben2020mixed} 
(see also Ref.~\cite{rath2022entanglement} for a measurement of a fourth
order polynomial of the density matrix).
In App.~\ref{app:measuring}, we present for completeness a numerical study of statistical errors related to the estimation of $\tilde{r}_2$ with randomized measurements. We find that $\tilde{r_2}$ can be faithfully estimated for $N=8$ for the three entanglement phases with a number of measurements that is compatible with current experimental possibilities.

\subsection{Protocols with multiple copies}
Protocols based on performing measurements on multiple physical copies also allow us access to R\'enyi entropies~\cite{Alves2004,Daley2012,Islam2015,bluvstein2021quantum}, and can be adapted to measure PT moments~\cite{graymachine2018}. The idea is to rewrite PT moments as an expectation value of a permutation operator on the extended state $\rho^{\otimes n}$. While implementing with high-fidelity such a collective measurement on multiple copies can be seen as demanding from a technical point of view, the advantage compared to randomized measurements protools is that the required number of measurements simply scales as \mbox{$N_{\mathrm{tot}}\sim (1-p_n)/\epsilon^2\approx 1/\epsilon^2$} \cite{graymachine2018}.

\section{Entanglement detection via partial transpose moments}\label{sec:detection}

In this section, we mainly consider $r_2$ evaluated for a  single state $\rho$. We will show that the inequality $r_2>1$ detects a special class of entangled states. As this condition can be seen as a sufficient condition for entanglement based on PT moments, we then compare it to another such condition, which involves only the second and the third moment, namely the $p_3$-PPT condition introduced in Ref.~\cite{elben2020mixed}. Furthermore, we introduce the $p_3$ negativity and study this quantity in the context of Haar random states.

\subsection{Detecting entanglement via $r_2$}

As shown in Refs.~\cite{elben2020mixed,neven2021symmetry,yu2021optimal}, PT moments are well suited to detect entanglement and a complete set of inequalities involving PT moments can be derived which are satisfied if and only if the state has a positive partial transpose. Stated differently, any state which violates at least one of the inequalities is necessarily NPT and therefore entangled. Here, we use this insight to show that $r_2$ evaluated on a single state detects entanglement. To stress that we consider here single states, we use the notation $r_2(\rho)$ in the following. Let us now show the following simple  observation 

\begin{Observation} Any bipartite state $\rho$ with  $r_2(\rho)>1$ is entangled. 
\end{Observation}

Despite the fact that this observation is a consequence of the subsequent observation, we present here a proof of it, as it illustrates a connection between the negativity spectrum and the condition $r_2(\rho)>1$.

\begin{proof} 

We denote by $\lambda_i$ the eigenvalues of the partial transpose of $\rho$ and define 
\begin{eqnarray}
\alpha_{n} &=& \frac{1}{2}
\sum_{i,j} (\lambda_i\lambda_j)^{n-1}
(\lambda_i-\lambda_j)^2
(\lambda_i+\lambda_j).
\end{eqnarray}
Using the fact that $p_n=\sum_i \lambda_i^n$, we  have 
\begin{eqnarray}
\alpha_{n} &=& 
p_{n+2}p_{n-1}(1-r_n).
\end{eqnarray}
For any separable $\rho$ 
we have $\lambda_i\geq 0$ and therefore $\alpha_n\ge 0$. Hence, 
$\alpha_n<0$ implies that $A$ and $B$ are entangled.
In particular, for $n=2$, and using the fact that \mbox{$p_4p_1=p_4>0$}, we obtain
$r_2>1$ implies that $A$ and $B$ are entangled. 
\end{proof}

In the situation of Haar random states, we see that the value $r_2$ fluctuating around $3/2$ in the entanglement saturation phase is an evidence of mixed-state entanglement. However, in the maximally entangled phase we have $r_2$ of order $1$. 
Clearly, the condition above is not necessary for entanglement. In fact, as we will show next, the $p_3$--PPT condition, i.e. $p_3>p_2^2$, is strictly stronger than the condition $r_2>1$, as stated in the following observation. 

\begin{Observation} For any bipartite state $\rho$ with  $r_2(\rho)>1$ it holds that the entanglement contained in the state is detected by the $p_3$--PPT condition. 
\end{Observation}
\begin{proof} 
It is easy to show (see Lemma 1 of Ref.~\cite{yu2021optimal}) that for any state $\rho$ it holds that 
\begin{equation}
    p_2p_4\ge p_3^2\label{eq:Hamburger},
\end{equation}

Using this Lemma, we will show now that if $\rho$ satisfies the $p_3$-PPT condition, i.e. if  $p_3\ge p_2^2$ then $r_2\le 1$. Multiplying the left and right hand side of these two inequalities respectively and dividing by the strictly positive number $p_2$,  we obtain 

\begin{equation}
    p_4p_3\ge p_3^2p_2.
\end{equation}

Due to the prerequisite $p_3\ge p_2^2$, we have that 
$p_3>0$. Hence, after dividing the inequality above by $p_3p_4>0$, we obtain 
$r_2=p_2p_3/p_4\le 1$. This shows that if $r_2>1$, then $p_3<p_2^2$. 
\end{proof}

\subsection{Introducing the $p_3$-negativity}
\label{sec:p3neg}

Finally we investigate here to which extent the $p_3$-PPT condition can be used to detect entanglement for random states. To this end we  find it instructive to introduce the `$p_3-$negativity' 
\begin{equation}
    \mathcal{E}_3(\rho)=\frac{1}{2}\log_2(p_2^2/p_3).
\end{equation}
Note that the $p_3$-PPT condition is equivalent to the condition $\mathcal{E}_3(\rho)<0$. 
In addition, for stabilizer states (see Sec.~\ref{sec:stabilizers}), we showed that  $\mathcal{E}_3=\mathcal{E}$. For random states we also define the quantity
$\tilde{\mathcal{E}}_3=\frac{1}{2}\log_2(\mathbb{E}[p_2]^2/\mathbb{E}[p_3])$ obtained after averaging the PT moments, and which can be thus calculated analytically.

As shown in App.~\ref{app:p3neg}, for such random states 
the value of $\tilde{\mathcal{E}}_3$ closely resembles the one of the average negativity  $\mathbb{E}[\mathcal{E}](\rho)$. In particular, while $r_2$ does not differentiate between the PPT phase and the maximally entangled phase ($r_2=1$ in both phases), we have $\tilde{\mathcal{E}}_3(\rho)\approx 0$ in the PPT phase, and $\tilde{\mathcal{E}}_3(\rho)\sim  \min(N_A,N_B)$ in the maximally entangled phase. Thus $\tilde{\mathcal{E}}_3$ can be used to distinguish these two phases. 

As a final remark, for all the random induced mixed states that we have considered, cf details on the numerical simulations in App.~\ref{app:p3neg}, we have observed that the following inequality holds $\mathcal{E}_3(\rho)\le \mathcal{E}(\rho)$. The question of whether the $p_3$-negativity can be proven to be a lower bound to the negativity for any quantum state is left for further work.

\section{Conclusion}
\label{sec:conclusion}
The ratio $r_2$ (and $\tilde{r}_2$) provides a tool to study the entanglement of mixed states, from only the first four moment of the partial transpose. It can be computed numerically and for small system sizes measured experimentally to probe the entanglement phase diagram of random states~\cite{shapourian2021entanglement}, and identify sharp differences compared to Clifford, MPS, Gaussian fermionic states. The value of $r_2$ reflects in particular universal  properties of mixed-state entanglement, in relation to the negativity spectrum.

These results raise interesting prospects regarding the dynamics of quantum circuits, where entanglement grows as a consequence of unitary time evolution, but is also affected by decoherence and or measurements~\cite{fisher2022random}. In this context, it will be in particular important to understand how PT moments reveal the emergence of Haar random states in random quantum circuits, in comparison e.g., with random Clifford  circuits.

Another interesting outlook for our work could be to discover other types of dimensionless ratios, which can tell us about entanglement beyond the PPT condition, for instance in relation to the realignment criterion~\cite{rudolph2000separability,ChenCCNR2003,rudolph2005further}.

\section{Acknowledgements}
We thank A. Rath, C. Lancien, R. Kueng for useful discussions.
Work in Grenoble is funded by the French National Resarch Agency via the JCJC project QRand (ANR-20-CE47-0005), and via the France 2030 programs EPIQ (ANR-22-PETQ-0007), and QUBITAF (ANR-22-PETQ-0004).
B.V., P.Z., and M.V. acknowledge funding from the Austrian Science Foundation (FWF, P 32597 N). 
J.C. and B.K. are grateful for the support
of the Austrian Science Fund (FWF): stand alone
project P32273-N27 and the SFB BeyondC F 7107-N38. The work of V.V. was partly
supported by the ERC under grant number 758329
(AGEnTh), and by the MIUR Programme FARE (MEPH).
\bibliography{biblioRandom.bib}

\begin{thebibliography}{94}%
\makeatletter
\providecommand \@ifxundefined [1]{%
 \@ifx{#1\undefined}
}%
\providecommand \@ifnum [1]{%
 \ifnum #1\expandafter \@firstoftwo
 \else \expandafter \@secondoftwo
 \fi
}%
\providecommand \@ifx [1]{%
 \ifx #1\expandafter \@firstoftwo
 \else \expandafter \@secondoftwo
 \fi
}%
\providecommand \natexlab [1]{#1}%
\providecommand \enquote  [1]{``#1''}%
\providecommand \bibnamefont  [1]{#1}%
\providecommand \bibfnamefont [1]{#1}%
\providecommand \citenamefont [1]{#1}%
\providecommand \href@noop [0]{\@secondoftwo}%
\providecommand \href [0]{\begingroup \@sanitize@url \@href}%
\providecommand \@href[1]{\@@startlink{#1}\@@href}%
\providecommand \@@href[1]{\endgroup#1\@@endlink}%
\providecommand \@sanitize@url [0]{\catcode `\\12\catcode `\$12\catcode
  `\&12\catcode `\#12\catcode `\^12\catcode `\_12\catcode `\%12\relax}%
\providecommand \@@startlink[1]{}%
\providecommand \@@endlink[0]{}%
\providecommand \url  [0]{\begingroup\@sanitize@url \@url }%
\providecommand \@url [1]{\endgroup\@href {#1}{\urlprefix }}%
\providecommand \urlprefix  [0]{URL }%
\providecommand \Eprint [0]{\href }%
\providecommand \doibase [0]{https://doi.org/}%
\providecommand \selectlanguage [0]{\@gobble}%
\providecommand \bibinfo  [0]{\@secondoftwo}%
\providecommand \bibfield  [0]{\@secondoftwo}%
\providecommand \translation [1]{[#1]}%
\providecommand \BibitemOpen [0]{}%
\providecommand \bibitemStop [0]{}%
\providecommand \bibitemNoStop [0]{.\EOS\space}%
\providecommand \EOS [0]{\spacefactor3000\relax}%
\providecommand \BibitemShut  [1]{\csname bibitem#1\endcsname}%
\let\auto@bib@innerbib\@empty
\bibitem [{\citenamefont {Altman}\ \emph {et~al.}(2021)\citenamefont {Altman},
  \citenamefont {Brown}, \citenamefont {Carleo}, \citenamefont {Carr},
  \citenamefont {Demler}, \citenamefont {Chin}, \citenamefont {DeMarco},
  \citenamefont {Economou}, \citenamefont {Eriksson}, \citenamefont {Fu},
  \citenamefont {Greiner}, \citenamefont {Hazzard}, \citenamefont {Hulet},
  \citenamefont {Koll\'ar}, \citenamefont {Lev}, \citenamefont {Lukin},
  \citenamefont {Ma}, \citenamefont {Mi}, \citenamefont {Misra}, \citenamefont
  {Monroe}, \citenamefont {Murch}, \citenamefont {Nazario}, \citenamefont {Ni},
  \citenamefont {Potter}, \citenamefont {Roushan}, \citenamefont {Saffman},
  \citenamefont {Schleier-Smith}, \citenamefont {Siddiqi}, \citenamefont
  {Simmonds}, \citenamefont {Singh}, \citenamefont {Spielman}, \citenamefont
  {Temme}, \citenamefont {Weiss}, \citenamefont {Vu\ifmmode \check{c}\else
  \v{c}\fi{}kovi\ifmmode~\acute{c}\else \'{c}\fi{}}, \citenamefont
  {Vuleti\ifmmode~\acute{c}\else \'{c}\fi{}}, \citenamefont {Ye},\ and\
  \citenamefont {Zwierlein}}]{Altman2021qsim}%
  \BibitemOpen
  \bibfield  {author} {\bibinfo {author} {\bibfnamefont {E.}~\bibnamefont
  {Altman}}, \bibinfo {author} {\bibfnamefont {K.~R.}\ \bibnamefont {Brown}},
  \bibinfo {author} {\bibfnamefont {G.}~\bibnamefont {Carleo}}, \bibinfo
  {author} {\bibfnamefont {L.~D.}\ \bibnamefont {Carr}}, \bibinfo {author}
  {\bibfnamefont {E.}~\bibnamefont {Demler}}, \bibinfo {author} {\bibfnamefont
  {C.}~\bibnamefont {Chin}}, \bibinfo {author} {\bibfnamefont {B.}~\bibnamefont
  {DeMarco}}, \bibinfo {author} {\bibfnamefont {S.~E.}\ \bibnamefont
  {Economou}}, \bibinfo {author} {\bibfnamefont {M.~A.}\ \bibnamefont
  {Eriksson}}, \bibinfo {author} {\bibfnamefont {K.-M.~C.}\ \bibnamefont {Fu}},
  \bibinfo {author} {\bibfnamefont {M.}~\bibnamefont {Greiner}}, \bibinfo
  {author} {\bibfnamefont {K.~R.}\ \bibnamefont {Hazzard}}, \bibinfo {author}
  {\bibfnamefont {R.~G.}\ \bibnamefont {Hulet}}, \bibinfo {author}
  {\bibfnamefont {A.~J.}\ \bibnamefont {Koll\'ar}}, \bibinfo {author}
  {\bibfnamefont {B.~L.}\ \bibnamefont {Lev}}, \bibinfo {author} {\bibfnamefont
  {M.~D.}\ \bibnamefont {Lukin}}, \bibinfo {author} {\bibfnamefont
  {R.}~\bibnamefont {Ma}}, \bibinfo {author} {\bibfnamefont {X.}~\bibnamefont
  {Mi}}, \bibinfo {author} {\bibfnamefont {S.}~\bibnamefont {Misra}}, \bibinfo
  {author} {\bibfnamefont {C.}~\bibnamefont {Monroe}}, \bibinfo {author}
  {\bibfnamefont {K.}~\bibnamefont {Murch}}, \bibinfo {author} {\bibfnamefont
  {Z.}~\bibnamefont {Nazario}}, \bibinfo {author} {\bibfnamefont {K.-K.}\
  \bibnamefont {Ni}}, \bibinfo {author} {\bibfnamefont {A.~C.}\ \bibnamefont
  {Potter}}, \bibinfo {author} {\bibfnamefont {P.}~\bibnamefont {Roushan}},
  \bibinfo {author} {\bibfnamefont {M.}~\bibnamefont {Saffman}}, \bibinfo
  {author} {\bibfnamefont {M.}~\bibnamefont {Schleier-Smith}}, \bibinfo
  {author} {\bibfnamefont {I.}~\bibnamefont {Siddiqi}}, \bibinfo {author}
  {\bibfnamefont {R.}~\bibnamefont {Simmonds}}, \bibinfo {author}
  {\bibfnamefont {M.}~\bibnamefont {Singh}}, \bibinfo {author} {\bibfnamefont
  {I.}~\bibnamefont {Spielman}}, \bibinfo {author} {\bibfnamefont
  {K.}~\bibnamefont {Temme}}, \bibinfo {author} {\bibfnamefont {D.~S.}\
  \bibnamefont {Weiss}}, \bibinfo {author} {\bibfnamefont {J.}~\bibnamefont
  {Vu\ifmmode \check{c}\else \v{c}\fi{}kovi\ifmmode~\acute{c}\else
  \'{c}\fi{}}}, \bibinfo {author} {\bibfnamefont {V.}~\bibnamefont
  {Vuleti\ifmmode~\acute{c}\else \'{c}\fi{}}}, \bibinfo {author} {\bibfnamefont
  {J.}~\bibnamefont {Ye}},\ and\ \bibinfo {author} {\bibfnamefont
  {M.}~\bibnamefont {Zwierlein}},\ }\bibfield  {title} {\bibinfo {title}
  {Quantum simulators: Architectures and opportunities},\ }\href
  {https://doi.org/10.1103/PRXQuantum.2.017003} {\bibfield  {journal} {\bibinfo
   {journal} {PRX Quantum}\ }\textbf {\bibinfo {volume} {2}},\ \bibinfo {pages}
  {017003} (\bibinfo {year} {2021})}\BibitemShut {NoStop}%
\bibitem [{\citenamefont {Shapourian}\ \emph {et~al.}(2021)\citenamefont
  {Shapourian}, \citenamefont {Liu}, \citenamefont {Kudler-Flam},\ and\
  \citenamefont {Vishwanath}}]{shapourian2021entanglement}%
  \BibitemOpen
  \bibfield  {author} {\bibinfo {author} {\bibfnamefont {H.}~\bibnamefont
  {Shapourian}}, \bibinfo {author} {\bibfnamefont {S.}~\bibnamefont {Liu}},
  \bibinfo {author} {\bibfnamefont {J.}~\bibnamefont {Kudler-Flam}},\ and\
  \bibinfo {author} {\bibfnamefont {A.}~\bibnamefont {Vishwanath}},\ }\bibfield
   {title} {\bibinfo {title} {Entanglement negativity spectrum of random mixed
  states: A diagrammatic approach},\ }\href
  {https://doi.org/10.1103/PRXQuantum.2.030347} {\bibfield  {journal} {\bibinfo
   {journal} {PRX Quantum}\ }\textbf {\bibinfo {volume} {2}},\ \bibinfo {pages}
  {030347} (\bibinfo {year} {2021})}\BibitemShut {NoStop}%
\bibitem [{\citenamefont {Aubrun}(2010)}]{Au11}%
  \BibitemOpen
  \bibfield  {author} {\bibinfo {author} {\bibfnamefont {G.}~\bibnamefont
  {Aubrun}},\ }\href@noop {} {\bibinfo {title} {Partial transposition of random
  states and non-centered semicircular distributions}} (\bibinfo {year}
  {2010}),\ \Eprint {https://arxiv.org/abs/1011.0275} {arXiv:1011.0275}
  \BibitemShut {NoStop}%
\bibitem [{\citenamefont {Aubrun}\ \emph {et~al.}(2012)\citenamefont {Aubrun},
  \citenamefont {Szarek},\ and\ \citenamefont {Ye}}]{Au12}%
  \BibitemOpen
  \bibfield  {author} {\bibinfo {author} {\bibfnamefont {G.}~\bibnamefont
  {Aubrun}}, \bibinfo {author} {\bibfnamefont {S.~J.}\ \bibnamefont {Szarek}},\
  and\ \bibinfo {author} {\bibfnamefont {D.}~\bibnamefont {Ye}},\ }\bibfield
  {title} {\bibinfo {title} {Phase transitions for random states and a
  semicircle law for the partial transpose},\ }\href
  {https://doi.org/10.1103/PhysRevA.85.030302} {\bibfield  {journal} {\bibinfo
  {journal} {Phys. Rev. A}\ }\textbf {\bibinfo {volume} {85}},\ \bibinfo
  {pages} {030302} (\bibinfo {year} {2012})}\BibitemShut {NoStop}%
\bibitem [{\citenamefont {Aubrun}\ \emph {et~al.}(2013)\citenamefont {Aubrun},
  \citenamefont {Szarek},\ and\ \citenamefont {Ye}}]{Au13}%
  \BibitemOpen
  \bibfield  {author} {\bibinfo {author} {\bibfnamefont {G.}~\bibnamefont
  {Aubrun}}, \bibinfo {author} {\bibfnamefont {S.~J.}\ \bibnamefont {Szarek}},\
  and\ \bibinfo {author} {\bibfnamefont {D.}~\bibnamefont {Ye}},\ }\bibfield
  {title} {\bibinfo {title} {Entanglement thresholds for random induced
  states},\ }\href {https://doi.org/10.1002/cpa.21460} {\bibfield  {journal}
  {\bibinfo  {journal} {Communications on Pure and Applied Mathematics}\
  }\textbf {\bibinfo {volume} {67}},\ \bibinfo {pages} {129} (\bibinfo {year}
  {2013})}\BibitemShut {NoStop}%
\bibitem [{\citenamefont {Fisher}\ \emph {et~al.}(2022)\citenamefont {Fisher},
  \citenamefont {Khemani}, \citenamefont {Nahum},\ and\ \citenamefont
  {Vijay}}]{fisher2022random}%
  \BibitemOpen
  \bibfield  {author} {\bibinfo {author} {\bibfnamefont {M.~P.~A.}\
  \bibnamefont {Fisher}}, \bibinfo {author} {\bibfnamefont {V.}~\bibnamefont
  {Khemani}}, \bibinfo {author} {\bibfnamefont {A.}~\bibnamefont {Nahum}},\
  and\ \bibinfo {author} {\bibfnamefont {S.}~\bibnamefont {Vijay}},\
  }\href@noop {} {\bibinfo {title} {Random quantum circuits}} (\bibinfo {year}
  {2022}),\ \Eprint {https://arxiv.org/abs/2207.14280} {arXiv:2207.14280}
  \BibitemShut {NoStop}%
\bibitem [{\citenamefont {Potter}\ and\ \citenamefont
  {Vasseur}(2022)}]{Potter_2022}%
  \BibitemOpen
  \bibfield  {author} {\bibinfo {author} {\bibfnamefont {A.~C.}\ \bibnamefont
  {Potter}}\ and\ \bibinfo {author} {\bibfnamefont {R.}~\bibnamefont
  {Vasseur}},\ }\bibfield  {title} {\bibinfo {title} {Entanglement dynamics in
  hybrid quantum circuits},\ }in\ \href
  {https://doi.org/10.1007/978-3-031-03998-0_9} {\emph {\bibinfo {booktitle}
  {Quantum Science and Technology}}}\ (\bibinfo  {publisher} {Springer
  International Publishing},\ \bibinfo {year} {2022})\ pp.\ \bibinfo {pages}
  {211--249}\BibitemShut {NoStop}%
\bibitem [{\citenamefont {Peres}(1996)}]{Pe96}%
  \BibitemOpen
  \bibfield  {author} {\bibinfo {author} {\bibfnamefont {A.}~\bibnamefont
  {Peres}},\ }\bibfield  {title} {\bibinfo {title} {Separability criterion for
  density matrices},\ }\href {https://doi.org/10.1103/PhysRevLett.77.1413}
  {\bibfield  {journal} {\bibinfo  {journal} {Phys. Rev. Lett.}\ }\textbf
  {\bibinfo {volume} {77}},\ \bibinfo {pages} {1413} (\bibinfo {year}
  {1996})}\BibitemShut {NoStop}%
\bibitem [{\citenamefont {Horodecki}(2003)}]{Horodecki2003}%
  \BibitemOpen
  \bibfield  {author} {\bibinfo {author} {\bibfnamefont {P.}~\bibnamefont
  {Horodecki}},\ }\bibfield  {title} {\bibinfo {title} {Measuring quantum
  entanglement without prior state reconstruction},\ }\href
  {https://doi.org/10.1103/PhysRevLett.90.167901} {\bibfield  {journal}
  {\bibinfo  {journal} {Phys. Rev. Lett.}\ }\textbf {\bibinfo {volume} {90}},\
  \bibinfo {pages} {167901} (\bibinfo {year} {2003})}\BibitemShut {NoStop}%
\bibitem [{\citenamefont {Calabrese}\ \emph {et~al.}(2012)\citenamefont
  {Calabrese}, \citenamefont {Cardy},\ and\ \citenamefont
  {Tonni}}]{calabresenegativity2012}%
  \BibitemOpen
  \bibfield  {author} {\bibinfo {author} {\bibfnamefont {P.}~\bibnamefont
  {Calabrese}}, \bibinfo {author} {\bibfnamefont {J.}~\bibnamefont {Cardy}},\
  and\ \bibinfo {author} {\bibfnamefont {E.}~\bibnamefont {Tonni}},\ }\bibfield
   {title} {\bibinfo {title} {Entanglement negativity in quantum field
  theory},\ }\href {https://doi.org/10.1103/PhysRevLett.109.130502} {\bibfield
  {journal} {\bibinfo  {journal} {Phys. Rev. Lett.}\ }\textbf {\bibinfo
  {volume} {109}},\ \bibinfo {pages} {130502} (\bibinfo {year}
  {2012})}\BibitemShut {NoStop}%
\bibitem [{\citenamefont {Elben}\ \emph {et~al.}(2020)\citenamefont {Elben},
  \citenamefont {Kueng}, \citenamefont {Huang}, \citenamefont {van Bijnen},
  \citenamefont {Kokail}, \citenamefont {Dalmonte}, \citenamefont {Calabrese},
  \citenamefont {Kraus}, \citenamefont {Preskill}, \citenamefont {Zoller},\
  and\ \citenamefont {Vermersch}}]{elben2020mixed}%
  \BibitemOpen
  \bibfield  {author} {\bibinfo {author} {\bibfnamefont {A.}~\bibnamefont
  {Elben}}, \bibinfo {author} {\bibfnamefont {R.}~\bibnamefont {Kueng}},
  \bibinfo {author} {\bibfnamefont {H.-Y.~R.}\ \bibnamefont {Huang}}, \bibinfo
  {author} {\bibfnamefont {R.}~\bibnamefont {van Bijnen}}, \bibinfo {author}
  {\bibfnamefont {C.}~\bibnamefont {Kokail}}, \bibinfo {author} {\bibfnamefont
  {M.}~\bibnamefont {Dalmonte}}, \bibinfo {author} {\bibfnamefont
  {P.}~\bibnamefont {Calabrese}}, \bibinfo {author} {\bibfnamefont
  {B.}~\bibnamefont {Kraus}}, \bibinfo {author} {\bibfnamefont
  {J.}~\bibnamefont {Preskill}}, \bibinfo {author} {\bibfnamefont
  {P.}~\bibnamefont {Zoller}},\ and\ \bibinfo {author} {\bibfnamefont
  {B.}~\bibnamefont {Vermersch}},\ }\bibfield  {title} {\bibinfo {title}
  {Mixed-state entanglement from local randomized measurements},\ }\href
  {https://doi.org/10.1103/PhysRevLett.125.200501} {\bibfield  {journal}
  {\bibinfo  {journal} {Phys. Rev. Lett.}\ }\textbf {\bibinfo {volume} {125}},\
  \bibinfo {pages} {200501} (\bibinfo {year} {2020})}\BibitemShut {NoStop}%
\bibitem [{\citenamefont {Neven}\ \emph {et~al.}(2021)\citenamefont {Neven},
  \citenamefont {Carrasco}, \citenamefont {Vitale}, \citenamefont {Kokail},
  \citenamefont {Elben}, \citenamefont {Dalmonte}, \citenamefont {Calabrese},
  \citenamefont {Zoller}, \citenamefont {Vermersch}, \citenamefont {Kueng},\
  and\ \citenamefont {Kraus}}]{neven2021symmetry}%
  \BibitemOpen
  \bibfield  {author} {\bibinfo {author} {\bibfnamefont {A.}~\bibnamefont
  {Neven}}, \bibinfo {author} {\bibfnamefont {J.}~\bibnamefont {Carrasco}},
  \bibinfo {author} {\bibfnamefont {V.}~\bibnamefont {Vitale}}, \bibinfo
  {author} {\bibfnamefont {C.}~\bibnamefont {Kokail}}, \bibinfo {author}
  {\bibfnamefont {A.}~\bibnamefont {Elben}}, \bibinfo {author} {\bibfnamefont
  {M.}~\bibnamefont {Dalmonte}}, \bibinfo {author} {\bibfnamefont
  {P.}~\bibnamefont {Calabrese}}, \bibinfo {author} {\bibfnamefont
  {P.}~\bibnamefont {Zoller}}, \bibinfo {author} {\bibfnamefont
  {B.}~\bibnamefont {Vermersch}}, \bibinfo {author} {\bibfnamefont
  {R.}~\bibnamefont {Kueng}},\ and\ \bibinfo {author} {\bibfnamefont
  {B.}~\bibnamefont {Kraus}},\ }\bibfield  {title} {\bibinfo {title}
  {{Symmetry-resolved entanglement detection using partial transpose
  moments}},\ }\href {https://doi.org/10.1038/s41534-021-00487-y} {\bibfield
  {journal} {\bibinfo  {journal} {npj Quantum Information}\ }\textbf {\bibinfo
  {volume} {7}},\ \bibinfo {pages} {152} (\bibinfo {year} {2021})}\BibitemShut
  {NoStop}%
\bibitem [{\citenamefont {Wybo}\ \emph {et~al.}(2020)\citenamefont {Wybo},
  \citenamefont {Knap},\ and\ \citenamefont {Pollmann}}]{wybo2020mbl}%
  \BibitemOpen
  \bibfield  {author} {\bibinfo {author} {\bibfnamefont {E.}~\bibnamefont
  {Wybo}}, \bibinfo {author} {\bibfnamefont {M.}~\bibnamefont {Knap}},\ and\
  \bibinfo {author} {\bibfnamefont {F.}~\bibnamefont {Pollmann}},\ }\bibfield
  {title} {\bibinfo {title} {Entanglement dynamics of a many-body localized
  system coupled to a bath},\ }\bibfield  {journal} {\bibinfo  {journal}
  {Physical Review B}\ }\textbf {\bibinfo {volume} {102}},\ \href
  {https://doi.org/10.1103/physrevb.102.064304} {10.1103/physrevb.102.064304}
  (\bibinfo {year} {2020})\BibitemShut {NoStop}%
\bibitem [{\citenamefont {Wu}\ \emph {et~al.}(2020)\citenamefont {Wu},
  \citenamefont {Lu}, \citenamefont {Chung}, \citenamefont {Kao},\ and\
  \citenamefont {Grover}}]{wu2020montecarlofiniteT}%
  \BibitemOpen
  \bibfield  {author} {\bibinfo {author} {\bibfnamefont {K.-H.}\ \bibnamefont
  {Wu}}, \bibinfo {author} {\bibfnamefont {T.-C.}\ \bibnamefont {Lu}}, \bibinfo
  {author} {\bibfnamefont {C.-M.}\ \bibnamefont {Chung}}, \bibinfo {author}
  {\bibfnamefont {Y.-J.}\ \bibnamefont {Kao}},\ and\ \bibinfo {author}
  {\bibfnamefont {T.}~\bibnamefont {Grover}},\ }\bibfield  {title} {\bibinfo
  {title} {Entanglement renyi negativity across a finite temperature
  transition: A monte carlo study},\ }\href
  {https://doi.org/10.1103/PhysRevLett.125.140603} {\bibfield  {journal}
  {\bibinfo  {journal} {Phys. Rev. Lett.}\ }\textbf {\bibinfo {volume} {125}},\
  \bibinfo {pages} {140603} (\bibinfo {year} {2020})}\BibitemShut {NoStop}%
\bibitem [{\citenamefont {Feldman}\ \emph {et~al.}(2022)\citenamefont
  {Feldman}, \citenamefont {Kshetrimayum}, \citenamefont {Eisert},\ and\
  \citenamefont {Goldstein}}]{feldman2022entanglement}%
  \BibitemOpen
  \bibfield  {author} {\bibinfo {author} {\bibfnamefont {N.}~\bibnamefont
  {Feldman}}, \bibinfo {author} {\bibfnamefont {A.}~\bibnamefont
  {Kshetrimayum}}, \bibinfo {author} {\bibfnamefont {J.}~\bibnamefont
  {Eisert}},\ and\ \bibinfo {author} {\bibfnamefont {M.}~\bibnamefont
  {Goldstein}},\ }\bibfield  {title} {\bibinfo {title} {Entanglement estimation
  in tensor network states via sampling},\ }\href
  {https://doi.org/10.1103/PRXQuantum.3.030312} {\bibfield  {journal} {\bibinfo
   {journal} {PRX Quantum}\ }\textbf {\bibinfo {volume} {3}},\ \bibinfo {pages}
  {030312} (\bibinfo {year} {2022})}\BibitemShut {NoStop}%
\bibitem [{\citenamefont {Zhou}\ \emph {et~al.}(2020)\citenamefont {Zhou},
  \citenamefont {Zeng},\ and\ \citenamefont {Liu}}]{zhou2020single}%
  \BibitemOpen
  \bibfield  {author} {\bibinfo {author} {\bibfnamefont {Y.}~\bibnamefont
  {Zhou}}, \bibinfo {author} {\bibfnamefont {P.}~\bibnamefont {Zeng}},\ and\
  \bibinfo {author} {\bibfnamefont {Z.}~\bibnamefont {Liu}},\ }\bibfield
  {title} {\bibinfo {title} {Single-copies estimation of entanglement
  negativity},\ }\href {https://doi.org/10.1103/PhysRevLett.125.200502}
  {\bibfield  {journal} {\bibinfo  {journal} {Phys. Rev. Lett.}\ }\textbf
  {\bibinfo {volume} {125}},\ \bibinfo {pages} {200502} (\bibinfo {year}
  {2020})}\BibitemShut {NoStop}%
\bibitem [{\citenamefont {Elben}\ \emph {et~al.}(2022)\citenamefont {Elben},
  \citenamefont {Flammia}, \citenamefont {Huang}, \citenamefont {Kueng},
  \citenamefont {Preskill}, \citenamefont {Vermersch},\ and\ \citenamefont
  {Zoller}}]{elben2022therandomized}%
  \BibitemOpen
  \bibfield  {author} {\bibinfo {author} {\bibfnamefont {A.}~\bibnamefont
  {Elben}}, \bibinfo {author} {\bibfnamefont {S.~T.}\ \bibnamefont {Flammia}},
  \bibinfo {author} {\bibfnamefont {H.-Y.}\ \bibnamefont {Huang}}, \bibinfo
  {author} {\bibfnamefont {R.}~\bibnamefont {Kueng}}, \bibinfo {author}
  {\bibfnamefont {J.}~\bibnamefont {Preskill}}, \bibinfo {author}
  {\bibfnamefont {B.}~\bibnamefont {Vermersch}},\ and\ \bibinfo {author}
  {\bibfnamefont {P.}~\bibnamefont {Zoller}},\ }\bibfield  {title} {\bibinfo
  {title} {The randomized measurement toolbox},\ }\href
  {https://doi.org/10.1038/s42254-022-00535-2} {\bibfield  {journal} {\bibinfo
  {journal} {Nature Reviews Physics}\ } (\bibinfo {year} {2022})}\BibitemShut
  {NoStop}%
\bibitem [{\citenamefont {Bravyi}\ \emph {et~al.}(2006)\citenamefont {Bravyi},
  \citenamefont {Fattal},\ and\ \citenamefont {Gottesman}}]{bravyi2006ghz}%
  \BibitemOpen
  \bibfield  {author} {\bibinfo {author} {\bibfnamefont {S.}~\bibnamefont
  {Bravyi}}, \bibinfo {author} {\bibfnamefont {D.}~\bibnamefont {Fattal}},\
  and\ \bibinfo {author} {\bibfnamefont {D.}~\bibnamefont {Gottesman}},\
  }\bibfield  {title} {\bibinfo {title} {Ghz extraction yield for multipartite
  stabilizer states},\ }\href {https://doi.org/10.1063/1.2203431} {\bibfield
  {journal} {\bibinfo  {journal} {Journal of Mathematical Physics}\ }\textbf
  {\bibinfo {volume} {47}},\ \bibinfo {pages} {062106} (\bibinfo {year}
  {2006})}\BibitemShut {NoStop}%
\bibitem [{\citenamefont {Horodecki}\ \emph {et~al.}(1996)\citenamefont
  {Horodecki}, \citenamefont {Horodecki},\ and\ \citenamefont
  {Horodecki}}]{HHH96}%
  \BibitemOpen
  \bibfield  {author} {\bibinfo {author} {\bibfnamefont {M.}~\bibnamefont
  {Horodecki}}, \bibinfo {author} {\bibfnamefont {P.}~\bibnamefont
  {Horodecki}},\ and\ \bibinfo {author} {\bibfnamefont {R.}~\bibnamefont
  {Horodecki}},\ }\bibfield  {title} {\bibinfo {title} {Separability of mixed
  states: necessary and sufficient conditions},\ }\href
  {https://doi.org/https://doi.org/10.1016/S0375-9601(96)00706-2} {\bibfield
  {journal} {\bibinfo  {journal} {Physics Letters A}\ }\textbf {\bibinfo
  {volume} {223}},\ \bibinfo {pages} {1} (\bibinfo {year} {1996})}\BibitemShut
  {NoStop}%
\bibitem [{\citenamefont {Vidal}\ and\ \citenamefont {Werner}(2002)}]{VW02}%
  \BibitemOpen
  \bibfield  {author} {\bibinfo {author} {\bibfnamefont {G.}~\bibnamefont
  {Vidal}}\ and\ \bibinfo {author} {\bibfnamefont {R.~F.}\ \bibnamefont
  {Werner}},\ }\bibfield  {title} {\bibinfo {title} {Computable measure of
  entanglement},\ }\href {https://doi.org/10.1103/PhysRevA.65.032314}
  {\bibfield  {journal} {\bibinfo  {journal} {Phys. Rev. A}\ }\textbf {\bibinfo
  {volume} {65}},\ \bibinfo {pages} {032314} (\bibinfo {year}
  {2002})}\BibitemShut {NoStop}%
\bibitem [{\citenamefont {Plenio}(2005)}]{Pl05}%
  \BibitemOpen
  \bibfield  {author} {\bibinfo {author} {\bibfnamefont {M.~B.}\ \bibnamefont
  {Plenio}},\ }\bibfield  {title} {\bibinfo {title} {Logarithmic negativity: A
  full entanglement monotone that is not convex},\ }\href
  {https://doi.org/10.1103/PhysRevLett.95.090503} {\bibfield  {journal}
  {\bibinfo  {journal} {Phys. Rev. Lett.}\ }\textbf {\bibinfo {volume} {95}},\
  \bibinfo {pages} {090503} (\bibinfo {year} {2005})}\BibitemShut {NoStop}%
\bibitem [{\citenamefont {Yu}\ \emph {et~al.}(2021)\citenamefont {Yu},
  \citenamefont {Imai},\ and\ \citenamefont {G\"uhne}}]{yu2021optimal}%
  \BibitemOpen
  \bibfield  {author} {\bibinfo {author} {\bibfnamefont {X.-D.}\ \bibnamefont
  {Yu}}, \bibinfo {author} {\bibfnamefont {S.}~\bibnamefont {Imai}},\ and\
  \bibinfo {author} {\bibfnamefont {O.}~\bibnamefont {G\"uhne}},\ }\bibfield
  {title} {\bibinfo {title} {Optimal entanglement certification from moments of
  the partial transpose},\ }\href
  {https://doi.org/10.1103/PhysRevLett.127.060504} {\bibfield  {journal}
  {\bibinfo  {journal} {Phys. Rev. Lett.}\ }\textbf {\bibinfo {volume} {127}},\
  \bibinfo {pages} {060504} (\bibinfo {year} {2021})}\BibitemShut {NoStop}%
\bibitem [{\citenamefont {D’Alessio}\ \emph {et~al.}(2016)\citenamefont
  {D’Alessio}, \citenamefont {Kafri}, \citenamefont {Polkovnikov},\ and\
  \citenamefont {Rigol}}]{dalessio2016ETH}%
  \BibitemOpen
  \bibfield  {author} {\bibinfo {author} {\bibfnamefont {L.}~\bibnamefont
  {D’Alessio}}, \bibinfo {author} {\bibfnamefont {Y.}~\bibnamefont {Kafri}},
  \bibinfo {author} {\bibfnamefont {A.}~\bibnamefont {Polkovnikov}},\ and\
  \bibinfo {author} {\bibfnamefont {M.}~\bibnamefont {Rigol}},\ }\bibfield
  {title} {\bibinfo {title} {From quantum chaos and eigenstate thermalization
  to statistical mechanics and thermodynamics},\ }\href
  {https://doi.org/10.1080/00018732.2016.1198134} {\bibfield  {journal}
  {\bibinfo  {journal} {Advances in Physics}\ }\textbf {\bibinfo {volume}
  {65}},\ \bibinfo {pages} {239–362} (\bibinfo {year} {2016})}\BibitemShut
  {NoStop}%
\bibitem [{\citenamefont {Nahum}\ \emph {et~al.}(2017)\citenamefont {Nahum},
  \citenamefont {Ruhman}, \citenamefont {Vijay},\ and\ \citenamefont
  {Haah}}]{nahum2017quantum}%
  \BibitemOpen
  \bibfield  {author} {\bibinfo {author} {\bibfnamefont {A.}~\bibnamefont
  {Nahum}}, \bibinfo {author} {\bibfnamefont {J.}~\bibnamefont {Ruhman}},
  \bibinfo {author} {\bibfnamefont {S.}~\bibnamefont {Vijay}},\ and\ \bibinfo
  {author} {\bibfnamefont {J.}~\bibnamefont {Haah}},\ }\bibfield  {title}
  {\bibinfo {title} {Quantum entanglement growth under random unitary
  dynamics},\ }\href {https://doi.org/10.1103/PhysRevX.7.031016} {\bibfield
  {journal} {\bibinfo  {journal} {Phys. Rev. X}\ }\textbf {\bibinfo {volume}
  {7}},\ \bibinfo {pages} {031016} (\bibinfo {year} {2017})}\BibitemShut
  {NoStop}%
\bibitem [{\citenamefont {Arute}\ \emph {et~al.}(2019)\citenamefont {Arute},
  \citenamefont {Arya}, \citenamefont {Babbush}, \citenamefont {Bacon},
  \citenamefont {Bardin}, \citenamefont {Barends}, \citenamefont {Biswas},
  \citenamefont {Boixo}, \citenamefont {Brandao}, \citenamefont {Buell},
  \citenamefont {Burkett}, \citenamefont {Chen}, \citenamefont {Chen},
  \citenamefont {Chiaro}, \citenamefont {Collins}, \citenamefont {Courtney},
  \citenamefont {Dunsworth}, \citenamefont {Farhi}, \citenamefont {Foxen},
  \citenamefont {Fowler}, \citenamefont {Gidney}, \citenamefont {Giustina},
  \citenamefont {Graff}, \citenamefont {Guerin}, \citenamefont {Habegger},
  \citenamefont {Harrigan}, \citenamefont {Hartmann}, \citenamefont {Ho},
  \citenamefont {Hoffmann}, \citenamefont {Huang}, \citenamefont {Humble},
  \citenamefont {Isakov}, \citenamefont {Jeffrey}, \citenamefont {Jiang},
  \citenamefont {Kafri}, \citenamefont {Kechedzhi}, \citenamefont {Kelly},
  \citenamefont {Klimov}, \citenamefont {Knysh}, \citenamefont {Korotkov},
  \citenamefont {Kostritsa}, \citenamefont {Landhuis}, \citenamefont
  {Lindmark}, \citenamefont {Lucero}, \citenamefont {Lyakh}, \citenamefont
  {Mandrà}, \citenamefont {McClean}, \citenamefont {McEwen}, \citenamefont
  {Megrant}, \citenamefont {Mi}, \citenamefont {Michielsen}, \citenamefont
  {Mohseni}, \citenamefont {Mutus}, \citenamefont {Naaman}, \citenamefont
  {Neeley}, \citenamefont {Neill}, \citenamefont {Niu}, \citenamefont {Ostby},
  \citenamefont {Petukhov}, \citenamefont {Platt}, \citenamefont {Quintana},
  \citenamefont {Rieffel}, \citenamefont {Roushan}, \citenamefont {Rubin},
  \citenamefont {Sank}, \citenamefont {Satzinger}, \citenamefont {Smelyanskiy},
  \citenamefont {Sung}, \citenamefont {Trevithick}, \citenamefont
  {Vainsencher}, \citenamefont {Villalonga}, \citenamefont {White},
  \citenamefont {Yao}, \citenamefont {Yeh}, \citenamefont {Zalcman},
  \citenamefont {Neven},\ and\ \citenamefont {Martinis}}]{arute2019quantum}%
  \BibitemOpen
  \bibfield  {author} {\bibinfo {author} {\bibfnamefont {F.}~\bibnamefont
  {Arute}}, \bibinfo {author} {\bibfnamefont {K.}~\bibnamefont {Arya}},
  \bibinfo {author} {\bibfnamefont {R.}~\bibnamefont {Babbush}}, \bibinfo
  {author} {\bibfnamefont {D.}~\bibnamefont {Bacon}}, \bibinfo {author}
  {\bibfnamefont {J.~C.}\ \bibnamefont {Bardin}}, \bibinfo {author}
  {\bibfnamefont {R.}~\bibnamefont {Barends}}, \bibinfo {author} {\bibfnamefont
  {R.}~\bibnamefont {Biswas}}, \bibinfo {author} {\bibfnamefont
  {S.}~\bibnamefont {Boixo}}, \bibinfo {author} {\bibfnamefont {F.~G. S.~L.}\
  \bibnamefont {Brandao}}, \bibinfo {author} {\bibfnamefont {D.~A.}\
  \bibnamefont {Buell}}, \bibinfo {author} {\bibfnamefont {B.}~\bibnamefont
  {Burkett}}, \bibinfo {author} {\bibfnamefont {Y.}~\bibnamefont {Chen}},
  \bibinfo {author} {\bibfnamefont {Z.}~\bibnamefont {Chen}}, \bibinfo {author}
  {\bibfnamefont {B.}~\bibnamefont {Chiaro}}, \bibinfo {author} {\bibfnamefont
  {R.}~\bibnamefont {Collins}}, \bibinfo {author} {\bibfnamefont
  {W.}~\bibnamefont {Courtney}}, \bibinfo {author} {\bibfnamefont
  {A.}~\bibnamefont {Dunsworth}}, \bibinfo {author} {\bibfnamefont
  {E.}~\bibnamefont {Farhi}}, \bibinfo {author} {\bibfnamefont
  {B.}~\bibnamefont {Foxen}}, \bibinfo {author} {\bibfnamefont
  {A.}~\bibnamefont {Fowler}}, \bibinfo {author} {\bibfnamefont
  {C.}~\bibnamefont {Gidney}}, \bibinfo {author} {\bibfnamefont
  {M.}~\bibnamefont {Giustina}}, \bibinfo {author} {\bibfnamefont
  {R.}~\bibnamefont {Graff}}, \bibinfo {author} {\bibfnamefont
  {K.}~\bibnamefont {Guerin}}, \bibinfo {author} {\bibfnamefont
  {S.}~\bibnamefont {Habegger}}, \bibinfo {author} {\bibfnamefont {M.~P.}\
  \bibnamefont {Harrigan}}, \bibinfo {author} {\bibfnamefont {M.~J.}\
  \bibnamefont {Hartmann}}, \bibinfo {author} {\bibfnamefont {A.}~\bibnamefont
  {Ho}}, \bibinfo {author} {\bibfnamefont {M.}~\bibnamefont {Hoffmann}},
  \bibinfo {author} {\bibfnamefont {T.}~\bibnamefont {Huang}}, \bibinfo
  {author} {\bibfnamefont {T.~S.}\ \bibnamefont {Humble}}, \bibinfo {author}
  {\bibfnamefont {S.~V.}\ \bibnamefont {Isakov}}, \bibinfo {author}
  {\bibfnamefont {E.}~\bibnamefont {Jeffrey}}, \bibinfo {author} {\bibfnamefont
  {Z.}~\bibnamefont {Jiang}}, \bibinfo {author} {\bibfnamefont
  {D.}~\bibnamefont {Kafri}}, \bibinfo {author} {\bibfnamefont
  {K.}~\bibnamefont {Kechedzhi}}, \bibinfo {author} {\bibfnamefont
  {J.}~\bibnamefont {Kelly}}, \bibinfo {author} {\bibfnamefont {P.~V.}\
  \bibnamefont {Klimov}}, \bibinfo {author} {\bibfnamefont {S.}~\bibnamefont
  {Knysh}}, \bibinfo {author} {\bibfnamefont {A.}~\bibnamefont {Korotkov}},
  \bibinfo {author} {\bibfnamefont {F.}~\bibnamefont {Kostritsa}}, \bibinfo
  {author} {\bibfnamefont {D.}~\bibnamefont {Landhuis}}, \bibinfo {author}
  {\bibfnamefont {M.}~\bibnamefont {Lindmark}}, \bibinfo {author}
  {\bibfnamefont {E.}~\bibnamefont {Lucero}}, \bibinfo {author} {\bibfnamefont
  {D.}~\bibnamefont {Lyakh}}, \bibinfo {author} {\bibfnamefont
  {S.}~\bibnamefont {Mandrà}}, \bibinfo {author} {\bibfnamefont {J.~R.}\
  \bibnamefont {McClean}}, \bibinfo {author} {\bibfnamefont {M.}~\bibnamefont
  {McEwen}}, \bibinfo {author} {\bibfnamefont {A.}~\bibnamefont {Megrant}},
  \bibinfo {author} {\bibfnamefont {X.}~\bibnamefont {Mi}}, \bibinfo {author}
  {\bibfnamefont {K.}~\bibnamefont {Michielsen}}, \bibinfo {author}
  {\bibfnamefont {M.}~\bibnamefont {Mohseni}}, \bibinfo {author} {\bibfnamefont
  {J.}~\bibnamefont {Mutus}}, \bibinfo {author} {\bibfnamefont
  {O.}~\bibnamefont {Naaman}}, \bibinfo {author} {\bibfnamefont
  {M.}~\bibnamefont {Neeley}}, \bibinfo {author} {\bibfnamefont
  {C.}~\bibnamefont {Neill}}, \bibinfo {author} {\bibfnamefont {M.~Y.}\
  \bibnamefont {Niu}}, \bibinfo {author} {\bibfnamefont {E.}~\bibnamefont
  {Ostby}}, \bibinfo {author} {\bibfnamefont {A.}~\bibnamefont {Petukhov}},
  \bibinfo {author} {\bibfnamefont {J.~C.}\ \bibnamefont {Platt}}, \bibinfo
  {author} {\bibfnamefont {C.}~\bibnamefont {Quintana}}, \bibinfo {author}
  {\bibfnamefont {E.~G.}\ \bibnamefont {Rieffel}}, \bibinfo {author}
  {\bibfnamefont {P.}~\bibnamefont {Roushan}}, \bibinfo {author} {\bibfnamefont
  {N.~C.}\ \bibnamefont {Rubin}}, \bibinfo {author} {\bibfnamefont
  {D.}~\bibnamefont {Sank}}, \bibinfo {author} {\bibfnamefont {K.~J.}\
  \bibnamefont {Satzinger}}, \bibinfo {author} {\bibfnamefont {V.}~\bibnamefont
  {Smelyanskiy}}, \bibinfo {author} {\bibfnamefont {K.~J.}\ \bibnamefont
  {Sung}}, \bibinfo {author} {\bibfnamefont {M.~D.}\ \bibnamefont
  {Trevithick}}, \bibinfo {author} {\bibfnamefont {A.}~\bibnamefont
  {Vainsencher}}, \bibinfo {author} {\bibfnamefont {B.}~\bibnamefont
  {Villalonga}}, \bibinfo {author} {\bibfnamefont {T.}~\bibnamefont {White}},
  \bibinfo {author} {\bibfnamefont {Z.~J.}\ \bibnamefont {Yao}}, \bibinfo
  {author} {\bibfnamefont {P.}~\bibnamefont {Yeh}}, \bibinfo {author}
  {\bibfnamefont {A.}~\bibnamefont {Zalcman}}, \bibinfo {author} {\bibfnamefont
  {H.}~\bibnamefont {Neven}},\ and\ \bibinfo {author} {\bibfnamefont {J.~M.}\
  \bibnamefont {Martinis}},\ }\bibfield  {title} {\bibinfo {title} {Quantum
  supremacy using a programmable superconducting processor},\ }\href
  {https://doi.org/10.1038/s41586-019-1666-5} {\bibfield  {journal} {\bibinfo
  {journal} {Nature}\ }\textbf {\bibinfo {volume} {574}},\ \bibinfo {pages}
  {505–510} (\bibinfo {year} {2019})}\BibitemShut {NoStop}%
\bibitem [{\citenamefont {Lashkari}\ \emph {et~al.}(2013)\citenamefont
  {Lashkari}, \citenamefont {Stanford}, \citenamefont {Hastings}, \citenamefont
  {Osborne},\ and\ \citenamefont {Hayden}}]{lashkari2013scrambling}%
  \BibitemOpen
  \bibfield  {author} {\bibinfo {author} {\bibfnamefont {N.}~\bibnamefont
  {Lashkari}}, \bibinfo {author} {\bibfnamefont {D.}~\bibnamefont {Stanford}},
  \bibinfo {author} {\bibfnamefont {M.}~\bibnamefont {Hastings}}, \bibinfo
  {author} {\bibfnamefont {T.}~\bibnamefont {Osborne}},\ and\ \bibinfo {author}
  {\bibfnamefont {P.}~\bibnamefont {Hayden}},\ }\bibfield  {title} {\bibinfo
  {title} {Towards the fast scrambling conjecture},\ }\bibfield  {journal}
  {\bibinfo  {journal} {Journal of High Energy Physics}\ }\href
  {https://doi.org/10.1007/jhep04(2013)022} {10.1007/jhep04(2013)022} (\bibinfo
  {year} {2013})\BibitemShut {NoStop}%
\bibitem [{\citenamefont {Hosur}\ \emph {et~al.}(2016)\citenamefont {Hosur},
  \citenamefont {Qi}, \citenamefont {Roberts},\ and\ \citenamefont
  {Yoshida}}]{hosur2016chaoschannels}%
  \BibitemOpen
  \bibfield  {author} {\bibinfo {author} {\bibfnamefont {P.}~\bibnamefont
  {Hosur}}, \bibinfo {author} {\bibfnamefont {X.-L.}\ \bibnamefont {Qi}},
  \bibinfo {author} {\bibfnamefont {D.~A.}\ \bibnamefont {Roberts}},\ and\
  \bibinfo {author} {\bibfnamefont {B.}~\bibnamefont {Yoshida}},\ }\bibfield
  {title} {\bibinfo {title} {Chaos in quantum channels},\ }\bibfield  {journal}
  {\bibinfo  {journal} {Journal of High Energy Physics}\ }\href
  {https://doi.org/10.1007/jhep02(2016)004} {10.1007/jhep02(2016)004} (\bibinfo
  {year} {2016})\BibitemShut {NoStop}%
\bibitem [{\citenamefont {Nahum}\ \emph {et~al.}(2018)\citenamefont {Nahum},
  \citenamefont {Vijay},\ and\ \citenamefont {Haah}}]{nahum2018operator}%
  \BibitemOpen
  \bibfield  {author} {\bibinfo {author} {\bibfnamefont {A.}~\bibnamefont
  {Nahum}}, \bibinfo {author} {\bibfnamefont {S.}~\bibnamefont {Vijay}},\ and\
  \bibinfo {author} {\bibfnamefont {J.}~\bibnamefont {Haah}},\ }\bibfield
  {title} {\bibinfo {title} {Operator spreading in random unitary circuits},\
  }\href {https://doi.org/10.1103/PhysRevX.8.021014} {\bibfield  {journal}
  {\bibinfo  {journal} {Phys. Rev. X}\ }\textbf {\bibinfo {volume} {8}},\
  \bibinfo {pages} {021014} (\bibinfo {year} {2018})}\BibitemShut {NoStop}%
\bibitem [{\citenamefont {Hayden}\ and\ \citenamefont
  {Preskill}(2007)}]{hayden2007BH}%
  \BibitemOpen
  \bibfield  {author} {\bibinfo {author} {\bibfnamefont {P.}~\bibnamefont
  {Hayden}}\ and\ \bibinfo {author} {\bibfnamefont {J.}~\bibnamefont
  {Preskill}},\ }\bibfield  {title} {\bibinfo {title} {Black holes as mirrors:
  quantum information in random subsystems},\ }\href
  {https://doi.org/10.1088/1126-6708/2007/09/120} {\bibfield  {journal}
  {\bibinfo  {journal} {Journal of High Energy Physics}\ }\textbf {\bibinfo
  {volume} {2007}},\ \bibinfo {pages} {120–120} (\bibinfo {year}
  {2007})}\BibitemShut {NoStop}%
\bibitem [{\citenamefont {Penington}\ \emph {et~al.}(2020)\citenamefont
  {Penington}, \citenamefont {Shenker}, \citenamefont {Stanford},\ and\
  \citenamefont {Yang}}]{penington2020replica}%
  \BibitemOpen
  \bibfield  {author} {\bibinfo {author} {\bibfnamefont {G.}~\bibnamefont
  {Penington}}, \bibinfo {author} {\bibfnamefont {S.~H.}\ \bibnamefont
  {Shenker}}, \bibinfo {author} {\bibfnamefont {D.}~\bibnamefont {Stanford}},\
  and\ \bibinfo {author} {\bibfnamefont {Z.}~\bibnamefont {Yang}},\ }\href@noop
  {} {\bibinfo {title} {Replica wormholes and the black hole interior}}
  (\bibinfo {year} {2020}),\ \Eprint {https://arxiv.org/abs/1911.11977}
  {arXiv:1911.11977 [hep-th]} \BibitemShut {NoStop}%
\bibitem [{\citenamefont {Piroli}\ \emph {et~al.}(2020)\citenamefont {Piroli},
  \citenamefont {Sünderhauf},\ and\ \citenamefont {Qi}}]{piroli2020BH}%
  \BibitemOpen
  \bibfield  {author} {\bibinfo {author} {\bibfnamefont {L.}~\bibnamefont
  {Piroli}}, \bibinfo {author} {\bibfnamefont {C.}~\bibnamefont
  {Sünderhauf}},\ and\ \bibinfo {author} {\bibfnamefont {X.-L.}\ \bibnamefont
  {Qi}},\ }\bibfield  {title} {\bibinfo {title} {A random unitary circuit model
  for black hole evaporation},\ }\bibfield  {journal} {\bibinfo  {journal}
  {Journal of High Energy Physics}\ }\href
  {https://doi.org/10.1007/jhep04(2020)063} {10.1007/jhep04(2020)063} (\bibinfo
  {year} {2020})\BibitemShut {NoStop}%
\bibitem [{\citenamefont {Bohigas}\ \emph {et~al.}(1984)\citenamefont
  {Bohigas}, \citenamefont {Giannoni},\ and\ \citenamefont
  {Schmit}}]{bohigas1984chaosuniversality}%
  \BibitemOpen
  \bibfield  {author} {\bibinfo {author} {\bibfnamefont {O.}~\bibnamefont
  {Bohigas}}, \bibinfo {author} {\bibfnamefont {M.~J.}\ \bibnamefont
  {Giannoni}},\ and\ \bibinfo {author} {\bibfnamefont {C.}~\bibnamefont
  {Schmit}},\ }\bibfield  {title} {\bibinfo {title} {Characterization of
  chaotic quantum spectra and universality of level fluctuation laws},\ }\href
  {https://doi.org/10.1103/PhysRevLett.52.1} {\bibfield  {journal} {\bibinfo
  {journal} {Phys. Rev. Lett.}\ }\textbf {\bibinfo {volume} {52}},\ \bibinfo
  {pages} {1} (\bibinfo {year} {1984})}\BibitemShut {NoStop}%
\bibitem [{\citenamefont {Guhr}\ \emph {et~al.}(1998)\citenamefont {Guhr},
  \citenamefont {Müller–Groeling},\ and\ \citenamefont
  {Weidenmüller}}]{GUHR1998rmtchaos}%
  \BibitemOpen
  \bibfield  {author} {\bibinfo {author} {\bibfnamefont {T.}~\bibnamefont
  {Guhr}}, \bibinfo {author} {\bibfnamefont {A.}~\bibnamefont
  {Müller–Groeling}},\ and\ \bibinfo {author} {\bibfnamefont {H.~A.}\
  \bibnamefont {Weidenmüller}},\ }\bibfield  {title} {\bibinfo {title}
  {Random-matrix theories in quantum physics: common concepts},\ }\href
  {https://doi.org/https://doi.org/10.1016/S0370-1573(97)00088-4} {\bibfield
  {journal} {\bibinfo  {journal} {Physics Reports}\ }\textbf {\bibinfo {volume}
  {299}},\ \bibinfo {pages} {189} (\bibinfo {year} {1998})}\BibitemShut
  {NoStop}%
\bibitem [{\citenamefont {Kos}\ \emph {et~al.}(2018)\citenamefont {Kos},
  \citenamefont {Ljubotina},\ and\ \citenamefont {Prosen}}]{kos2018RMTmbchaos}%
  \BibitemOpen
  \bibfield  {author} {\bibinfo {author} {\bibfnamefont {P.}~\bibnamefont
  {Kos}}, \bibinfo {author} {\bibfnamefont {M.}~\bibnamefont {Ljubotina}},\
  and\ \bibinfo {author} {\bibfnamefont {T.~c.~v.}\ \bibnamefont {Prosen}},\
  }\bibfield  {title} {\bibinfo {title} {Many-body quantum chaos: Analytic
  connection to random matrix theory},\ }\href
  {https://doi.org/10.1103/PhysRevX.8.021062} {\bibfield  {journal} {\bibinfo
  {journal} {Phys. Rev. X}\ }\textbf {\bibinfo {volume} {8}},\ \bibinfo {pages}
  {021062} (\bibinfo {year} {2018})}\BibitemShut {NoStop}%
\bibitem [{\citenamefont {Chen}\ and\ \citenamefont
  {Ludwig}(2018)}]{chen2018chaosuniversality}%
  \BibitemOpen
  \bibfield  {author} {\bibinfo {author} {\bibfnamefont {X.}~\bibnamefont
  {Chen}}\ and\ \bibinfo {author} {\bibfnamefont {A.~W.~W.}\ \bibnamefont
  {Ludwig}},\ }\bibfield  {title} {\bibinfo {title} {Universal spectral
  correlations in the chaotic wave function and the development of quantum
  chaos},\ }\href {https://doi.org/10.1103/PhysRevB.98.064309} {\bibfield
  {journal} {\bibinfo  {journal} {Phys. Rev. B}\ }\textbf {\bibinfo {volume}
  {98}},\ \bibinfo {pages} {064309} (\bibinfo {year} {2018})}\BibitemShut
  {NoStop}%
\bibitem [{\citenamefont {Page}(1993)}]{Page1993}%
  \BibitemOpen
  \bibfield  {author} {\bibinfo {author} {\bibfnamefont {D.~N.}\ \bibnamefont
  {Page}},\ }\bibfield  {title} {\bibinfo {title} {Average entropy of a
  subsystem},\ }\href {https://doi.org/10.1103/PhysRevLett.71.1291} {\bibfield
  {journal} {\bibinfo  {journal} {Phys. Rev. Lett.}\ }\textbf {\bibinfo
  {volume} {71}},\ \bibinfo {pages} {1291} (\bibinfo {year}
  {1993})}\BibitemShut {NoStop}%
\bibitem [{\citenamefont {Haug}\ and\ \citenamefont
  {Piroli}(2022)}]{PiroliNonstabilizerness}%
  \BibitemOpen
  \bibfield  {author} {\bibinfo {author} {\bibfnamefont {T.}~\bibnamefont
  {Haug}}\ and\ \bibinfo {author} {\bibfnamefont {L.}~\bibnamefont {Piroli}},\
  }\href {https://arxiv.org/abs/2207.13076} {\bibinfo {title} {Quantifying
  nonstabilizerness of matrix product states}} (\bibinfo {year} {2022}),\
  \Eprint {https://arxiv.org/abs/2207.13076} {arXiv:2207.13076} \BibitemShut
  {NoStop}%
\bibitem [{\citenamefont {Haferkamp}\ \emph {et~al.}(2020)\citenamefont
  {Haferkamp}, \citenamefont {Montealegre-Mora}, \citenamefont {Heinrich},
  \citenamefont {Eisert}, \citenamefont {Gross},\ and\ \citenamefont
  {Roth}}]{haferkamp2020quantum}%
  \BibitemOpen
  \bibfield  {author} {\bibinfo {author} {\bibfnamefont {J.}~\bibnamefont
  {Haferkamp}}, \bibinfo {author} {\bibfnamefont {F.}~\bibnamefont
  {Montealegre-Mora}}, \bibinfo {author} {\bibfnamefont {M.}~\bibnamefont
  {Heinrich}}, \bibinfo {author} {\bibfnamefont {J.}~\bibnamefont {Eisert}},
  \bibinfo {author} {\bibfnamefont {D.}~\bibnamefont {Gross}},\ and\ \bibinfo
  {author} {\bibfnamefont {I.}~\bibnamefont {Roth}},\ }\href@noop {} {\bibinfo
  {title} {Quantum homeopathy works: Efficient unitary designs with a
  system-size independent number of non-clifford gates}} (\bibinfo {year}
  {2020}),\ \Eprint {https://arxiv.org/abs/2002.09524} {arXiv:2002.09524
  [quant-ph]} \BibitemShut {NoStop}%
\bibitem [{\citenamefont {Carteret}(2005)}]{carteret2005noiseless}%
  \BibitemOpen
  \bibfield  {author} {\bibinfo {author} {\bibfnamefont {H.~A.}\ \bibnamefont
  {Carteret}},\ }\bibfield  {title} {\bibinfo {title} {Noiseless quantum
  circuits for the peres separability criterion},\ }\href
  {https://doi.org/10.1103/PhysRevLett.94.040502} {\bibfield  {journal}
  {\bibinfo  {journal} {Phys. Rev. Lett.}\ }\textbf {\bibinfo {volume} {94}},\
  \bibinfo {pages} {040502} (\bibinfo {year} {2005})}\BibitemShut {NoStop}%
\bibitem [{\citenamefont {Gray}\ \emph {et~al.}(2018)\citenamefont {Gray},
  \citenamefont {Banchi}, \citenamefont {Bayat},\ and\ \citenamefont
  {Bose}}]{graymachine2018}%
  \BibitemOpen
  \bibfield  {author} {\bibinfo {author} {\bibfnamefont {J.}~\bibnamefont
  {Gray}}, \bibinfo {author} {\bibfnamefont {L.}~\bibnamefont {Banchi}},
  \bibinfo {author} {\bibfnamefont {A.}~\bibnamefont {Bayat}},\ and\ \bibinfo
  {author} {\bibfnamefont {S.}~\bibnamefont {Bose}},\ }\bibfield  {title}
  {\bibinfo {title} {Machine-learning-assisted many-body entanglement
  measurement},\ }\href {https://doi.org/10.1103/PhysRevLett.121.150503}
  {\bibfield  {journal} {\bibinfo  {journal} {Phys. Rev. Lett.}\ }\textbf
  {\bibinfo {volume} {121}},\ \bibinfo {pages} {150503} (\bibinfo {year}
  {2018})}\BibitemShut {NoStop}%
\bibitem [{\citenamefont {Collins}\ and\ \citenamefont
  {{\'{S}}niady}(2006)}]{Collins2006RMT}%
  \BibitemOpen
  \bibfield  {author} {\bibinfo {author} {\bibfnamefont {B.}~\bibnamefont
  {Collins}}\ and\ \bibinfo {author} {\bibfnamefont {P.}~\bibnamefont
  {{\'{S}}niady}},\ }\bibfield  {title} {\bibinfo {title} {Integration with
  respect to the haar measure on unitary, orthogonal and symplectic group},\
  }\href {https://doi.org/10.1007/s00220-006-1554-3} {\bibfield  {journal}
  {\bibinfo  {journal} {Communications in Mathematical Physics}\ }\textbf
  {\bibinfo {volume} {264}},\ \bibinfo {pages} {773} (\bibinfo {year}
  {2006})}\BibitemShut {NoStop}%
\bibitem [{\citenamefont {Elben}\ \emph {et~al.}(2019)\citenamefont {Elben},
  \citenamefont {Vermersch}, \citenamefont {Roos},\ and\ \citenamefont
  {Zoller}}]{Elben2019correlations}%
  \BibitemOpen
  \bibfield  {author} {\bibinfo {author} {\bibfnamefont {A.}~\bibnamefont
  {Elben}}, \bibinfo {author} {\bibfnamefont {B.}~\bibnamefont {Vermersch}},
  \bibinfo {author} {\bibfnamefont {C.~F.}\ \bibnamefont {Roos}},\ and\
  \bibinfo {author} {\bibfnamefont {P.}~\bibnamefont {Zoller}},\ }\bibfield
  {title} {\bibinfo {title} {Statistical correlations between locally
  randomized measurements: A toolbox for probing entanglement in many-body
  quantum states},\ }\href {https://doi.org/10.1103/PhysRevA.99.052323}
  {\bibfield  {journal} {\bibinfo  {journal} {Phys. Rev. A}\ }\textbf {\bibinfo
  {volume} {99}},\ \bibinfo {pages} {052323} (\bibinfo {year}
  {2019})}\BibitemShut {NoStop}%
\bibitem [{Note1()}]{Note1}%
  \BibitemOpen
  \bibinfo {note} {The normalization condition implies that $k_2=(1-k_1\lambda
  _1)/\lambda _2$. Inserting this expression for $k_2$ into the equation
  $r_2=1$ leads to $k_1\lambda _1(-1+k_1\lambda _1)(\lambda _1-\lambda
  _2)^2(\lambda _1+\lambda _2)=0$. Hence, the only non-trivial solutions to
  this equation and the normalization condition are $\lambda _1=\pm \lambda
  _2=(k_1\pm k_2)^{-1}$. In both cases $\epsilon _i=0,1$ for all
  $i$.}\BibitemShut {Stop}%
\bibitem [{\citenamefont {Gottesman}(1998)}]{gottesman1998heisenberg}%
  \BibitemOpen
  \bibfield  {author} {\bibinfo {author} {\bibfnamefont {D.}~\bibnamefont
  {Gottesman}},\ }\href@noop {} {\bibinfo {title} {The heisenberg
  representation of quantum computers}} (\bibinfo {year} {1998}),\ \Eprint
  {https://arxiv.org/abs/quant-ph/9807006} {arXiv:quant-ph/9807006 [quant-ph]}
  \BibitemShut {NoStop}%
\bibitem [{\citenamefont {Aaronson}\ and\ \citenamefont
  {Gottesman}(2004)}]{gottesman2004simulation}%
  \BibitemOpen
  \bibfield  {author} {\bibinfo {author} {\bibfnamefont {S.}~\bibnamefont
  {Aaronson}}\ and\ \bibinfo {author} {\bibfnamefont {D.}~\bibnamefont
  {Gottesman}},\ }\bibfield  {title} {\bibinfo {title} {Improved simulation of
  stabilizer circuits},\ }\href {https://doi.org/10.1103/PhysRevA.70.052328}
  {\bibfield  {journal} {\bibinfo  {journal} {Phys. Rev. A}\ }\textbf {\bibinfo
  {volume} {70}},\ \bibinfo {pages} {052328} (\bibinfo {year}
  {2004})}\BibitemShut {NoStop}%
\bibitem [{\citenamefont {Leone}\ \emph {et~al.}(2021)\citenamefont {Leone},
  \citenamefont {Oliviero}, \citenamefont {Zhou},\ and\ \citenamefont
  {Hamma}}]{leone2021quantum}%
  \BibitemOpen
  \bibfield  {author} {\bibinfo {author} {\bibfnamefont {L.}~\bibnamefont
  {Leone}}, \bibinfo {author} {\bibfnamefont {S.~F.~E.}\ \bibnamefont
  {Oliviero}}, \bibinfo {author} {\bibfnamefont {Y.}~\bibnamefont {Zhou}},\
  and\ \bibinfo {author} {\bibfnamefont {A.}~\bibnamefont {Hamma}},\ }\bibfield
   {title} {\bibinfo {title} {Quantum chaos is quantum},\ }\href
  {https://doi.org/10.22331/q-2021-05-04-453} {\bibfield  {journal} {\bibinfo
  {journal} {Quantum}\ }\textbf {\bibinfo {volume} {5}},\ \bibinfo {pages}
  {453} (\bibinfo {year} {2021})}\BibitemShut {NoStop}%
\bibitem [{\citenamefont {Leone}\ \emph {et~al.}(2022)\citenamefont {Leone},
  \citenamefont {Oliviero},\ and\ \citenamefont {Hamma}}]{leone2022magic}%
  \BibitemOpen
  \bibfield  {author} {\bibinfo {author} {\bibfnamefont {L.}~\bibnamefont
  {Leone}}, \bibinfo {author} {\bibfnamefont {S.~F.~E.}\ \bibnamefont
  {Oliviero}},\ and\ \bibinfo {author} {\bibfnamefont {A.}~\bibnamefont
  {Hamma}},\ }\bibfield  {title} {\bibinfo {title} {Stabilizer r\'enyi
  entropy},\ }\href {https://doi.org/10.1103/PhysRevLett.128.050402} {\bibfield
   {journal} {\bibinfo  {journal} {Phys. Rev. Lett.}\ }\textbf {\bibinfo
  {volume} {128}},\ \bibinfo {pages} {050402} (\bibinfo {year}
  {2022})}\BibitemShut {NoStop}%
\bibitem [{\citenamefont {Browaeys}\ and\ \citenamefont
  {Lahaye}(2020)}]{Browaeys2020}%
  \BibitemOpen
  \bibfield  {author} {\bibinfo {author} {\bibfnamefont {A.}~\bibnamefont
  {Browaeys}}\ and\ \bibinfo {author} {\bibfnamefont {T.}~\bibnamefont
  {Lahaye}},\ }\bibfield  {title} {\bibinfo {title} {Many-body physics with
  individually controlled rydberg atoms},\ }\href
  {https://doi.org/10.1038/s41567-019-0733-z} {\bibfield  {journal} {\bibinfo
  {journal} {Nature Physics}\ }\textbf {\bibinfo {volume} {16}},\ \bibinfo
  {pages} {132} (\bibinfo {year} {2020})}\BibitemShut {NoStop}%
\bibitem [{\citenamefont {Bernien}\ \emph {et~al.}(2017)\citenamefont
  {Bernien}, \citenamefont {Schwartz}, \citenamefont {Keesling}, \citenamefont
  {Levine}, \citenamefont {Omran}, \citenamefont {Pichler}, \citenamefont
  {Choi}, \citenamefont {Zibrov}, \citenamefont {Endres}, \citenamefont
  {Greiner}, \citenamefont {Vuleti{\'c}},\ and\ \citenamefont
  {Lukin}}]{Bernien2017}%
  \BibitemOpen
  \bibfield  {author} {\bibinfo {author} {\bibfnamefont {H.}~\bibnamefont
  {Bernien}}, \bibinfo {author} {\bibfnamefont {S.}~\bibnamefont {Schwartz}},
  \bibinfo {author} {\bibfnamefont {A.}~\bibnamefont {Keesling}}, \bibinfo
  {author} {\bibfnamefont {H.}~\bibnamefont {Levine}}, \bibinfo {author}
  {\bibfnamefont {A.}~\bibnamefont {Omran}}, \bibinfo {author} {\bibfnamefont
  {H.}~\bibnamefont {Pichler}}, \bibinfo {author} {\bibfnamefont
  {S.}~\bibnamefont {Choi}}, \bibinfo {author} {\bibfnamefont {A.~S.}\
  \bibnamefont {Zibrov}}, \bibinfo {author} {\bibfnamefont {M.}~\bibnamefont
  {Endres}}, \bibinfo {author} {\bibfnamefont {M.}~\bibnamefont {Greiner}},
  \bibinfo {author} {\bibfnamefont {V.}~\bibnamefont {Vuleti{\'c}}},\ and\
  \bibinfo {author} {\bibfnamefont {M.~D.}\ \bibnamefont {Lukin}},\ }\bibfield
  {title} {\bibinfo {title} {Probing many-body dynamics on a 51-atom quantum
  simulator},\ }\href {https://doi.org/10.1038/nature24622} {\bibfield
  {journal} {\bibinfo  {journal} {Nature}\ }\textbf {\bibinfo {volume} {551}},\
  \bibinfo {pages} {579} (\bibinfo {year} {2017})}\BibitemShut {NoStop}%
\bibitem [{\citenamefont {Ebadi}\ \emph {et~al.}(2021)\citenamefont {Ebadi},
  \citenamefont {Wang}, \citenamefont {Levine}, \citenamefont {Keesling},
  \citenamefont {Semeghini}, \citenamefont {Omran}, \citenamefont {Bluvstein},
  \citenamefont {Samajdar}, \citenamefont {Pichler}, \citenamefont {Ho},
  \citenamefont {Choi}, \citenamefont {Sachdev}, \citenamefont {Greiner},
  \citenamefont {Vuleti{\'c}},\ and\ \citenamefont {Lukin}}]{Ebadi2021}%
  \BibitemOpen
  \bibfield  {author} {\bibinfo {author} {\bibfnamefont {S.}~\bibnamefont
  {Ebadi}}, \bibinfo {author} {\bibfnamefont {T.~T.}\ \bibnamefont {Wang}},
  \bibinfo {author} {\bibfnamefont {H.}~\bibnamefont {Levine}}, \bibinfo
  {author} {\bibfnamefont {A.}~\bibnamefont {Keesling}}, \bibinfo {author}
  {\bibfnamefont {G.}~\bibnamefont {Semeghini}}, \bibinfo {author}
  {\bibfnamefont {A.}~\bibnamefont {Omran}}, \bibinfo {author} {\bibfnamefont
  {D.}~\bibnamefont {Bluvstein}}, \bibinfo {author} {\bibfnamefont
  {R.}~\bibnamefont {Samajdar}}, \bibinfo {author} {\bibfnamefont
  {H.}~\bibnamefont {Pichler}}, \bibinfo {author} {\bibfnamefont {W.~W.}\
  \bibnamefont {Ho}}, \bibinfo {author} {\bibfnamefont {S.}~\bibnamefont
  {Choi}}, \bibinfo {author} {\bibfnamefont {S.}~\bibnamefont {Sachdev}},
  \bibinfo {author} {\bibfnamefont {M.}~\bibnamefont {Greiner}}, \bibinfo
  {author} {\bibfnamefont {V.}~\bibnamefont {Vuleti{\'c}}},\ and\ \bibinfo
  {author} {\bibfnamefont {M.~D.}\ \bibnamefont {Lukin}},\ }\bibfield  {title}
  {\bibinfo {title} {Quantum phases of matter on a 256-atom programmable
  quantum simulator},\ }\href {https://doi.org/10.1038/s41586-021-03582-4}
  {\bibfield  {journal} {\bibinfo  {journal} {Nature}\ }\textbf {\bibinfo
  {volume} {595}},\ \bibinfo {pages} {227} (\bibinfo {year}
  {2021})}\BibitemShut {NoStop}%
\bibitem [{\citenamefont {de~L{\'{e}}s{\'{e}}leuc}\ \emph
  {et~al.}(2019)\citenamefont {de~L{\'{e}}s{\'{e}}leuc}, \citenamefont
  {Lienhard}, \citenamefont {Scholl}, \citenamefont {Barredo}, \citenamefont
  {Weber}, \citenamefont {Lang}, \citenamefont {Büchler}, \citenamefont
  {Lahaye},\ and\ \citenamefont {Browaeys}}]{Syl2019}%
  \BibitemOpen
  \bibfield  {author} {\bibinfo {author} {\bibfnamefont {S.}~\bibnamefont
  {de~L{\'{e}}s{\'{e}}leuc}}, \bibinfo {author} {\bibfnamefont
  {V.}~\bibnamefont {Lienhard}}, \bibinfo {author} {\bibfnamefont
  {P.}~\bibnamefont {Scholl}}, \bibinfo {author} {\bibfnamefont
  {D.}~\bibnamefont {Barredo}}, \bibinfo {author} {\bibfnamefont
  {S.}~\bibnamefont {Weber}}, \bibinfo {author} {\bibfnamefont
  {N.}~\bibnamefont {Lang}}, \bibinfo {author} {\bibfnamefont {H.~P.}\
  \bibnamefont {Büchler}}, \bibinfo {author} {\bibfnamefont {T.}~\bibnamefont
  {Lahaye}},\ and\ \bibinfo {author} {\bibfnamefont {A.}~\bibnamefont
  {Browaeys}},\ }\bibfield  {title} {\bibinfo {title} {Observation of a
  symmetry-protected topological phase of interacting bosons with rydberg
  atoms},\ }\href {https://doi.org/10.1126/science.aav9105} {\bibfield
  {journal} {\bibinfo  {journal} {Science}\ }\textbf {\bibinfo {volume}
  {365}},\ \bibinfo {pages} {775} (\bibinfo {year} {2019})}\BibitemShut
  {NoStop}%
\bibitem [{\citenamefont {Semeghini}\ \emph {et~al.}(2021)\citenamefont
  {Semeghini}, \citenamefont {Levine}, \citenamefont {Keesling}, \citenamefont
  {Ebadi}, \citenamefont {Wang}, \citenamefont {Bluvstein}, \citenamefont
  {Verresen}, \citenamefont {Pichler}, \citenamefont {Kalinowski},
  \citenamefont {Samajdar}, \citenamefont {Omran}, \citenamefont {Sachdev},
  \citenamefont {Vishwanath}, \citenamefont {Greiner}, \citenamefont
  {Vuletić},\ and\ \citenamefont {Lukin}}]{Semeghini2021}%
  \BibitemOpen
  \bibfield  {author} {\bibinfo {author} {\bibfnamefont {G.}~\bibnamefont
  {Semeghini}}, \bibinfo {author} {\bibfnamefont {H.}~\bibnamefont {Levine}},
  \bibinfo {author} {\bibfnamefont {A.}~\bibnamefont {Keesling}}, \bibinfo
  {author} {\bibfnamefont {S.}~\bibnamefont {Ebadi}}, \bibinfo {author}
  {\bibfnamefont {T.~T.}\ \bibnamefont {Wang}}, \bibinfo {author}
  {\bibfnamefont {D.}~\bibnamefont {Bluvstein}}, \bibinfo {author}
  {\bibfnamefont {R.}~\bibnamefont {Verresen}}, \bibinfo {author}
  {\bibfnamefont {H.}~\bibnamefont {Pichler}}, \bibinfo {author} {\bibfnamefont
  {M.}~\bibnamefont {Kalinowski}}, \bibinfo {author} {\bibfnamefont
  {R.}~\bibnamefont {Samajdar}}, \bibinfo {author} {\bibfnamefont
  {A.}~\bibnamefont {Omran}}, \bibinfo {author} {\bibfnamefont
  {S.}~\bibnamefont {Sachdev}}, \bibinfo {author} {\bibfnamefont
  {A.}~\bibnamefont {Vishwanath}}, \bibinfo {author} {\bibfnamefont
  {M.}~\bibnamefont {Greiner}}, \bibinfo {author} {\bibfnamefont
  {V.}~\bibnamefont {Vuletić}},\ and\ \bibinfo {author} {\bibfnamefont
  {M.~D.}\ \bibnamefont {Lukin}},\ }\bibfield  {title} {\bibinfo {title}
  {Probing topological spin liquids on a programmable quantum simulator},\
  }\href {https://doi.org/10.1126/science.abi8794} {\bibfield  {journal}
  {\bibinfo  {journal} {Science}\ }\textbf {\bibinfo {volume} {374}},\ \bibinfo
  {pages} {1242–1247} (\bibinfo {year} {2021})}\BibitemShut {NoStop}%
\bibitem [{\citenamefont {Ho}\ \emph {et~al.}(2019)\citenamefont {Ho},
  \citenamefont {Choi}, \citenamefont {Pichler},\ and\ \citenamefont
  {Lukin}}]{Wen2019}%
  \BibitemOpen
  \bibfield  {author} {\bibinfo {author} {\bibfnamefont {W.~W.}\ \bibnamefont
  {Ho}}, \bibinfo {author} {\bibfnamefont {S.}~\bibnamefont {Choi}}, \bibinfo
  {author} {\bibfnamefont {H.}~\bibnamefont {Pichler}},\ and\ \bibinfo {author}
  {\bibfnamefont {M.~D.}\ \bibnamefont {Lukin}},\ }\bibfield  {title} {\bibinfo
  {title} {Periodic orbits, entanglement, and quantum many-body scars in
  constrained models: Matrix product state approach},\ }\href
  {https://doi.org/10.1103/PhysRevLett.122.040603} {\bibfield  {journal}
  {\bibinfo  {journal} {Phys. Rev. Lett.}\ }\textbf {\bibinfo {volume} {122}},\
  \bibinfo {pages} {040603} (\bibinfo {year} {2019})}\BibitemShut {NoStop}%
\bibitem [{\citenamefont {Serbyn}\ \emph {et~al.}(2021)\citenamefont {Serbyn},
  \citenamefont {Abanin},\ and\ \citenamefont {Papi{\'c}}}]{Serbyn2021}%
  \BibitemOpen
  \bibfield  {author} {\bibinfo {author} {\bibfnamefont {M.}~\bibnamefont
  {Serbyn}}, \bibinfo {author} {\bibfnamefont {D.~A.}\ \bibnamefont {Abanin}},\
  and\ \bibinfo {author} {\bibfnamefont {Z.}~\bibnamefont {Papi{\'c}}},\
  }\bibfield  {title} {\bibinfo {title} {Quantum many-body scars and weak
  breaking of ergodicity},\ }\href {https://doi.org/10.1038/s41567-021-01230-2}
  {\bibfield  {journal} {\bibinfo  {journal} {Nature Physics}\ }\textbf
  {\bibinfo {volume} {17}},\ \bibinfo {pages} {675} (\bibinfo {year}
  {2021})}\BibitemShut {NoStop}%
\bibitem [{\citenamefont {Turner}\ \emph {et~al.}(2018)\citenamefont {Turner},
  \citenamefont {Michailidis}, \citenamefont {Abanin}, \citenamefont {Serbyn},\
  and\ \citenamefont {Papi\ifmmode~\acute{c}\else \'{c}\fi{}}}]{TurnerB2018}%
  \BibitemOpen
  \bibfield  {author} {\bibinfo {author} {\bibfnamefont {C.~J.}\ \bibnamefont
  {Turner}}, \bibinfo {author} {\bibfnamefont {A.~A.}\ \bibnamefont
  {Michailidis}}, \bibinfo {author} {\bibfnamefont {D.~A.}\ \bibnamefont
  {Abanin}}, \bibinfo {author} {\bibfnamefont {M.}~\bibnamefont {Serbyn}},\
  and\ \bibinfo {author} {\bibfnamefont {Z.}~\bibnamefont
  {Papi\ifmmode~\acute{c}\else \'{c}\fi{}}},\ }\bibfield  {title} {\bibinfo
  {title} {Quantum scarred eigenstates in a rydberg atom chain: Entanglement,
  breakdown of thermalization, and stability to perturbations},\ }\href
  {https://doi.org/10.1103/PhysRevB.98.155134} {\bibfield  {journal} {\bibinfo
  {journal} {Phys. Rev. B}\ }\textbf {\bibinfo {volume} {98}},\ \bibinfo
  {pages} {155134} (\bibinfo {year} {2018})}\BibitemShut {NoStop}%
\bibitem [{\citenamefont {Lin}\ \emph {et~al.}(2020)\citenamefont {Lin},
  \citenamefont {Calvera},\ and\ \citenamefont {Hsieh}}]{Cheng-Ju2020}%
  \BibitemOpen
  \bibfield  {author} {\bibinfo {author} {\bibfnamefont {C.-J.}\ \bibnamefont
  {Lin}}, \bibinfo {author} {\bibfnamefont {V.}~\bibnamefont {Calvera}},\ and\
  \bibinfo {author} {\bibfnamefont {T.~H.}\ \bibnamefont {Hsieh}},\ }\bibfield
  {title} {\bibinfo {title} {Quantum many-body scar states in two-dimensional
  rydberg atom arrays},\ }\href {https://doi.org/10.1103/PhysRevB.101.220304}
  {\bibfield  {journal} {\bibinfo  {journal} {Phys. Rev. B}\ }\textbf {\bibinfo
  {volume} {101}},\ \bibinfo {pages} {220304} (\bibinfo {year}
  {2020})}\BibitemShut {NoStop}%
\bibitem [{\citenamefont {Notarnicola}\ \emph {et~al.}(2021)\citenamefont
  {Notarnicola}, \citenamefont {Elben}, \citenamefont {Lahaye}, \citenamefont
  {Browaeys}, \citenamefont {Montangero},\ and\ \citenamefont
  {Vermersch}}]{notarnicola2021randomized}%
  \BibitemOpen
  \bibfield  {author} {\bibinfo {author} {\bibfnamefont {S.}~\bibnamefont
  {Notarnicola}}, \bibinfo {author} {\bibfnamefont {A.}~\bibnamefont {Elben}},
  \bibinfo {author} {\bibfnamefont {T.}~\bibnamefont {Lahaye}}, \bibinfo
  {author} {\bibfnamefont {A.}~\bibnamefont {Browaeys}}, \bibinfo {author}
  {\bibfnamefont {S.}~\bibnamefont {Montangero}},\ and\ \bibinfo {author}
  {\bibfnamefont {B.}~\bibnamefont {Vermersch}},\ }\href@noop {} {\bibinfo
  {title} {A randomized measurement toolbox for rydberg quantum technologies}}
  (\bibinfo {year} {2021}),\ \Eprint {https://arxiv.org/abs/2112.11046}
  {arXiv:2112.11046 [quant-ph]} \BibitemShut {NoStop}%
\bibitem [{\citenamefont {Bluvstein}\ \emph {et~al.}(2021)\citenamefont
  {Bluvstein}, \citenamefont {Levine}, \citenamefont {Semeghini}, \citenamefont
  {Wang}, \citenamefont {Ebadi}, \citenamefont {Kalinowski}, \citenamefont
  {Keesling}, \citenamefont {Maskara}, \citenamefont {Pichler}, \citenamefont
  {Greiner}, \citenamefont {Vuletic},\ and\ \citenamefont
  {Lukin}}]{bluvstein2021quantum}%
  \BibitemOpen
  \bibfield  {author} {\bibinfo {author} {\bibfnamefont {D.}~\bibnamefont
  {Bluvstein}}, \bibinfo {author} {\bibfnamefont {H.}~\bibnamefont {Levine}},
  \bibinfo {author} {\bibfnamefont {G.}~\bibnamefont {Semeghini}}, \bibinfo
  {author} {\bibfnamefont {T.~T.}\ \bibnamefont {Wang}}, \bibinfo {author}
  {\bibfnamefont {S.}~\bibnamefont {Ebadi}}, \bibinfo {author} {\bibfnamefont
  {M.}~\bibnamefont {Kalinowski}}, \bibinfo {author} {\bibfnamefont
  {A.}~\bibnamefont {Keesling}}, \bibinfo {author} {\bibfnamefont
  {N.}~\bibnamefont {Maskara}}, \bibinfo {author} {\bibfnamefont
  {H.}~\bibnamefont {Pichler}}, \bibinfo {author} {\bibfnamefont
  {M.}~\bibnamefont {Greiner}}, \bibinfo {author} {\bibfnamefont
  {V.}~\bibnamefont {Vuletic}},\ and\ \bibinfo {author} {\bibfnamefont {M.~D.}\
  \bibnamefont {Lukin}},\ }\href@noop {} {\bibinfo {title} {A quantum processor
  based on coherent transport of entangled atom arrays}} (\bibinfo {year}
  {2021}),\ \Eprint {https://arxiv.org/abs/2112.03923} {arXiv:2112.03923
  [quant-ph]} \BibitemShut {NoStop}%
\bibitem [{\citenamefont {Eisert}\ \emph {et~al.}(2010)\citenamefont {Eisert},
  \citenamefont {Cramer},\ and\ \citenamefont {Plenio}}]{eisert2010colloquium}%
  \BibitemOpen
  \bibfield  {author} {\bibinfo {author} {\bibfnamefont {J.}~\bibnamefont
  {Eisert}}, \bibinfo {author} {\bibfnamefont {M.}~\bibnamefont {Cramer}},\
  and\ \bibinfo {author} {\bibfnamefont {M.~B.}\ \bibnamefont {Plenio}},\
  }\bibfield  {title} {\bibinfo {title} {Colloquium: Area laws for the
  entanglement entropy},\ }\href {https://doi.org/10.1103/RevModPhys.82.277}
  {\bibfield  {journal} {\bibinfo  {journal} {Rev. Mod. Phys.}\ }\textbf
  {\bibinfo {volume} {82}},\ \bibinfo {pages} {277} (\bibinfo {year}
  {2010})}\BibitemShut {NoStop}%
\bibitem [{\citenamefont {Hastings}(2007)}]{hastings2007}%
  \BibitemOpen
  \bibfield  {author} {\bibinfo {author} {\bibfnamefont {M.~B.}\ \bibnamefont
  {Hastings}},\ }\bibfield  {title} {\bibinfo {title} {An area law for
  one-dimensional quantum systems},\ }\href
  {https://doi.org/10.1088/1742-5468/2007/08/p08024} {\bibfield  {journal}
  {\bibinfo  {journal} {Journal of Statistical Mechanics: Theory and
  Experiment}\ }\textbf {\bibinfo {volume} {2007}},\ \bibinfo {pages}
  {P08024–P08024} (\bibinfo {year} {2007})}\BibitemShut {NoStop}%
\bibitem [{\citenamefont {Schollwöck}(2011)}]{schollwock2011mps}%
  \BibitemOpen
  \bibfield  {author} {\bibinfo {author} {\bibfnamefont {U.}~\bibnamefont
  {Schollwöck}},\ }\bibfield  {title} {\bibinfo {title} {The density-matrix
  renormalization group in the age of matrix product states},\ }\href
  {https://doi.org/10.1016/j.aop.2010.09.012} {\bibfield  {journal} {\bibinfo
  {journal} {Annals of Physics}\ }\textbf {\bibinfo {volume} {326}},\ \bibinfo
  {pages} {96–192} (\bibinfo {year} {2011})}\BibitemShut {NoStop}%
\bibitem [{\citenamefont {Garnerone}\ \emph {et~al.}(2010)\citenamefont
  {Garnerone}, \citenamefont {de~Oliveira},\ and\ \citenamefont
  {Zanardi}}]{garnerone2010tipicality}%
  \BibitemOpen
  \bibfield  {author} {\bibinfo {author} {\bibfnamefont {S.}~\bibnamefont
  {Garnerone}}, \bibinfo {author} {\bibfnamefont {T.~R.}\ \bibnamefont
  {de~Oliveira}},\ and\ \bibinfo {author} {\bibfnamefont {P.}~\bibnamefont
  {Zanardi}},\ }\bibfield  {title} {\bibinfo {title} {Typicality in random
  matrix product states},\ }\href {https://doi.org/10.1103/PhysRevA.81.032336}
  {\bibfield  {journal} {\bibinfo  {journal} {Phys. Rev. A}\ }\textbf {\bibinfo
  {volume} {81}},\ \bibinfo {pages} {032336} (\bibinfo {year}
  {2010})}\BibitemShut {NoStop}%
\bibitem [{\citenamefont {Ruggiero}\ \emph {et~al.}(2016)\citenamefont
  {Ruggiero}, \citenamefont {Alba},\ and\ \citenamefont
  {Calabrese}}]{ruggiero2016randomsinglet}%
  \BibitemOpen
  \bibfield  {author} {\bibinfo {author} {\bibfnamefont {P.}~\bibnamefont
  {Ruggiero}}, \bibinfo {author} {\bibfnamefont {V.}~\bibnamefont {Alba}},\
  and\ \bibinfo {author} {\bibfnamefont {P.}~\bibnamefont {Calabrese}},\
  }\bibfield  {title} {\bibinfo {title} {Entanglement negativity in random spin
  chains},\ }\href {https://doi.org/10.1103/PhysRevB.94.035152} {\bibfield
  {journal} {\bibinfo  {journal} {Phys. Rev. B}\ }\textbf {\bibinfo {volume}
  {94}},\ \bibinfo {pages} {035152} (\bibinfo {year} {2016})}\BibitemShut
  {NoStop}%
\bibitem [{\citenamefont {Haferkamp}\ \emph {et~al.}(2021)\citenamefont
  {Haferkamp}, \citenamefont {Bertoni}, \citenamefont {Roth},\ and\
  \citenamefont {Eisert}}]{eisert2021randomMPS}%
  \BibitemOpen
  \bibfield  {author} {\bibinfo {author} {\bibfnamefont {J.}~\bibnamefont
  {Haferkamp}}, \bibinfo {author} {\bibfnamefont {C.}~\bibnamefont {Bertoni}},
  \bibinfo {author} {\bibfnamefont {I.}~\bibnamefont {Roth}},\ and\ \bibinfo
  {author} {\bibfnamefont {J.}~\bibnamefont {Eisert}},\ }\bibfield  {title}
  {\bibinfo {title} {Emergent statistical mechanics from properties of
  disordered random matrix product states},\ }\href
  {https://doi.org/10.1103/PRXQuantum.2.040308} {\bibfield  {journal} {\bibinfo
   {journal} {PRX Quantum}\ }\textbf {\bibinfo {volume} {2}},\ \bibinfo {pages}
  {040308} (\bibinfo {year} {2021})}\BibitemShut {NoStop}%
\bibitem [{\citenamefont {Murciano}\ \emph {et~al.}(2021)\citenamefont
  {Murciano}, \citenamefont {Alba},\ and\ \citenamefont
  {Calabrese}}]{murciano2021quench}%
  \BibitemOpen
  \bibfield  {author} {\bibinfo {author} {\bibfnamefont {S.}~\bibnamefont
  {Murciano}}, \bibinfo {author} {\bibfnamefont {V.}~\bibnamefont {Alba}},\
  and\ \bibinfo {author} {\bibfnamefont {P.}~\bibnamefont {Calabrese}},\
  }\href@noop {} {\bibinfo {title} {Quench dynamics of r\'enyi negativities and
  the quasiparticle picture}} (\bibinfo {year} {2021}),\ \Eprint
  {https://arxiv.org/abs/2110.14589} {arXiv:2110.14589 [cond-mat.stat-mech]}
  \BibitemShut {NoStop}%
\bibitem [{\citenamefont {Eisler}\ and\ \citenamefont
  {Zimbor{\'{a}}s}(2015)}]{Eisler_2015}%
  \BibitemOpen
  \bibfield  {author} {\bibinfo {author} {\bibfnamefont {V.}~\bibnamefont
  {Eisler}}\ and\ \bibinfo {author} {\bibfnamefont {Z.}~\bibnamefont
  {Zimbor{\'{a}}s}},\ }\bibfield  {title} {\bibinfo {title} {On the partial
  transpose of fermionic gaussian states},\ }\href
  {https://doi.org/10.1088/1367-2630/17/5/053048} {\bibfield  {journal}
  {\bibinfo  {journal} {New Journal of Physics}\ }\textbf {\bibinfo {volume}
  {17}},\ \bibinfo {pages} {053048} (\bibinfo {year} {2015})}\BibitemShut
  {NoStop}%
\bibitem [{\citenamefont {Eisler}\ and\ \citenamefont
  {Zimbor\'as}(2016)}]{eislerzimboras2016}%
  \BibitemOpen
  \bibfield  {author} {\bibinfo {author} {\bibfnamefont {V.}~\bibnamefont
  {Eisler}}\ and\ \bibinfo {author} {\bibfnamefont {Z.}~\bibnamefont
  {Zimbor\'as}},\ }\bibfield  {title} {\bibinfo {title} {Entanglement
  negativity in two-dimensional free lattice models},\ }\href
  {https://doi.org/10.1103/PhysRevB.93.115148} {\bibfield  {journal} {\bibinfo
  {journal} {Phys. Rev. B}\ }\textbf {\bibinfo {volume} {93}},\ \bibinfo
  {pages} {115148} (\bibinfo {year} {2016})}\BibitemShut {NoStop}%
\bibitem [{\citenamefont {Valiant}(2001)}]{Valiant2001}%
  \BibitemOpen
  \bibfield  {author} {\bibinfo {author} {\bibfnamefont {L.~G.}\ \bibnamefont
  {Valiant}},\ }\bibfield  {title} {\bibinfo {title} {Quantum computers that
  can be simulated classically in polynomial time},\ }in\ \href
  {https://doi.org/10.1145/380752.380785} {\emph {\bibinfo {booktitle}
  {Proceedings of the Thirty-Third Annual ACM Symposium on Theory of
  Computing}}}\ (\bibinfo  {publisher} {Association for Computing Machinery},\
  \bibinfo {address} {New York, NY, USA},\ \bibinfo {year} {2001})\BibitemShut
  {NoStop}%
\bibitem [{\citenamefont {Jozsa}\ and\ \citenamefont
  {Miyake}(2008{\natexlab{a}})}]{JM08}%
  \BibitemOpen
  \bibfield  {author} {\bibinfo {author} {\bibfnamefont {R.}~\bibnamefont
  {Jozsa}}\ and\ \bibinfo {author} {\bibfnamefont {A.}~\bibnamefont {Miyake}},\
  }\bibfield  {title} {\bibinfo {title} {Matchgates and classical simulation of
  quantum circuits},\ }\href@noop {} {\bibfield  {journal} {\bibinfo  {journal}
  {Proc. R. Soc. A}\ }\textbf {\bibinfo {volume} {464}},\ \bibinfo {pages}
  {3089} (\bibinfo {year} {2008}{\natexlab{a}})}\BibitemShut {NoStop}%
\bibitem [{\citenamefont {Terhal}\ and\ \citenamefont
  {DiVincenzo}(2002)}]{TD02}%
  \BibitemOpen
  \bibfield  {author} {\bibinfo {author} {\bibfnamefont {B.~M.}\ \bibnamefont
  {Terhal}}\ and\ \bibinfo {author} {\bibfnamefont {D.~P.}\ \bibnamefont
  {DiVincenzo}},\ }\bibfield  {title} {\bibinfo {title} {Classical simulation
  of noninteracting-fermion quantum circuits},\ }\href
  {https://doi.org/10.1103/PhysRevA.65.032325} {\bibfield  {journal} {\bibinfo
  {journal} {Phys. Rev. A}\ }\textbf {\bibinfo {volume} {65}},\ \bibinfo
  {pages} {032325} (\bibinfo {year} {2002})}\BibitemShut {NoStop}%
\bibitem [{\citenamefont {Peschel}(2003)}]{Peschel_2003}%
  \BibitemOpen
  \bibfield  {author} {\bibinfo {author} {\bibfnamefont {I.}~\bibnamefont
  {Peschel}},\ }\bibfield  {title} {\bibinfo {title} {Calculation of reduced
  density matrices from correlation functions},\ }\href
  {https://doi.org/10.1088/0305-4470/36/14/101} {\bibfield  {journal} {\bibinfo
   {journal} {J. Phys. A}\ }\textbf {\bibinfo {volume} {36}},\ \bibinfo {pages}
  {L205–L208} (\bibinfo {year} {2003})}\BibitemShut {NoStop}%
\bibitem [{\citenamefont {Peschel}\ and\ \citenamefont
  {Eisler}(2009)}]{Peschel_2009}%
  \BibitemOpen
  \bibfield  {author} {\bibinfo {author} {\bibfnamefont {I.}~\bibnamefont
  {Peschel}}\ and\ \bibinfo {author} {\bibfnamefont {V.}~\bibnamefont
  {Eisler}},\ }\bibfield  {title} {\bibinfo {title} {Reduced density matrices
  and entanglement entropy in free lattice models},\ }\href
  {https://doi.org/10.1088/1751-8113/42/50/504003} {\bibfield  {journal}
  {\bibinfo  {journal} {J. Phys. A}\ }\textbf {\bibinfo {volume} {42}},\
  \bibinfo {pages} {504003} (\bibinfo {year} {2009})}\BibitemShut {NoStop}%
\bibitem [{\citenamefont {Bravyi}\ and\ \citenamefont {Kitaev}(2002)}]{BK02}%
  \BibitemOpen
  \bibfield  {author} {\bibinfo {author} {\bibfnamefont {S.~B.}\ \bibnamefont
  {Bravyi}}\ and\ \bibinfo {author} {\bibfnamefont {A.~Y.}\ \bibnamefont
  {Kitaev}},\ }\bibfield  {title} {\bibinfo {title} {Fermionic quantum
  computation},\ }\href
  {https://doi.org/https://doi.org/10.1006/aphy.2002.6254} {\bibfield
  {journal} {\bibinfo  {journal} {Annals of Physics}\ }\textbf {\bibinfo
  {volume} {298}},\ \bibinfo {pages} {210} (\bibinfo {year}
  {2002})}\BibitemShut {NoStop}%
\bibitem [{\citenamefont {Shapourian}\ \emph {et~al.}(2017)\citenamefont
  {Shapourian}, \citenamefont {Shiozaki},\ and\ \citenamefont
  {Ryu}}]{shapourian2017}%
  \BibitemOpen
  \bibfield  {author} {\bibinfo {author} {\bibfnamefont {H.}~\bibnamefont
  {Shapourian}}, \bibinfo {author} {\bibfnamefont {K.}~\bibnamefont
  {Shiozaki}},\ and\ \bibinfo {author} {\bibfnamefont {S.}~\bibnamefont
  {Ryu}},\ }\bibfield  {title} {\bibinfo {title} {Partial time-reversal
  transformation and entanglement negativity in fermionic systems},\ }\href
  {https://doi.org/10.1103/PhysRevB.95.165101} {\bibfield  {journal} {\bibinfo
  {journal} {Phys. Rev. B}\ }\textbf {\bibinfo {volume} {95}},\ \bibinfo
  {pages} {165101} (\bibinfo {year} {2017})}\BibitemShut {NoStop}%
\bibitem [{\citenamefont {Murciano}\ \emph {et~al.}(2022)\citenamefont
  {Murciano}, \citenamefont {Vitale}, \citenamefont {Dalmonte},\ and\
  \citenamefont {Calabrese}}]{murciano2022negativity}%
  \BibitemOpen
  \bibfield  {author} {\bibinfo {author} {\bibfnamefont {S.}~\bibnamefont
  {Murciano}}, \bibinfo {author} {\bibfnamefont {V.}~\bibnamefont {Vitale}},
  \bibinfo {author} {\bibfnamefont {M.}~\bibnamefont {Dalmonte}},\ and\
  \bibinfo {author} {\bibfnamefont {P.}~\bibnamefont {Calabrese}},\ }\bibfield
  {title} {\bibinfo {title} {Negativity hamiltonian: An operator
  characterization of mixed-state entanglement},\ }\href
  {https://doi.org/10.1103/PhysRevLett.128.140502} {\bibfield  {journal}
  {\bibinfo  {journal} {Phys. Rev. Lett.}\ }\textbf {\bibinfo {volume} {128}},\
  \bibinfo {pages} {140502} (\bibinfo {year} {2022})}\BibitemShut {NoStop}%
\bibitem [{Note2()}]{Note2}%
  \BibitemOpen
  \bibinfo {note} {In all other cases, one can bring the systems in this order
  by applying the corresponding fermionic SWAP gates, defined by $\mathinner
  {|{ab}\rangle }\DOTSB \mapstochar \rightarrow (-1)^{ab}\mathinner
  {|{ba}\rangle }$, to the state $\mathinner {|{\Phi }\rangle }$.}\BibitemShut
  {Stop}%
\bibitem [{\citenamefont {Jozsa}\ and\ \citenamefont
  {Miyake}(2008{\natexlab{b}})}]{josza2008mg}%
  \BibitemOpen
  \bibfield  {author} {\bibinfo {author} {\bibfnamefont {R.}~\bibnamefont
  {Jozsa}}\ and\ \bibinfo {author} {\bibfnamefont {A.}~\bibnamefont {Miyake}},\
  }\bibfield  {title} {\bibinfo {title} {Matchgates and classical simulation of
  quantum circuits},\ }\href {https://doi.org/10.1098/rspa.2008.0189}
  {\bibfield  {journal} {\bibinfo  {journal} {Proceedings of the Royal Society
  A: Mathematical, Physical and Engineering Sciences}\ }\textbf {\bibinfo
  {volume} {464}},\ \bibinfo {pages} {3089–3106} (\bibinfo {year}
  {2008}{\natexlab{b}})}\BibitemShut {NoStop}%
\bibitem [{\citenamefont {Helsen}\ \emph {et~al.}(2022)\citenamefont {Helsen},
  \citenamefont {Nezami}, \citenamefont {Reagor},\ and\ \citenamefont
  {Walter}}]{HNRW22}%
  \BibitemOpen
  \bibfield  {author} {\bibinfo {author} {\bibfnamefont {J.}~\bibnamefont
  {Helsen}}, \bibinfo {author} {\bibfnamefont {S.}~\bibnamefont {Nezami}},
  \bibinfo {author} {\bibfnamefont {M.}~\bibnamefont {Reagor}},\ and\ \bibinfo
  {author} {\bibfnamefont {M.}~\bibnamefont {Walter}},\ }\bibfield  {title}
  {\bibinfo {title} {Matchgate benchmarking: {S}calable benchmarking of a
  continuous family of many-qubit gates},\ }\href
  {https://doi.org/10.22331/q-2022-02-21-657} {\bibfield  {journal} {\bibinfo
  {journal} {{Quantum}}\ }\textbf {\bibinfo {volume} {6}},\ \bibinfo {pages}
  {657} (\bibinfo {year} {2022})}\BibitemShut {NoStop}%
\bibitem [{\citenamefont {Hebenstreit}\ \emph {et~al.}(2019)\citenamefont
  {Hebenstreit}, \citenamefont {Jozsa}, \citenamefont {Kraus}, \citenamefont
  {Strelchuk},\ and\ \citenamefont {Yoganathan}}]{hebenstreit2019magic}%
  \BibitemOpen
  \bibfield  {author} {\bibinfo {author} {\bibfnamefont {M.}~\bibnamefont
  {Hebenstreit}}, \bibinfo {author} {\bibfnamefont {R.}~\bibnamefont {Jozsa}},
  \bibinfo {author} {\bibfnamefont {B.}~\bibnamefont {Kraus}}, \bibinfo
  {author} {\bibfnamefont {S.}~\bibnamefont {Strelchuk}},\ and\ \bibinfo
  {author} {\bibfnamefont {M.}~\bibnamefont {Yoganathan}},\ }\bibfield  {title}
  {\bibinfo {title} {All pure fermionic non-gaussian states are magic states
  for matchgate computations},\ }\href
  {https://doi.org/10.1103/PhysRevLett.123.080503} {\bibfield  {journal}
  {\bibinfo  {journal} {Phys. Rev. Lett.}\ }\textbf {\bibinfo {volume} {123}},\
  \bibinfo {pages} {080503} (\bibinfo {year} {2019})}\BibitemShut {NoStop}%
\bibitem [{\citenamefont {Hebenstreit}\ \emph {et~al.}(2020)\citenamefont
  {Hebenstreit}, \citenamefont {Jozsa}, \citenamefont {Kraus},\ and\
  \citenamefont {Strelchuk}}]{hebenstreit2020mg}%
  \BibitemOpen
  \bibfield  {author} {\bibinfo {author} {\bibfnamefont {M.}~\bibnamefont
  {Hebenstreit}}, \bibinfo {author} {\bibfnamefont {R.}~\bibnamefont {Jozsa}},
  \bibinfo {author} {\bibfnamefont {B.}~\bibnamefont {Kraus}},\ and\ \bibinfo
  {author} {\bibfnamefont {S.}~\bibnamefont {Strelchuk}},\ }\bibfield  {title}
  {\bibinfo {title} {Computational power of matchgates with supplementary
  resources},\ }\href {https://doi.org/10.1103/PhysRevA.102.052604} {\bibfield
  {journal} {\bibinfo  {journal} {Phys. Rev. A}\ }\textbf {\bibinfo {volume}
  {102}},\ \bibinfo {pages} {052604} (\bibinfo {year} {2020})}\BibitemShut
  {NoStop}%
\bibitem [{\citenamefont {Huang}\ \emph {et~al.}(2020)\citenamefont {Huang},
  \citenamefont {Kueng},\ and\ \citenamefont {Preskill}}]{huang2020shadows}%
  \BibitemOpen
  \bibfield  {author} {\bibinfo {author} {\bibfnamefont {H.-Y.}\ \bibnamefont
  {Huang}}, \bibinfo {author} {\bibfnamefont {R.}~\bibnamefont {Kueng}},\ and\
  \bibinfo {author} {\bibfnamefont {J.}~\bibnamefont {Preskill}},\ }\bibfield
  {title} {\bibinfo {title} {Predicting many properties of a quantum system
  from very few measurements},\ }\href
  {https://doi.org/10.1038/s41567-020-0932-7} {\bibfield  {journal} {\bibinfo
  {journal} {Nature Physics}\ }\textbf {\bibinfo {volume} {16}},\ \bibinfo
  {pages} {1050–1057} (\bibinfo {year} {2020})}\BibitemShut {NoStop}%
\bibitem [{\citenamefont {Rath}\ \emph {et~al.}(2021)\citenamefont {Rath},
  \citenamefont {Branciard}, \citenamefont {Minguzzi},\ and\ \citenamefont
  {Vermersch}}]{rath2021quantum}%
  \BibitemOpen
  \bibfield  {author} {\bibinfo {author} {\bibfnamefont {A.}~\bibnamefont
  {Rath}}, \bibinfo {author} {\bibfnamefont {C.}~\bibnamefont {Branciard}},
  \bibinfo {author} {\bibfnamefont {A.}~\bibnamefont {Minguzzi}},\ and\
  \bibinfo {author} {\bibfnamefont {B.}~\bibnamefont {Vermersch}},\ }\bibfield
  {title} {\bibinfo {title} {Quantum fisher information from randomized
  measurements},\ }\href {https://doi.org/10.1103/PhysRevLett.127.260501}
  {\bibfield  {journal} {\bibinfo  {journal} {Phys. Rev. Lett.}\ }\textbf
  {\bibinfo {volume} {127}},\ \bibinfo {pages} {260501} (\bibinfo {year}
  {2021})}\BibitemShut {NoStop}%
\bibitem [{\citenamefont {Haah}\ \emph {et~al.}(2017)\citenamefont {Haah},
  \citenamefont {Harrow}, \citenamefont {Ji}, \citenamefont {Wu},\ and\
  \citenamefont {Yu}}]{haah2017sample}%
  \BibitemOpen
  \bibfield  {author} {\bibinfo {author} {\bibfnamefont {J.}~\bibnamefont
  {Haah}}, \bibinfo {author} {\bibfnamefont {A.~W.}\ \bibnamefont {Harrow}},
  \bibinfo {author} {\bibfnamefont {Z.}~\bibnamefont {Ji}}, \bibinfo {author}
  {\bibfnamefont {X.}~\bibnamefont {Wu}},\ and\ \bibinfo {author}
  {\bibfnamefont {N.}~\bibnamefont {Yu}},\ }\bibfield  {title} {\bibinfo
  {title} {Sample-optimal tomography of quantum states},\ }\href
  {https://doi.org/10.1109/TIT.2017.2719044} {\bibfield  {journal} {\bibinfo
  {journal} {IEEE Transactions on Information Theory}\ }\textbf {\bibinfo
  {volume} {63}},\ \bibinfo {pages} {5628} (\bibinfo {year}
  {2017})}\BibitemShut {NoStop}%
\bibitem [{\citenamefont {Rath}\ \emph {et~al.}(2022)\citenamefont {Rath},
  \citenamefont {Vitale}, \citenamefont {Murciano}, \citenamefont {Votto},
  \citenamefont {Dubail}, \citenamefont {Kueng}, \citenamefont {Branciard},
  \citenamefont {Calabrese},\ and\ \citenamefont
  {Vermersch}}]{rath2022entanglement}%
  \BibitemOpen
  \bibfield  {author} {\bibinfo {author} {\bibfnamefont {A.}~\bibnamefont
  {Rath}}, \bibinfo {author} {\bibfnamefont {V.}~\bibnamefont {Vitale}},
  \bibinfo {author} {\bibfnamefont {S.}~\bibnamefont {Murciano}}, \bibinfo
  {author} {\bibfnamefont {M.}~\bibnamefont {Votto}}, \bibinfo {author}
  {\bibfnamefont {J.}~\bibnamefont {Dubail}}, \bibinfo {author} {\bibfnamefont
  {R.}~\bibnamefont {Kueng}}, \bibinfo {author} {\bibfnamefont
  {C.}~\bibnamefont {Branciard}}, \bibinfo {author} {\bibfnamefont
  {P.}~\bibnamefont {Calabrese}},\ and\ \bibinfo {author} {\bibfnamefont
  {B.}~\bibnamefont {Vermersch}},\ }\href@noop {} {\bibinfo {title}
  {Entanglement barrier and its symmetry resolution: theory and experiment}}
  (\bibinfo {year} {2022}),\ \Eprint {https://arxiv.org/abs/2209.04393}
  {arXiv:2209.04393} \BibitemShut {NoStop}%
\bibitem [{\citenamefont {Alves}\ and\ \citenamefont
  {Jaksch}(2004)}]{Alves2004}%
  \BibitemOpen
  \bibfield  {author} {\bibinfo {author} {\bibfnamefont {C.~M.}\ \bibnamefont
  {Alves}}\ and\ \bibinfo {author} {\bibfnamefont {D.}~\bibnamefont {Jaksch}},\
  }\bibfield  {title} {\bibinfo {title} {{Multipartite entanglement detection
  in bosons}},\ }\href {https://doi.org/10.1103/PhysRevLett.93.110501}
  {\bibfield  {journal} {\bibinfo  {journal} {Phys. Rev. Lett.}\ }\textbf
  {\bibinfo {volume} {93}},\ \bibinfo {pages} {110501} (\bibinfo {year}
  {2004})}\BibitemShut {NoStop}%
\bibitem [{\citenamefont {Daley}\ \emph {et~al.}(2012)\citenamefont {Daley},
  \citenamefont {Pichler}, \citenamefont {Schachenmayer},\ and\ \citenamefont
  {Zoller}}]{Daley2012}%
  \BibitemOpen
  \bibfield  {author} {\bibinfo {author} {\bibfnamefont {A.~J.}\ \bibnamefont
  {Daley}}, \bibinfo {author} {\bibfnamefont {H.}~\bibnamefont {Pichler}},
  \bibinfo {author} {\bibfnamefont {J.}~\bibnamefont {Schachenmayer}},\ and\
  \bibinfo {author} {\bibfnamefont {P.}~\bibnamefont {Zoller}},\ }\bibfield
  {title} {\bibinfo {title} {{Measuring entanglement growth in quench dynamics
  of bosons in an optical lattice}},\ }\href
  {https://doi.org/10.1103/PhysRevLett.109.020505} {\bibfield  {journal}
  {\bibinfo  {journal} {Phys. Rev. Lett.}\ }\textbf {\bibinfo {volume} {109}},\
  \bibinfo {pages} {020505} (\bibinfo {year} {2012})}\BibitemShut {NoStop}%
\bibitem [{\citenamefont {Islam}\ \emph {et~al.}(2015)\citenamefont {Islam},
  \citenamefont {Ma}, \citenamefont {Preiss}, \citenamefont {{Eric Tai}},
  \citenamefont {Lukin}, \citenamefont {Rispoli},\ and\ \citenamefont
  {Greiner}}]{Islam2015}%
  \BibitemOpen
  \bibfield  {author} {\bibinfo {author} {\bibfnamefont {R.}~\bibnamefont
  {Islam}}, \bibinfo {author} {\bibfnamefont {R.}~\bibnamefont {Ma}}, \bibinfo
  {author} {\bibfnamefont {P.~M.}\ \bibnamefont {Preiss}}, \bibinfo {author}
  {\bibfnamefont {M.}~\bibnamefont {{Eric Tai}}}, \bibinfo {author}
  {\bibfnamefont {A.}~\bibnamefont {Lukin}}, \bibinfo {author} {\bibfnamefont
  {M.}~\bibnamefont {Rispoli}},\ and\ \bibinfo {author} {\bibfnamefont
  {M.}~\bibnamefont {Greiner}},\ }\bibfield  {title} {\bibinfo {title}
  {{Measuring entanglement entropy in a quantum many-body system}},\ }\href
  {https://doi.org/10.1038/nature15750} {\bibfield  {journal} {\bibinfo
  {journal} {Nature}\ }\textbf {\bibinfo {volume} {528}},\ \bibinfo {pages}
  {77} (\bibinfo {year} {2015})}\BibitemShut {NoStop}%
\bibitem [{\citenamefont {Rudolph}(2000)}]{rudolph2000separability}%
  \BibitemOpen
  \bibfield  {author} {\bibinfo {author} {\bibfnamefont {O.}~\bibnamefont
  {Rudolph}},\ }\bibfield  {title} {\bibinfo {title} {A separability criterion
  for density operators},\ }\href {https://doi.org/10.1088/0305-4470/33/21/308}
  {\bibfield  {journal} {\bibinfo  {journal} {Journal of Physics A:
  Mathematical and General}\ }\textbf {\bibinfo {volume} {33}},\ \bibinfo
  {pages} {3951} (\bibinfo {year} {2000})}\BibitemShut {NoStop}%
\bibitem [{\citenamefont {Chen}\ and\ \citenamefont {Wu}(2003)}]{ChenCCNR2003}%
  \BibitemOpen
  \bibfield  {author} {\bibinfo {author} {\bibfnamefont {K.}~\bibnamefont
  {Chen}}\ and\ \bibinfo {author} {\bibfnamefont {L.-A.}\ \bibnamefont {Wu}},\
  }\bibfield  {title} {\bibinfo {title} {A matrix realignment method for
  recognizing entanglement},\ }\href {https://arxiv.org/abs/quant-ph/0205017}
  {\bibfield  {journal} {\bibinfo  {journal} {Quantum Info. Comput.}\ }\textbf
  {\bibinfo {volume} {3}},\ \bibinfo {pages} {193–202} (\bibinfo {year}
  {2003})}\BibitemShut {NoStop}%
\bibitem [{\citenamefont {Rudolph}(2005)}]{rudolph2005further}%
  \BibitemOpen
  \bibfield  {author} {\bibinfo {author} {\bibfnamefont {O.}~\bibnamefont
  {Rudolph}},\ }\bibfield  {title} {\bibinfo {title} {Further results on the
  cross norm criterion for separability},\ }\href
  {https://link.springer.com/article/10.1007/s11128-005-5664-1} {\bibfield
  {journal} {\bibinfo  {journal} {Quantum Info. Proc.}\ }\textbf {\bibinfo
  {volume} {4}},\ \bibinfo {pages} {219} (\bibinfo {year} {2005})}\BibitemShut
  {NoStop}%
\bibitem [{Note3()}]{Note3}%
  \BibitemOpen
  \bibinfo {note} {Note that clearly whenever we consider expectation values of
  operators $O_n=O_{AB}\otimes 1_C$ acting non-trivially only in subsystems
  $AB$ of the $n$ copies, then ${\protect \mathbb E}[o_n]=Tr(O_{AB}{\protect
  \mathbb E}[\rho _{AB}]$ where ${\protect \mathbb E}[\rho _{AB}]$ is the
  expectation value of Haar-random {\protect \em induced} mixed
  states.}\BibitemShut {Stop}%
\bibitem [{\citenamefont {Lee}\ and\ \citenamefont
  {Forthofer}(2006)}]{Lee2006stat}%
  \BibitemOpen
  \bibfield  {author} {\bibinfo {author} {\bibfnamefont {E.}~\bibnamefont
  {Lee}}\ and\ \bibinfo {author} {\bibfnamefont {R.}~\bibnamefont
  {Forthofer}},\ }\href {https://doi.org/10.4135/9781412983341} {\emph
  {\bibinfo {title} {{Analyzing Complex Survey Data}}}}\ (\bibinfo  {publisher}
  {SAGE Publications, Inc.},\ \bibinfo {year} {2006})\BibitemShut {NoStop}%
\bibitem [{\citenamefont {Brydges}\ \emph {et~al.}(2019)\citenamefont
  {Brydges}, \citenamefont {Elben}, \citenamefont {Jurcevic}, \citenamefont
  {Vermersch}, \citenamefont {Maier}, \citenamefont {Lanyon}, \citenamefont
  {Zoller}, \citenamefont {Blatt},\ and\ \citenamefont {Roos}}]{brydges2019}%
  \BibitemOpen
  \bibfield  {author} {\bibinfo {author} {\bibfnamefont {T.}~\bibnamefont
  {Brydges}}, \bibinfo {author} {\bibfnamefont {A.}~\bibnamefont {Elben}},
  \bibinfo {author} {\bibfnamefont {P.}~\bibnamefont {Jurcevic}}, \bibinfo
  {author} {\bibfnamefont {B.}~\bibnamefont {Vermersch}}, \bibinfo {author}
  {\bibfnamefont {C.}~\bibnamefont {Maier}}, \bibinfo {author} {\bibfnamefont
  {B.~P.}\ \bibnamefont {Lanyon}}, \bibinfo {author} {\bibfnamefont
  {P.}~\bibnamefont {Zoller}}, \bibinfo {author} {\bibfnamefont
  {R.}~\bibnamefont {Blatt}},\ and\ \bibinfo {author} {\bibfnamefont {C.~F.}\
  \bibnamefont {Roos}},\ }\bibfield  {title} {\bibinfo {title} {Probing rényi
  entanglement entropy via randomized measurements},\ }\href
  {https://doi.org/10.1126/science.aau4963} {\bibfield  {journal} {\bibinfo
  {journal} {Science}\ }\textbf {\bibinfo {volume} {364}},\ \bibinfo {pages}
  {260–263} (\bibinfo {year} {2019})}\BibitemShut {NoStop}%
\bibitem [{\citenamefont {Satzinger}\ \emph {et~al.}(2021)\citenamefont
  {Satzinger}, \citenamefont {Liu}, \citenamefont {Smith}, \citenamefont
  {Knapp}, \citenamefont {Newman}, \citenamefont {Jones}, \citenamefont {Chen},
  \citenamefont {Quintana}, \citenamefont {Mi}, \citenamefont {Dunsworth},
  \citenamefont {Gidney}, \citenamefont {Aleiner}, \citenamefont {Arute},
  \citenamefont {Arya}, \citenamefont {Atalaya}, \citenamefont {Babbush},
  \citenamefont {Bardin}, \citenamefont {Barends}, \citenamefont {Basso},
  \citenamefont {Bengtsson}, \citenamefont {Bilmes}, \citenamefont {Broughton},
  \citenamefont {Buckley}, \citenamefont {Buell}, \citenamefont {Burkett},
  \citenamefont {Bushnell}, \citenamefont {Chiaro}, \citenamefont {Collins},
  \citenamefont {Courtney}, \citenamefont {Demura}, \citenamefont {Derk},
  \citenamefont {Eppens}, \citenamefont {Erickson}, \citenamefont {Faoro},
  \citenamefont {Farhi}, \citenamefont {Fowler}, \citenamefont {Foxen},
  \citenamefont {Giustina}, \citenamefont {Greene}, \citenamefont {Gross},
  \citenamefont {Harrigan}, \citenamefont {Harrington}, \citenamefont {Hilton},
  \citenamefont {Hong}, \citenamefont {Huang}, \citenamefont {Huggins},
  \citenamefont {Ioffe}, \citenamefont {Isakov}, \citenamefont {Jeffrey},
  \citenamefont {Jiang}, \citenamefont {Kafri}, \citenamefont {Kechedzhi},
  \citenamefont {Khattar}, \citenamefont {Kim}, \citenamefont {Klimov},
  \citenamefont {Korotkov}, \citenamefont {Kostritsa}, \citenamefont
  {Landhuis}, \citenamefont {Laptev}, \citenamefont {Locharla}, \citenamefont
  {Lucero}, \citenamefont {Martin}, \citenamefont {McClean}, \citenamefont
  {McEwen}, \citenamefont {Miao}, \citenamefont {Mohseni}, \citenamefont
  {Montazeri}, \citenamefont {Mruczkiewicz}, \citenamefont {Mutus},
  \citenamefont {Naaman}, \citenamefont {Neeley}, \citenamefont {Neill},
  \citenamefont {Niu}, \citenamefont {O’Brien}, \citenamefont {Opremcak},
  \citenamefont {Pató}, \citenamefont {Petukhov}, \citenamefont {Rubin},
  \citenamefont {Sank}, \citenamefont {Shvarts}, \citenamefont {Strain},
  \citenamefont {Szalay}, \citenamefont {Villalonga}, \citenamefont {White},
  \citenamefont {Yao}, \citenamefont {Yeh}, \citenamefont {Yoo}, \citenamefont
  {Zalcman}, \citenamefont {Neven}, \citenamefont {Boixo}, \citenamefont
  {Megrant}, \citenamefont {Chen}, \citenamefont {Kelly}, \citenamefont
  {Smelyanskiy}, \citenamefont {Kitaev}, \citenamefont {Knap}, \citenamefont
  {Pollmann},\ and\ \citenamefont {Roushan}}]{satzinger2021toric}%
  \BibitemOpen
  \bibfield  {author} {\bibinfo {author} {\bibfnamefont {K.~J.}\ \bibnamefont
  {Satzinger}}, \bibinfo {author} {\bibfnamefont {Y.-J.}\ \bibnamefont {Liu}},
  \bibinfo {author} {\bibfnamefont {A.}~\bibnamefont {Smith}}, \bibinfo
  {author} {\bibfnamefont {C.}~\bibnamefont {Knapp}}, \bibinfo {author}
  {\bibfnamefont {M.}~\bibnamefont {Newman}}, \bibinfo {author} {\bibfnamefont
  {C.}~\bibnamefont {Jones}}, \bibinfo {author} {\bibfnamefont
  {Z.}~\bibnamefont {Chen}}, \bibinfo {author} {\bibfnamefont {C.}~\bibnamefont
  {Quintana}}, \bibinfo {author} {\bibfnamefont {X.}~\bibnamefont {Mi}},
  \bibinfo {author} {\bibfnamefont {A.}~\bibnamefont {Dunsworth}}, \bibinfo
  {author} {\bibfnamefont {C.}~\bibnamefont {Gidney}}, \bibinfo {author}
  {\bibfnamefont {I.}~\bibnamefont {Aleiner}}, \bibinfo {author} {\bibfnamefont
  {F.}~\bibnamefont {Arute}}, \bibinfo {author} {\bibfnamefont
  {K.}~\bibnamefont {Arya}}, \bibinfo {author} {\bibfnamefont {J.}~\bibnamefont
  {Atalaya}}, \bibinfo {author} {\bibfnamefont {R.}~\bibnamefont {Babbush}},
  \bibinfo {author} {\bibfnamefont {J.~C.}\ \bibnamefont {Bardin}}, \bibinfo
  {author} {\bibfnamefont {R.}~\bibnamefont {Barends}}, \bibinfo {author}
  {\bibfnamefont {J.}~\bibnamefont {Basso}}, \bibinfo {author} {\bibfnamefont
  {A.}~\bibnamefont {Bengtsson}}, \bibinfo {author} {\bibfnamefont
  {A.}~\bibnamefont {Bilmes}}, \bibinfo {author} {\bibfnamefont
  {M.}~\bibnamefont {Broughton}}, \bibinfo {author} {\bibfnamefont {B.~B.}\
  \bibnamefont {Buckley}}, \bibinfo {author} {\bibfnamefont {D.~A.}\
  \bibnamefont {Buell}}, \bibinfo {author} {\bibfnamefont {B.}~\bibnamefont
  {Burkett}}, \bibinfo {author} {\bibfnamefont {N.}~\bibnamefont {Bushnell}},
  \bibinfo {author} {\bibfnamefont {B.}~\bibnamefont {Chiaro}}, \bibinfo
  {author} {\bibfnamefont {R.}~\bibnamefont {Collins}}, \bibinfo {author}
  {\bibfnamefont {W.}~\bibnamefont {Courtney}}, \bibinfo {author}
  {\bibfnamefont {S.}~\bibnamefont {Demura}}, \bibinfo {author} {\bibfnamefont
  {A.~R.}\ \bibnamefont {Derk}}, \bibinfo {author} {\bibfnamefont
  {D.}~\bibnamefont {Eppens}}, \bibinfo {author} {\bibfnamefont
  {C.}~\bibnamefont {Erickson}}, \bibinfo {author} {\bibfnamefont
  {L.}~\bibnamefont {Faoro}}, \bibinfo {author} {\bibfnamefont
  {E.}~\bibnamefont {Farhi}}, \bibinfo {author} {\bibfnamefont {A.~G.}\
  \bibnamefont {Fowler}}, \bibinfo {author} {\bibfnamefont {B.}~\bibnamefont
  {Foxen}}, \bibinfo {author} {\bibfnamefont {M.}~\bibnamefont {Giustina}},
  \bibinfo {author} {\bibfnamefont {A.}~\bibnamefont {Greene}}, \bibinfo
  {author} {\bibfnamefont {J.~A.}\ \bibnamefont {Gross}}, \bibinfo {author}
  {\bibfnamefont {M.~P.}\ \bibnamefont {Harrigan}}, \bibinfo {author}
  {\bibfnamefont {S.~D.}\ \bibnamefont {Harrington}}, \bibinfo {author}
  {\bibfnamefont {J.}~\bibnamefont {Hilton}}, \bibinfo {author} {\bibfnamefont
  {S.}~\bibnamefont {Hong}}, \bibinfo {author} {\bibfnamefont {T.}~\bibnamefont
  {Huang}}, \bibinfo {author} {\bibfnamefont {W.~J.}\ \bibnamefont {Huggins}},
  \bibinfo {author} {\bibfnamefont {L.~B.}\ \bibnamefont {Ioffe}}, \bibinfo
  {author} {\bibfnamefont {S.~V.}\ \bibnamefont {Isakov}}, \bibinfo {author}
  {\bibfnamefont {E.}~\bibnamefont {Jeffrey}}, \bibinfo {author} {\bibfnamefont
  {Z.}~\bibnamefont {Jiang}}, \bibinfo {author} {\bibfnamefont
  {D.}~\bibnamefont {Kafri}}, \bibinfo {author} {\bibfnamefont
  {K.}~\bibnamefont {Kechedzhi}}, \bibinfo {author} {\bibfnamefont
  {T.}~\bibnamefont {Khattar}}, \bibinfo {author} {\bibfnamefont
  {S.}~\bibnamefont {Kim}}, \bibinfo {author} {\bibfnamefont {P.~V.}\
  \bibnamefont {Klimov}}, \bibinfo {author} {\bibfnamefont {A.~N.}\
  \bibnamefont {Korotkov}}, \bibinfo {author} {\bibfnamefont {F.}~\bibnamefont
  {Kostritsa}}, \bibinfo {author} {\bibfnamefont {D.}~\bibnamefont {Landhuis}},
  \bibinfo {author} {\bibfnamefont {P.}~\bibnamefont {Laptev}}, \bibinfo
  {author} {\bibfnamefont {A.}~\bibnamefont {Locharla}}, \bibinfo {author}
  {\bibfnamefont {E.}~\bibnamefont {Lucero}}, \bibinfo {author} {\bibfnamefont
  {O.}~\bibnamefont {Martin}}, \bibinfo {author} {\bibfnamefont {J.~R.}\
  \bibnamefont {McClean}}, \bibinfo {author} {\bibfnamefont {M.}~\bibnamefont
  {McEwen}}, \bibinfo {author} {\bibfnamefont {K.~C.}\ \bibnamefont {Miao}},
  \bibinfo {author} {\bibfnamefont {M.}~\bibnamefont {Mohseni}}, \bibinfo
  {author} {\bibfnamefont {S.}~\bibnamefont {Montazeri}}, \bibinfo {author}
  {\bibfnamefont {W.}~\bibnamefont {Mruczkiewicz}}, \bibinfo {author}
  {\bibfnamefont {J.}~\bibnamefont {Mutus}}, \bibinfo {author} {\bibfnamefont
  {O.}~\bibnamefont {Naaman}}, \bibinfo {author} {\bibfnamefont
  {M.}~\bibnamefont {Neeley}}, \bibinfo {author} {\bibfnamefont
  {C.}~\bibnamefont {Neill}}, \bibinfo {author} {\bibfnamefont {M.~Y.}\
  \bibnamefont {Niu}}, \bibinfo {author} {\bibfnamefont {T.~E.}\ \bibnamefont
  {O’Brien}}, \bibinfo {author} {\bibfnamefont {A.}~\bibnamefont {Opremcak}},
  \bibinfo {author} {\bibfnamefont {B.}~\bibnamefont {Pató}}, \bibinfo
  {author} {\bibfnamefont {A.}~\bibnamefont {Petukhov}}, \bibinfo {author}
  {\bibfnamefont {N.~C.}\ \bibnamefont {Rubin}}, \bibinfo {author}
  {\bibfnamefont {D.}~\bibnamefont {Sank}}, \bibinfo {author} {\bibfnamefont
  {V.}~\bibnamefont {Shvarts}}, \bibinfo {author} {\bibfnamefont
  {D.}~\bibnamefont {Strain}}, \bibinfo {author} {\bibfnamefont
  {M.}~\bibnamefont {Szalay}}, \bibinfo {author} {\bibfnamefont
  {B.}~\bibnamefont {Villalonga}}, \bibinfo {author} {\bibfnamefont {T.~C.}\
  \bibnamefont {White}}, \bibinfo {author} {\bibfnamefont {Z.}~\bibnamefont
  {Yao}}, \bibinfo {author} {\bibfnamefont {P.}~\bibnamefont {Yeh}}, \bibinfo
  {author} {\bibfnamefont {J.}~\bibnamefont {Yoo}}, \bibinfo {author}
  {\bibfnamefont {A.}~\bibnamefont {Zalcman}}, \bibinfo {author} {\bibfnamefont
  {H.}~\bibnamefont {Neven}}, \bibinfo {author} {\bibfnamefont
  {S.}~\bibnamefont {Boixo}}, \bibinfo {author} {\bibfnamefont
  {A.}~\bibnamefont {Megrant}}, \bibinfo {author} {\bibfnamefont
  {Y.}~\bibnamefont {Chen}}, \bibinfo {author} {\bibfnamefont {J.}~\bibnamefont
  {Kelly}}, \bibinfo {author} {\bibfnamefont {V.}~\bibnamefont {Smelyanskiy}},
  \bibinfo {author} {\bibfnamefont {A.}~\bibnamefont {Kitaev}}, \bibinfo
  {author} {\bibfnamefont {M.}~\bibnamefont {Knap}}, \bibinfo {author}
  {\bibfnamefont {F.}~\bibnamefont {Pollmann}},\ and\ \bibinfo {author}
  {\bibfnamefont {P.}~\bibnamefont {Roushan}},\ }\bibfield  {title} {\bibinfo
  {title} {Realizing topologically ordered states on a quantum processor},\
  }\href {https://doi.org/10.1126/science.abi8378} {\bibfield  {journal}
  {\bibinfo  {journal} {Science}\ }\textbf {\bibinfo {volume} {374}},\ \bibinfo
  {pages} {1237–1241} (\bibinfo {year} {2021})}\BibitemShut {NoStop}%
\end{thebibliography}%

\begin{appendix}

\section{Effect of finite sampling and concentration effects}\label{sec:cov}
In this section, we address the role of statistical fluctuations when estimating $\tilde r_2$ from a finite number $K$ of random states $\ket{\psi_k}$ for $k=1,\ldots,K$. To this end, we will first recall how the expectation values of PT moments $\mathbb{E}[p_n]$ (appearing in, e.g., Eq.~\eqref{eq:SWAPbis}) and more generally quantities of the form  $\mathbb{E}[p_n\,p_m]$ can be expressed in terms of permutation operators.

\subsection{Basic properties of Haar-random states}\label{app:weingarten}
Here we review some basic and well-known results from random matrix theory that will be used below.

We are interested in expectation values of the form $o_n={\rm tr}(O_n\rho_u^{\otimes n})$, where $O_n$ is an operator acting on $n$ copies of a state $\rho_u=u\ket{0^N}\bra{0^N}u^{\dagger}$. In fact, what we will need is the mean value ${\mathbb E}[o_n]$ over Haar-random unitaries $u\sim{\rm U}(2^N)$, i.e~\footnote{Note that clearly whenever we consider expectation values of operators $O_n=O_{AB}\otimes1_C$ acting non-trivially only in subsystems $AB$ of the $n$ copies, then ${\mathbb E}[o_n]=Tr(O_{AB}{\mathbb E}[\rho_{AB}]$ where ${\mathbb E}[\rho_{AB}]$ is the expectation value of Haar-random {\em induced} mixed states.}
\begin{equation}
{\mathbb E}[o_n]={\rm tr}(O_n{\mathbb E}[\rho_u^{\otimes n}])\,.
\end{equation}
More explicitly, for any random variable $f(u)$ defined over elements of the unitary group $u\in{\rm U}(2^N)$, the expression ${\mathbb E}[f(u)]$ stands for the mean value over Haar-random unitaries $u\sim{\rm U}(2^N)$ sampled uniformly from the unique invariant (Haar) measure ${\rm d}u$. In other words, we define 
\begin{equation}
    {\mathbb E}[f(u)]:=\int{\rm d}u\,f(u)\,.
\end{equation}

For any operator $O_n$ acting on $n$ copies of the Hilbert space of $N$ qubits, consider the map
\begin{equation}
\Phi_n(O_n):={\mathbb E}[u^{\otimes n}O_n(u^\dagger)^{\otimes n}]=\int{\rm d}u\,u^{\otimes n}O_n(u^\dagger)^{\otimes n}\,.
\end{equation}
It is well-known~\cite{Collins2006RMT,Elben2019correlations} that the previous map can be expressed via the so-called twirling formula as
\begin{equation}
\Phi_n(O_n)=\sum_{\sigma, \tau \in S_n}\text{Wg}(\sigma\tau^{-1})\,\text{tr}(\Pi(\tau)O_n)\,\Pi(\sigma)\,.
\end{equation}
In the previous expression, $\text{Wg}(\cdot)$ are the Weingarten functions defined (see, e.g., Ref.~\cite{Collins2006RMT}) for any permutation $\sigma\in S_n$, where $S_n$ is the symmetric group over $n$ elements; and $\Pi(\sigma)$ denote the  permutations operators (acting on $n$ copies), i.e.,
\begin{equation}
    \Pi(\tau)\ket{\phi_1}\otimes\cdots\otimes\ket{\phi_n}=\ket{\phi_{\tau(1)}}\otimes\cdots\otimes\ket{\phi_{\tau(n)}}\,.
\end{equation}

For any $O_n$ being supported on the symmetric subspace, such as $O_n=(\ket{0^N}\bra{0^N})^{\otimes n}$, the facts that $\Pi(\tau)O_n=Z$ for all permutations $\tau$ and that $\sum_{ \tau \in S_n}\text{Wg}(\sigma\tau^{-1})$ is constant, implies that the Weingarten function decouples from the permutation operator $\Pi(\sigma)$. Using then that $\Phi_{n}$ is trace preserving, we obtain
\begin{equation}\label{eq:app-twirl}
    \Phi_{n}((\ket{0^N}\bra{0^N})^{\otimes n}) = \frac{\sum_{\sigma \in S_n}\Pi(\sigma)}{\sum_{\sigma \in S_n}\text{tr}[\Pi(\sigma)]}\,,
\end{equation}
which is nothing but the projector onto the symmetric subspace, a result that it is well-known. From the previous formula one can obtain the relevant equations used in the main text in the context of Haar random states. In particular, writing the PT moments in terms of multicopy observables, and using the formulas above,  Eq.~\eqref{eq:RMTexpression} was obtained.

Finally, note that it is possible to apply the previous equations to expressions of the form ${\mathbb E}[o_n o'_m]$, where $o'_m = \text{tr}[O'_m\rho_u^{\otimes m}]$. To this end, simply note that
\begin{multline}
    o_n o'_m = \text{tr}[O_n\rho_u^{\otimes n}]\text{tr}[O'_m\rho_u^{\otimes m}] =\\
    \text{tr}[O_n \otimes O'_m\rho_u^{\otimes (m+n)}]\,,
\end{multline}
and thus, using the linearity of the mean value in the previous equation and Eq.~\eqref{eq:app-twirl} with $n\to n+m$,
\begin{multline}\label{eq:cov}
    \mathbb{E}[o_n o_m]=\text{tr}[O_n \otimes O'_m\mathbb{E}[\rho_u^{\otimes (m+n)}]]\\ =\sum_{\sigma\in S_{(m+n)}}\frac{\text{tr}[O_n \otimes O'_m \Pi(\sigma)]}{\text{tr}[\Pi(\sigma)]}\,.
\end{multline}
The latter formula (together with Eq.~\eqref{eq:RMTexpression}) will be used in Sec.~\ref{app:var} to compute covariances of the form ${\rm Cov}[p_n,p_m]={\mathbb E}[p_n\,p_m]-{\mathbb E}[p_n]{\mathbb E}[p_m]$ of PT moments.

\subsection{Variance in the estimation of $\tilde r_2$ for a finite number of random states}\label{app:var}
We now address the role of statistical fluctuations when estimating $\tilde{r}_2$ from a finite number $K$ of random states $\ket{\psi}_{k=1,\dots,K}$. 

Here, we consider that we build an estimation  $\tilde{r}_2^{(e)}$ from empirical averages of PT moments $p_n^{(e)}$ over the $K$ random states. To this end, a central assumption is that the statistical fluctuations of $p_n^{(e)}$ around the mean values are sufficiently small. This can be explicitly check using the formulas of the previous Sec.~\ref{app:weingarten} for low-order PT moments. In particular, for the estimated values $p_n^{(e)}$ of PT moments of order $n=2,3,4$, the corresponding variances are exponentially smaller that their expectations squared. We show this for the particular case of $n=2$ in Fig.~\ref{fig:phase_diagram_r_2}. Using now a Taylor expansion around these mean values we have 
\begin{equation}
    \tilde r_2^{(e)} = \frac{p_2^{(e)}p_3^{(e)}}{p_4^{(e)}}
    \approx \tilde r_2\left(1
    + \sum_{n=2}^4\frac{a_n(p_n^{(e)}-{\mathbb E}[p_n]])}{{\mathbb E}[p_n]}
    \right)
\end{equation}
with $a_{2,3}=1$, $a_4=-1$.

Based on this approximation, we can express the variance of $\tilde{r}_2^{(e)}$ as 
\begin{eqnarray}
    \frac{\mathrm{Var}[(\tilde r_2^{(e)})^2]}{\tilde r_2^2} &\approx& 
        \sum_{n} \frac{\mathrm{Var}[p_n^{(e)}]}{{\mathbb E}[p_n]^2}
        +2\sum_{n<n'} 
        \frac{a_na_{n'}\mathrm{Cov}[p_n^{(e)},p_{n'}^{(e)}]}{{\mathbb E}[p_n]{\mathbb E}[p_{n'}]}
        \nonumber
        \\
        &=& 
        \sum_{n} \frac{\mathrm{Var}[p_n]}{K{\mathbb E}[p_n]^2}
        \nonumber \\
        &+&2\sum_{n<n'} 
        \frac{a_na_{n'}\mathrm{Cov}[p_n,p_{n'}]}{K{\mathbb E}[p_n]{\mathbb E}[p_{n'}]}.
        \label{eq:LVr2}
\end{eqnarray}
Note that this type of approximation to the variance is known as linearized variance in statistics~\cite{Lee2006stat}. To write the second equality, we have used the fact that the estimations $p_n^{(e)}$ are built from $K$ independently sampled random states. 
Using App. \ref{app:weingarten}, we can  write the analytical expressions for the variances ${\rm Var}[p_n]$, as well as covariances ${\rm Cov}[p_n,p_m]$ (see also Ref.~\cite{leone2021quantum} for $n=2$). This allows us to approximate $\mathrm{Var}[\tilde r_2^{(e)}]$ (given in Eq.~\eqref{eq:LVr2})  analytically. 

In Fig.~\ref{fig:phase_diagram_r_2}, apart from showing the variance of $p_2^{(e)}$, we also plot that of $\tilde r_2^{(e)}$ as a function of $N$, and for different points of the phase diagram. Here, we consider the extreme situation $K=1$ where those quantities are estimated from a single random state using $p_n^{(e)}=p_n$, and $
\tilde r_2^{(e)}=r_2$. For small systems $N$, we also calculate numerically the exact variance of $r_2$ (without Taylor approximation),  and  we obtain an excellent agreement with our linearized variance approximation. 
We observe that the relative variance ${\rm Var}[r_2]/\tilde r_2^2$  decays exponentially with $N$. This means that in the thermodynamic limit $N\to \infty$, the statistical fluctuations of $r_2$ become negligible even in the extreme case where a single state is considered.

\section{Measuring $\tilde r_2$ with classical shadows of randomized measurements}\label{app:measuring}
In this section we summarize the protocol to access $\tilde r_2$ in an experiment based on the randomized measurement toolbox~\cite{elben2022therandomized} (and references therein), and estimate numerically the statistical errors that are due to the finite number of measurements.

We consider a randomized measurement protocol of PT moments for a systems of $N$ qubits based on `classical shadows'~\cite{huang2020shadows,elben2020mixed}. This relies on single qubit random unitaries $u=\bigotimes u_i$ that are sampled independently from the Haar measure. When applying $N_u$ such transformation on a quantum state $\rho_{AB}$, and performing for each transformation $N_m$ projective measurements, one can build unbiased estimations of the PT moments $p_n$ (see Ref.~\cite{elben2020mixed} for the precise estimation formula and variance bounds on the statistical estimators).
In the situation where we are interested in average PT moments over Haar random states, in order to access $\tilde r_2$,  we will consider that the measurement sequence is performed  on $N_s$ different random states sampled from the Haar measure. The total number of projective measurements to obtain an estimation of $\tilde r_2$ is therefore $N_sN_uN_m$.

We now assess numerically the required number of measurements to extract $\tilde r_2$ with a small statistical error. We consider a system with $N=8$ qubits, which is sufficiently large to observe the three entangled phase of Haar random states. In Fig.~\ref{fig:staterr}, we show the estimated value of $\tilde{r}_2$ as a function of $N_uN_m$ for different $N_m=1,10,100$. The number of states is fixed to $N_s=64$, which is sufficient to obtain convergence of the average PT moments to the Haar expectation values with excellent accuracy.
We consider different partitions $N_A,N_B,N_C$ corresponding to the PPT phase [panels a),b)], the maximally entangled phase [panels c),d)], and the entanglement saturation phase [panels e),f)].
The errors are computed using the jackknife resampling method.
As shown in the figure, we can estimate the value of $\tilde r_2$ with good accuracy, and thus identify the three entanglement phases with a number of measurements $64N_uN_m\sim 10^{5}-10^6$ that is compatible with current experimental possibilities~\cite{brydges2019, satzinger2021toric}.

\begin{figure}
    \centering
    \includegraphics[width=0.99\linewidth]{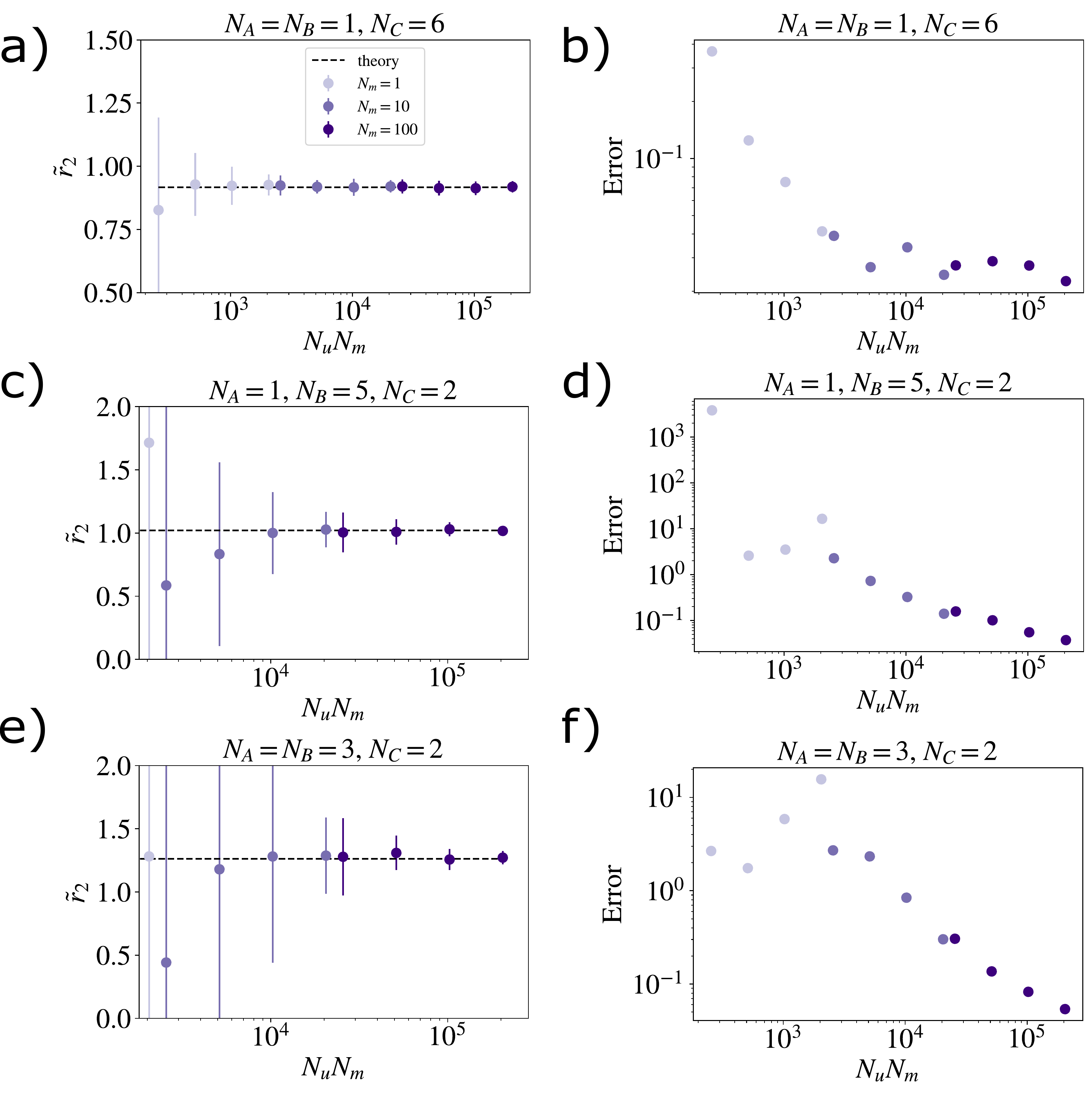}
    \caption{
    Estimates and their jackknife errors of $\tilde{r}_2$ for an ensemble of 64 $8$-qubit Haar random states via randomized measurements. The chosen partition sizes correspond to a), b) the PPT phase, c), d) the ME phase, and e), f) the ES phase.}
    \label{fig:staterr}
\end{figure}

\section{$p_3-$negativity for random states}\label{app:p3neg}

In this section we study the value of the $p_3$-negativity $\tilde{\mathcal{E}}_3 = \log_2(\mathbb{E}[p_2]^2/\mathbb{E}[p_3])/2$ of Haar random states.
In Fig.~\ref{fig:E3haar} a)-b), we represent $\tilde{\mathcal{E}}_3$ for $N_{AB}=256$, and $N_{AB}=10$ respectively.
In panels c) and d), we also compare $\tilde{\mathcal{E}}_3$ with the average value of $\mathcal{E}_3(\rho)$, and of the negativity $\mathcal{E}_3(\rho)$ for two `cuts' of the phase diagrams (see caption). 
For $N_{AB}=256$, the $p_3$-negativity reproduces without noticeable differences the phase diagram of the negativity shown in Ref.~\cite{shapourian2021entanglement}. In particular, we observe for $\tilde{\mathcal{E}}_3$ a linear dependence with $\min(N_A,N_B)$ in the maximally entangled phase, and a saturation in the entanglement saturation phase. These features are also visible at small system size $N_{AB}=10$. 

In Figs.~\ref{fig:E3haar} c)-d), we observe that $\tilde{\mathcal{E}}_3$ and $\mathbb{E}[\mathcal{E}_3(\rho)]$ seem to be always smaller than the averaged negativity. 
Interestingly, one can make a similar observation for single random states (i.e., not taking ensemble averages).
In Fig.~\ref{fig:E3haar}e), we represent the negativity $\mathcal{E}(\rho)$ as a function of $\mathcal{E}_3(\rho)$ for various Haar random states, which were obtained by taking different values of $2\le N_{AB}\le 10$, and $N_C\le N_{AB}+4$.
For each set of $(N_A,N_B,N_C)$, $50$ states were sampled. 
For all the random states that we numerically sampled, we observe that the $p_3$-negativity is always smaller than the negativity $\mathcal{E}(\rho)$. 
\begin{figure}[t]
    \centering
    \includegraphics[width=0.99\linewidth]{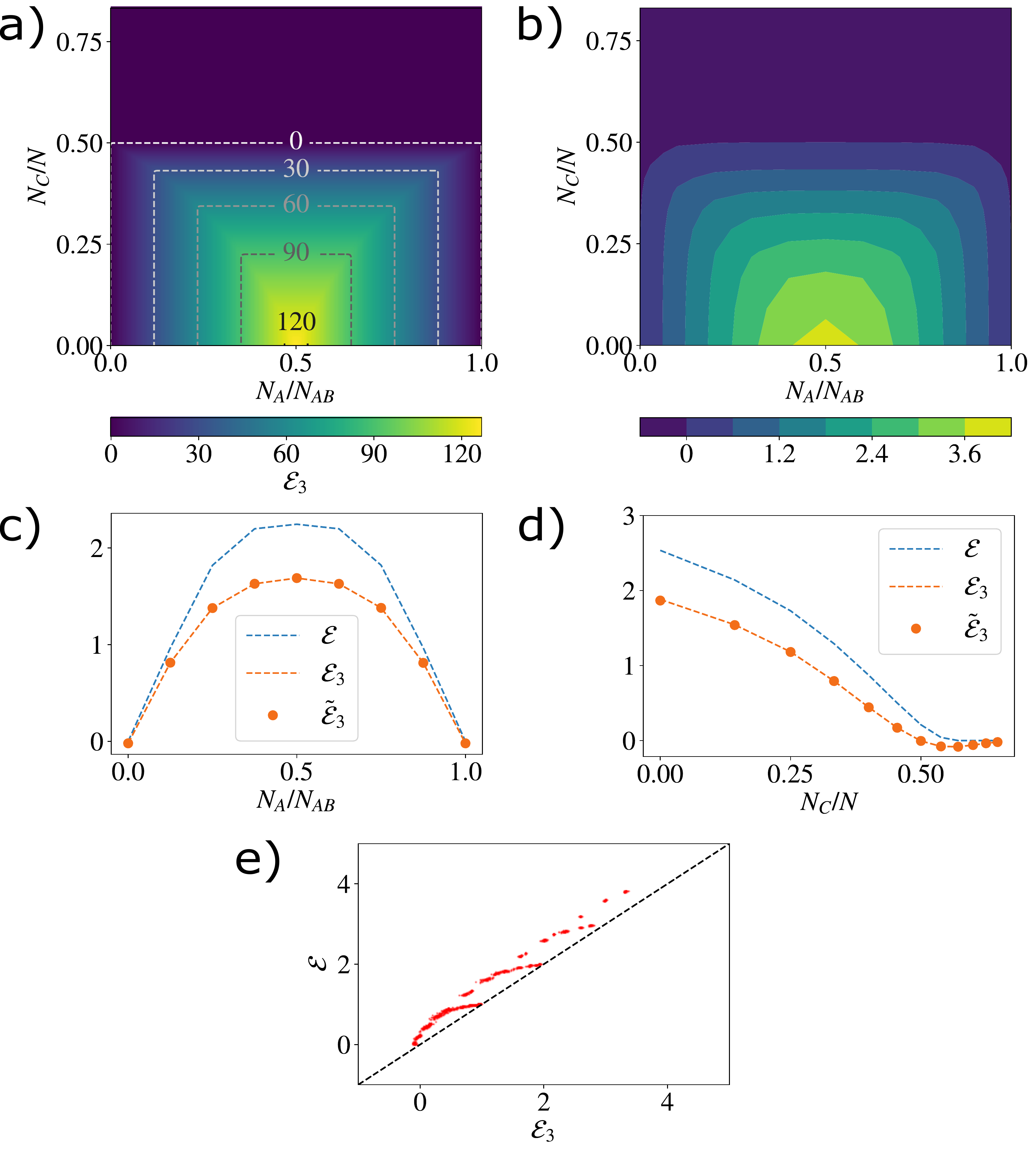}
    \caption{a), b) Phase diagram of Haar random states as probed by $\tilde{\mathcal{E}}_3$  for a) $N_{AB} = 256$ a), and b) $N_{AB}=10$. c), d) Comparison between the logarithmic negativity $\mathcal{E}(\rho)$ and the $p_3$-negativity $\mathcal{E}_3(\rho)$ for few-qubits Haar random states, respectively $N_{AB}=8$, $N_C=3$ c) and $N_{AB}=6$, $N_A=3$ d), averaged over $\sim 10^2$ samples. $\mathcal{E}_3(\rho)$ has been averaged both among individual states (dashed line) and by computing the logarithm of the average of the PT moments for the ensemble of states (points, see definition of $\tilde{\mathcal{E}_3}$ in the text).
    e)
    Comparison between the logarithmic negativity $\mathcal{E}(\rho)$ and the $p_3-$negativity $\mathcal{E}_3(\rho)\le \mathcal{E}(\rho)$ for single Haar random states. Red points are the values obtained numerically, while the black dashed line represents the function $\mathcal{E}=\mathcal{E}_3$.}
    \label{fig:E3haar}
\end{figure}

\section{Gaussian fermionic states}\label{app:freefermions}
Here, we give more details on the calculation of $\tilde r_2$ for Gaussian fermionic states studied in Sec.~\ref{sec:fermionic}. In particular, we explicitly show how the PT moments can be determined form the correlation matrix for fermionic Gaussian states  (see, e.g., Eq.~\ref{eq:FFPTmoments}).\\

As discussed in Sec.~\ref{sec:fermionic}, we consider here the case in which subsystems $A$, $B$ and $C$ are connected and $A$ and $B$ are adjacent. Our aim is to determine the PT moments of a fermionic Gaussian state $\rho_{AB}$ with a ($2N_{AB}\times 2N_{AB}$) correlation matrix $G'$. $G'$ is obtained from the correlation matrix of the whole system ($ABC)$, $G$, which is a  ($2N\times 2N$) matrix by deleting the rows and columns with indices that correspond to the modes in $C$ (Sec.~\ref{sec:fermionic}). For $\ket{\Psi}_{ABC}=U|0^N\rangle$, the correlation matrix $G$ is, as explained before, given by $G=R \cdot G_0 \cdot R^{T}$. Here, $R\in{\rm SO}(2N)$ is sampled uniformly random. 

To determine the PT moments of the fermionic Gaussian state $\rho_{AB}$ with correlation matrix $G'=\frac12\langle [c_i,c_j]\rangle_{\rho_{AB}}$  we proceed as follows. Let $W'$ be a ($2N_{AB}\times 2N_{AB}$) matrix such that
\begin{equation}\label{eq:GammaWrelation}
G'=\tanh{\frac{W'}{2}}\,.
\end{equation}
Then, we have ~\cite{Peschel_2003,Peschel_2009}
\begin{equation}
    \rho=\frac{1}{Z}\exp\left({\frac{1}{4}\sum_{kl}W'_{kl}c_kc_l}\right),
\end{equation}
where $Z$ is a normalization factor. As we will see below, the partial transpose $\rho^\Gamma$ of the state $\rho$ can be expressed in terms of a matrix $G^+$ that is constructed from the correlation matrix $G'$ as follows:
\begin{equation}\label{eq:gammap}
    G^+=
    \begin{pmatrix}
    G'_{AA} &  {\rm i}\,G'_{AB} \\[3mm]
    {\rm i}\,G'_{BA} & -G'_{BB}
    \end{pmatrix}.
\end{equation}

In the previous equation, the ($2N_X\times 2N_Y$) matrices $G'_{XY}$ for $X,Y\in\{A,B\}$ are submatrices of the ($2N_{AB}\times 2N_{AB}$) correlation matrix $G'$ obtained by taking rows (columns) with indices that correspond to modes in $X$ ($Y$). 
More precisely, the partial transpose of the state $\rho_{AB}$ can be expressed as~\cite{Eisler_2015}
\begin{equation}\label{eq:standard_partial_transpose}
    \rho^{\Gamma}=\frac{1-i}{2}O_{+}+\frac{1+i}{2}O_{-}\,.
\end{equation}
where the operators $O_{\pm}$ can be written as
\begin{equation}\label{eq:FFPTtransform}   O_{+}=O_{-}^{\dagger}=\frac{1}{Z}\exp\left({\frac{1}{4}\sum_{kl}W^{+}_{kl}c_kc_l }\right)
\end{equation}
with $W^+$ related to the previously defined matrix $G^+$ via the analogue of Eq.~\eqref{eq:GammaWrelation}. In other words,
\begin{equation}\label{eq:GammaWrelation2}
G^+=\tanh{\frac{W^+}{2}}\,.
\end{equation}

It is clear that the PT density matrix given by Eq.~\eqref{eq:standard_partial_transpose} is not a Gaussian operator but rather the sum of two of them and one can write
\begin{equation}
\tr(\rho^{\Gamma})^k=\frac1{2^{k/2}}
\sum_{(\sigma_1,\ldots,\sigma_k)}
\exp\left({-{\rm i}\frac{\pi}{4}\sum_{i=1}^k\sigma_i}\right)
\tr\left(\prod_{j=1}^{k}O_{\sigma_j}\right)\,,
\end{equation}
where the leftmost sum is over all possible tuples $(\sigma_1,\ldots,\sigma_k)\in\{+,-\}^k$. Then the PT moments of order $2$, $3$ and $4$ read \begin{equation}\label{eq:FFPTmoments}
    \begin{aligned}
    p_2=\tr \left( \rho^{\Gamma}\right)^2=&+\tr \left( O_{+}O_{-} \right),\\
    p_3=\tr \left( \rho^{\Gamma} \right)^3=&-\frac{1}{2}\tr\left( O_{+}^3\right)+\frac{3}{2}\tr\left( O_{+}^2O_{-}\right),\\
    p_4=\tr \left( \rho^{\Gamma} \right)^4=&-\frac{1}{2}\tr\left( O_{+}^4\right)+\tr\left( O_{+}^2O_{-}^2\right)+\\
    &+\frac{1}{2}\tr\left( O_{+}O_{-}O_{+}O_{-}\right),\\
    \end{aligned}
\end{equation}
and it is possible to compute them efficiently. 

\end{appendix}

\end{document}